\documentclass[12pt,letterpaper]{article}
\input{paper_style.sty} 
\usepackage{caption}
\usepackage{appendix}
\usepackage{xr}
\makeatletter

\begin{document}
\singlespacing
\title{Cointegrated Regressors: Where Folded Concave Penalties Go Beyond LASSO\thanks{Hongqi Chen: \protect\url{hongqichen@hnu.edu.cn}; Chengwang Liao: \protect\url{lchw@link.cuhk.edu.hk};
Ziwei Mei (Corresponding author): \protect\url{ziweimei@um.edu.mo},
Faculty of Business Administration, University of Macau, Taipa, Macao
SAR, China; Zhentao Shi: \protect\url{zhentao.shi@cuhk.edu.hk}. }}

\author{Hongqi Chen$^a$, Chengwang Liao$^b$, Ziwei Mei$^c$, Zhentao Shi$^b$ \vspace{0.6em}
\\
 $^a$Hunan University \\
 $^b$The Chinese University of Hong Kong \\
$^c$University of Macau}
 
\vspace{3em}
\date{\today}

\maketitle

\onehalfspacing
\begin{abstract}

LASSO is one of the most popular regularization methods for high-dimensional regression models. While cointegration is common in practical applications of predictive regressions with many regressors, the convex penalty of LASSO overly suppresses the coefficients of cointegrated regressors and therefore fails to deliver consistent parameter estimation. We employ folded concave penalties (FCP), including the well-known SCAD and MCP, to address this issue. Using FCP, we establish two asymptotic results: (i) a local oracle property for the penalized regression, and (ii) a global oracle property for regressions with a penalty-level constraint. In particular, the latter guarantees a global solution that consistently estimates the true coefficients and enjoys oracle prediction performance. Simulation studies and an application to predicting the euro area unemployment rate show that FCP improves variable selection and out-of-sample prediction relative to LASSO.

\end{abstract}
\medskip



\vfill

\newpage

\doublespacing

\defcitealias{mei2024lasso}{MS24}
\defcitealias{gao2026lasso}{GLMS26}

\section{Introduction}

Predictive regressions provide a simple and transparent framework by linking a target variable to past predictors.
In the era of big data, information sets available to forecasters have expanded dramatically. Macroeconomic databases, financial markets, and sectoral indicators now include hundreds of potentially useful time series. This abundance offers an opportunity to improve forecasts, but also creates a challenge of high dimensionality. Sparse estimation methods address this challenge by searching a large information set while retaining only a small collection of valuable predictors.

Economic and financial time series, however, exhibit diverse dynamic patterns. Prices, employment, and asset valuations are often highly persistent and may share common stochastic trends. Economic equilibrium restrictions and long-run comovement can then make nonstationary variables cointegrated---there exist linear combinations to transform nonstationary time series into stationary ones, which may carry predictive information. A sound statistical procedure must therefore preserve predictive signals in unknown cointegrating relationships. This paper shows that the choice of regularization penalty is decisive for that purpose.

LASSO has been an off-the-shelf option in handling high-dimensional sparse predictive models.
Recent discussions have expanded beyond the independently and identically distributed (i.i.d.)~or weakly dependent data to highly persistent time series regressors,  modeled as \textit{local unit roots} (LUR), namely an autoregression of order 1 (AR(1)) $x_t=\rho x_{t-1}+e_t$ with $\rho \approx 1$. \citet[MS24 hereafter]{mei2024lasso} and \citet[GLMS26 hereafter]{gao2026lasso}
derive the \emph{restricted eigenvalue} (RE) condition under high-dimensional unit roots ($\rho=1$) and LURs, respectively, showing that LASSO maintains consistency with highly persistent data.

Despite the advancements, however, cointegration has not been resolved for predictive regressions in high-dimensional settings, where  the number of regressors ($p$) can be larger than the sample size ($n$). High dimensionality forbids the researcher from pre-testing the dynamic features of regressors and cointegration relationship.
It poses a challenging question: \emph{
  How to consistently estimate a high-dimensional sparse predictive regression when the regressors contain stationary, nonstationary, and cointegrated variables?}
This paper provides a positive answer by leveraging the folded concave penalties (FCPs), whose most well-known examples are \citet{fan2001variable}'s \textit{smoothly clipped absolute deviation} (SCAD) and \citet{zhang2010nearly}'s \textit{minimax concave penalty} (MCP).

This paper makes three theoretical contributions,
each advancing one step at a time. First, it reveals the intrinsic difficulty of LASSO when facing cointegrated variables. LASSO imposes $\ell_1$-penalties on the coefficients individually, ignoring the underlying cointegrated system. This mismatch causes excessive shrinkage of the coefficients of cointegrated regressors, thereby leading to inconsistency. We explain LASSO's inconsistency in Sections \ref{subsec: inconsistency} and \ref{subsec:illustrative} with a bivariate toy example and illustrative simulations, with formal theoretical justifications for general cases in Proposition \ref{prop: inconsistent LASSO} of the Supplementary Materials. Such inconsistency goes beyond the model selection inconsistency due to the violation of the irrepresentable condition \citep{zhao2006model}. Instead, it invalidates the standard $\ell_1$ error bounds of LASSO, leading to deteriorated out-of-sample (OOS) prediction performance.

Second, we establish that FCPs enjoy the oracle property in the high-dimensional predictive regression with cointegration. Intuitively, with a properly chosen tuning parameter, the SCAD and MCP penalties are bounded and much milder than the unbounded $\ell_1$-penalty. This mild penalty avoids overly sparse solutions caused by excessive shrinkage, enabling oracle variable selection.
With standard regularity assumptions, Theorem \ref{thm: oracle}(a) shows that
the oracle least squares estimator is captured, asymptotically, by one of the local optimizers of the regression penalized by FCPs.

However, the local optimality is of no practical use. To find out the true coefficients from data to facilitate prediction, we must go beyond characterizing local optimality of the vanilla penalized problem. 
Under the i.i.d.~setting, \citet{shen2012likelihood} and \citet{chen2023sparse} have examined the minimizer of the criterion function subject to a bounded capped $\ell_1$ or $\ell_0$ penalty. Inspired by these works, we propose tailoring the regularization by minimizing the penalized loss function subject to an upper bound on the total FCP, which we refer to as  \textit{constrained FCP}. The bounded FCP favors sparse solutions and avoids overfitting. Our final and most important theoretical result, presented in
Theorem \ref{thm: oracle}(b), shows that the global minimizer of the constrained scheme asymptotically coincides with the oracle estimator.
For technical proofs, this constraint implies that the $\ell_1$ error of the estimator is controlled by a small subvector, thereby enabling the use of the RE condition to justify a well-defined global optimizer that consistently estimates the true parameter; see Remark \ref{rem: RE} for discussions.

\medskip

Our theoretical results are based on a deep investigation of the mechanisms in high-dimensional cointegrated time series, which motivates technical innovations to address the concerning issues. As the first step toward the oracle property, we pin down the convergence rate of the oracle least squares estimator with cointegrated regressors of growing dimensionality, which has been an unsolved problem in the literature. Due to nonstationarity, the inverse Gram matrix has an explosive maximum eigenvalue that complicates the error bounds, not to mention the unobserved rotation of stationary and persistent components in the underlying cointegrating system.
In addition, many technical subtleties, including the ``beta-min condition'', the sparsity requirement, and the deviation bound must be tailored for cointegration. Fortunately, with proper regularity assumptions, the error bounds that we deduce are comparable to those under pure unit roots.

Our theoretical predictions are supported by numerical studies. Monte Carlo simulations compare the performance of the representative FCPs, including the vanilla SCAD and MCP and their constrained versions, to LASSO. The results show that the FCPs substantially improve parameter estimation, variable selection, and OOS prediction. The constrained FCPs slightly improve the finite-sample performance of the vanilla FCPs. We apply these methods for OOS prediction of the unemployment rate of the euro area with high-dimensional cointegrated macroeconomic predictors, and the FCPs again outperform LASSO.

\textbf{Literature review.} This paper contributes to the literature of high-dimensional time series regressions. Estimation and inference with stationary time series are well established \citep{wong2020lasso,wang2022high,adamek2022lasso}, while machine learning theory for predictive regressions is a burgeoning direction.
In the $p/n\to 0$ moderate dimensional settings,
\citet{koo2020high}, \citet{smeekes2021automated}, \citet{lee2022lasso}, \citet{ fan2023predictive}, and \citet{phillips2024high} study the scenarios with  nonstationary regressors, and \citet{qin2026feature} apply canonical correlation analysis to select active predictors.
In the $p>n$ regime,
\citetalias{mei2024lasso} and \citetalias{gao2026lasso} are the first to provide consistency theory, and the latter further develops a debiasing inference procedure robust to nonstationarity.
Complementing these works, this paper deals with cointegrated variables in  high-dimensional predictive regressions.

Our paper stands on an extensive literature about FCPs. Since \citet{fan2001variable} and \citet{zhang2010nearly}, FCPs have been widely used for model selection and post-regularization inference \citep{zou2008one,shi2019linear,guo2025model}. Progress has been made in exploring global optimality of FCPs. When data are i.i.d., \citet{kim2008smoothly} show that the oracle estimator is the global minimizer of SCAD when $p/n \to c\in(0,1)$. When $p>n$, it is much more challenging to show the global optimality of a local solution \citep{fan2011nonconcave}. For i.i.d.~data, \citet{loh2015regularized} demonstrate that under \textit{restricted strong convexity} (RSC), all local minimizers
of SCAD or MCP are local to the true coefficients. Nevertheless, RSC fails under nonstationarity, since the minimum RE shrinks to zero, rather than being bounded below. Other papers explore the solution paths of specific algorithms such as local linear approximation (LLA, \citealp{fan2014strong}) and calibrated concave-convex procedure (CCCP, \citealp{wang2013calibrating}), while their asymptotic guarantees require an initial estimator like LASSO to be consistent, and therefore collapse under cointegration. Motivated by the \textit{restricted global optimality}  \citep{fan2011nonconcave}, we propose adding a restriction on the penalty to ensure a global optimizer that  achieves consistent estimation and improves finite-sample prediction. 

Another relevant strand of literature is on the eigen-analysis for large random matrices, which has been applied to unit root tests \citep{zhang2018clt}, cointegration tests  \citep{zhang2019identifying,bykhovskaya2022cointegration,bykhovskaya2024asymptotics}, and factor analysis \citep{onatski2021spurious,onatski2025spurious,qin2026feature}. Recent advancements on estimation and inference for nonstationary regression models include \citet{yang2020testing,yang2021unified}, \citet{masini2022counterfactual}, \citet{hecq2023granger,hecq2023inference}, \citet{li2025factor} and \citet{liao2026nickellmeetsstambaughtale}.
Random matrix theory is employed in dense regressions \citep{hastie2022surprises,he2024ridge,liao2026benignoverfittingeconomicforecasting,chang2026zero}, and its application under cointegration is yet to be explored.

\textbf{Organization}. Section \ref{sec: model} formulates the predictive regression model with cointegration and heuristically discusses the inconsistency of LASSO and the remedies by FCPs. Section \ref{sec: theory} establishes the asymptotic theory, with the FCP regularization for high-dimensional regressors involving cointegration. Section~\ref{sec: simulations} includes Monte Carlo simulations, and Section~\ref{sec: empirical} illustrates an empirical application of predicting the euro area unemployment rate.

\section[Setup and Estimation]{Setup and Estimation}\label{sec: model}
We lay out the notations prior to formal presentation. The sets of
natural numbers, integers, real numbers, and complex numbers are denoted
as $\mathbb{N}$, $\mathbb{Z}$, $\mathbb{R}$, and $\mathbb{C}$,
respectively. For a positive integer $N$, we use $[N]$ as shorthand for $\{1,2,\dots,N\}$. For an $n$-dimensional vector $\xx=(x_{t})_{t\in[n]}$, the $\ell_0$-norm $\|\xx\|_0=\sum_{t=1}^n\textbf{1}\{x_t\neq 0\}$ is the number of nonzero elements, the $\ell_1$-norm is $\left\Vert \xx\right\Vert _{1}=\sum_{t=1}^{n}\left|x_{t}\right|$, the $\ell_2$-norm is $\left\Vert \xx\right\Vert _{2}=\sqrt{\sum_{t=1}^{n}x_{t}^{2}}$, and the sup-norm is $\|\xx\|_{\infty}=\sup_{t\in[n]}|x_{t}|$. A subvector $\xx_{\mathcal{A}} = (x_t)_{t\in\mathcal{A}}$ collects the elements in the index set $\mathcal{A}\subset[n]$.
Let ${\bf 0}_{n}$ be an $n\times1$ zero vector, and $\bfI_{n}$ be the $n\times n$ identity matrix. For a generic matrix $\BB,$ let $B_{i,j}$ be the $(i,j)$th element, $\BB_{i,\cdot}$ be the $i$-th row, and $\BB^{\top}$ be its transpose. Let $\lambda_{\min}(\BB)$ and $\lambda_{\max}(\BB)$ be the minimum and maximum eigenvalues, respectively. For a set $\mathcal{A}$ of row indices, $\BB_{\mathcal{A},j}=(B_{i,j})_{i\in\mathcal{A}}$ is the subvector of the $j$-th column indexed by $\mathcal{A}$. An \emph{absolute constant} is a positive, finite constant invariant with the sample size. Define $a\wedge b:=\min\left\{ a,b\right\} $ for real numbers $a$ and $b$, and ${\rm sgn}(a)$ as the sign of $a$.  For positive sequences $a_n$ and $b_n$, ``$a_n\gg b_n$'' means $a_n/b_n\to\infty$, ``$a_{n}\stackrel{\mathrm{p}}{\preccurlyeq}b_{n}$'' means that there is an absolute constant  $c$ such that $a_{n}\leq cb_{n}$ with probability approaching one, and $a_n\stackrel{\mathrm{p}}{\asymp}b_n$ means that both $a_{n}/b_{n}$ and $b_{n}/a_{n}$ are bounded in probability. We use ``$\stackrel{\mathrm{p}}{\to}$'' to denote convergence in probability. For any time series $\{x_t\}_{t\in[n]}$, we write $\Delta x_t = x_t - x_{t-1}$ as the first-order difference.

\subsection{Predictive Regression} \label{subsec: predictive regression}
Let $\ww_t = (w_{j,t})_{j\in[p]}$ denote a $p$-vector of predictors. At time $n$, we are interested in using a linear combination $\ww_n^\top\ttheta$ to predict a future scalar outcome $y_{n+1}$. We learn the true parameter $\ttheta^*$ from the historical data via the predictive regression
\begin{equation}\label{eq: feasible}
    y_t = \ww_{t-1}^\top\ttheta^* + u_t,\quad t=1,2,...,n.
\end{equation}
To simplify notations, in the theoretical analysis we follow \citetalias{mei2024lasso} by omitting the intercept, which can be handled by the Frisch-Waugh-Lovell theorem. For simplicity, let
the initial value $\|\ww_{t=0}\|_\infty = O_p(1)$. Let $\mathcal{S}=\{j\in[p]:\theta_j^*\neq 0\}$ denote the index set of nonzero coefficients, and $\mathcal{S}^c = [p] \setminus \mathcal{S}$ be its complement. We call $w_{j,t}$ an \emph{active regressor} if $j\in\mathcal{S}$, and an \emph{inactive regressor} otherwise.

All regressors with potential predictability of the target variable should
be included in the regression specification in the first place.
The ``zoo'' of candidate predictors, therefore, may include a diverse range of ``species''. Some are stationary, others are highly persistent, and there can be cointegrated ones among the latter.
As representatives of major categories, we consider three types of regressors in $\ww_t = (\zz_t^\top,\xx_t^{{\rm ab}\top},\xx_t^\top)^\top$:
$\zz_t$ is a $p_z$-vector of stationary regressors, $\xx_t^{{\rm ab}}$ is a $p_{\rm ab}$-vector of cointegrated time series, and $\xx_t$ is a $p_x$-vector of LUR processes.
\citet{lee2022lasso} use the same setting in a low dimensional (fixed $p$) regression, whereas this paper allows high dimensionality with $p$ larger than $n$.

The LURs are modeled as
\begin{equation*}
    \xx_{t} = \left(\bfI_{p_x}+n^{-1} \CC\right) \xx_{t-1} + \ee_t
\end{equation*}
with a $p_x$-vector of
stationary innovations $\ee_t$ and a diagonal matrix $\CC = \text{diag}(c_1, \ldots, c_{p_x})$. The matrix $\bfI_{p_x} + n^{-1}\CC$ is close to the identity matrix in large samples, which characterizes high persistence and allows for local unit roots. The exact unit root in \citetalias{mei2024lasso} is our special case with $\CC$ being the zero matrix.
For the cointegrated regressors, we consider the system with cointegration rank $p_{\rm a}$, where $ \xx_t^{{\rm ab}}= (\xx_t^{\rm a\top},\xx_t^{\rm b\top})^\top$ and
\begin{equation}\label{eq: triangular} 
         \xx_t^{\rm a} = \AA \xx_t^{\rm b} + \vv_t^{\rm a}, \quad \quad
         \xx_t^{\rm b} = \left(\bfI_{p_{\rm b}}+ n^{-1} \CC^{\rm b} \right)\xx_{t-1}^{\rm b} + \ee_t^{\rm b}. 
\end{equation}
The $p_{\rm a}$-vector $\vv_t^{\rm a}$ is the stationary cointegrating residual and  $\ee_t^{\rm b}$ collects $p_{\rm b} = p_{\rm ab} - p_{\rm a}$ stationary components. The coefficients include a $p_{\rm a}\times p_{\rm b}$ matrix $\AA$ and a   diagonal matrix $\CC^{\rm b} = \text{diag}(c_1^{\rm b}, \ldots, c_{p_{\rm b}}^{\rm b})$. The cointegration in $\xx_t^{\rm ab}$ allows for nontrivial predictive power of nonstationary regressors for a stationary outcome.

The cointegrating system with the triangular form (\ref{eq: triangular}) accommodates general cointegration structures, admitting the following representations of the integrated latent factor and the error correction model, which are popular in the literature.

 \begin{example}[Integrated Latent Factor]\label{exp:factor} Following \citet{stock1988testing} and \citet{barigozzi2024inference}, suppose that the $p_{\rm ab}$ cointegrated regressors admit a latent factor model $\xx_t^{{\rm ab}} = \LLambda\ff_t + \uu^{\rm ab}_t$, where $\ff_t = \ff_{t-1} + \vareepsilon_t$ is a $p_{\rm b}$-vector of unit root factors and $(\uu_t^{\mathrm{ab}\top},\vareepsilon_t^\top)^\top$ are stationary innovations. Furthermore, assume that the $p_{\rm ab} \times p_{\rm b}$ loading matrix has full column rank $p_{\rm b}$, and therefore has a $p_{\rm b}\times p_{\rm b}$ full-rank submatrix. Without loss of generality, we can write $\LLambda = (\LLambda^{\rm a\top},\LLambda^{\rm b\top})^\top$ where $\LLambda^{\rm b}$ is a $p_{\rm b}\times p_{\rm b}$ full-rank matrix. Let $\xx_t^{\rm a}$ be the first $p_{\rm a}$ entries of $\xx_t^{{\rm ab}}$, and $\xx_t^{\rm b}$ be the last $p_{\rm b}$ ones. We then rewrite the factor model as
 \[
 \begin{pmatrix}
    \xx_t^{\rm a} \\
    \xx_t^{\rm b}
 \end{pmatrix}=\begin{pmatrix}
    \LLambda^{\rm a} \\
    \LLambda^{\rm b}
 \end{pmatrix}\ff_t +\begin{pmatrix}
    \uu_t^{\rm a} \\
    \uu_t^{\rm b}
 \end{pmatrix},\quad \ff_t = \ff_{t-1} + \vareepsilon_t.
\]
We can derive $\xx_t^{\rm b} = \xx_{t-1}^{\rm b} + \ee_t^{\rm b}$ with $\ee_t^{\rm b} = \Delta\uu_t^{\rm b} + \LLambda^{\rm b}\vareepsilon_t$, and $\xx_t^{\rm a} = \AA \xx_t^{\rm b} + \vv_t^{\rm a}$ with $\AA = \LLambda^{\rm a}(\LLambda^{\rm b})^{-1}$ and $\vv_t^{\rm a} = \uu_t^{\rm a} -  \AA \uu_t^{\rm b}$, which is a special case of (\ref{eq: triangular}).
 \end{example}

 \begin{example}[Error Correction Model] Following \citet{bykhovskaya2022cointegration}, suppose that the $p_{\rm ab}$ cointegrated regressors admit an error correction model $\Delta \xx_t^{\rm ab} = \aalpha^{{\rm ab}}\bbeta^{{\rm ab}\top} \xx_{t-1}^{\rm ab} + \vareepsilon_t$, where two $p_{\rm ab} \times p_{\rm a}$ matrices $\aalpha^{{\rm ab}}$ and $\bbeta^{{\rm ab}} $ are of full column rank, and $\bbeta^{{\rm ab}\top}\xx_t^{\rm ab}$ is stationary.  Let the $p_{\rm ab}\times p_{\rm b}$ matrices $\aalpha_\perp^{{\rm ab}}$ and $ \bbeta_\perp^{{\rm ab}}$ be the orthogonal complements of $\aalpha^{{\rm ab}}$ and $\bbeta^{{\rm ab}}$. By the Granger Representation Theorem \citep{engle1987co,hansen2005granger}, we can rewrite the error correction model as
 $\xx_t^{\rm ab} = \LLambda \ff_t + \ee^{\rm ab}_t$, where $\ff_t = \aalpha_\perp^{{\rm ab}\top}\sum_{s=1}^t \vareepsilon_s$ is a $p_{\rm b}$-vector of unit roots, $\LLambda    = \bbeta_\perp^{{\rm ab}}(\aalpha_\perp^{{\rm ab}\top}\bbeta_\perp^{{\rm ab}})^{-1}$, and the stationary $\ee^{\rm ab}_t$ is a linear process driven by $\vareepsilon_t$. Therefore, we can  formulate the error correction model as an integrated factor model, and derive the triangular representation \eqref{eq: triangular} following Example~\ref{exp:factor}.
 \end{example}

 \medskip

\subsection{Regularized Estimation} \label{subsec: regularization}

Sparsity means the proportion of nonzero coefficients in $\ttheta^*$ is small.
Regularized estimation is a natural choice to estimate the sparse coefficient $\ttheta^*$ in the predictive regression (\ref{eq: feasible}). Importantly, we remain agnostic about the underlying cointegration relationship and the identities of the stationary regressors, the cointegrated regressors, and the local unit roots.

We consider the following penalized regression:
\begin{equation}\label{eq: regularized estimation}
    \hat\ttheta = \arg\min_{\ttheta} \mathcal{L}_\lambda(\ttheta),\quad \mbox{where }
    \mathcal{L}_\lambda(\ttheta)=\dfrac{1}{n}\sum_{t=1}^n(y_t - \ww_{t-1}^\top \ttheta)^2 + \sum_{j=1}^p J_\lambda(\hat\sigma_j\theta_j),
\end{equation}
where $\hat\sigma_j = \sqrt{n^{-1}\sum_{t=1}^n (w_{j,t-1}-\bar w_j)^2}$ is the sample standard deviation of $w_{j,t-1}$, and $\bar w_j = n^{-1}\sum_{t=1}^n w_{j,t-1}$ is the sample mean. The coefficients are weighted by the sample standard deviation (s.d.)~for two purposes. First, the standardization achieves scale invariance that does not change the interpretation in applications. 
Second, as highlighted in \citetalias{mei2024lasso} and \citetalias{gao2026lasso}, standardization is essential in balancing the different asymptotic scales of local unit roots and stationary regressors, without which consistency cannot be achieved. 

Our discussion will cover these penalty forms for sparse regressions:
\begin{description}
  \item[LASSO]
  $J_\lambda(t) = \lambda |t|;$

  \item[SCAD] $J_\lambda(t) = \lambda |t|\cdot \textbf{1}\{|t|\leq \lambda\} -\frac{t^2 - 2a\lambda |t| + \lambda^2}{2(a-1)} \cdot \textbf{1}\{\lambda<|t|\leq a\lambda\} + \frac{(a+1)\lambda^2}{2}\cdot  \textbf{1}\{|t| > a\lambda\}$, $a > 2$;

  \item[MCP]
  $
  J_\lambda(t) = \left( \lambda |t| - \frac{t^2}{2a}\right)\cdot\textbf{1}\{|t| \leq a\lambda\} + \frac{a\lambda^2}{2}\cdot \textbf{1}\{|t| > a\lambda\}$,   $a > 1$.
\end{description} 
As we will discuss in the following sections, LASSO fails to achieve consistent estimation under cointegration, and FCPs like SCAD and MCP are valid remedies.

To provide sharper theoretical results to facilitate prediction in practice, we also consider the following constrained regularization for FCPs:
\begin{equation}\label{eq: constrained regularized estimation}
   \tilde\ttheta = \arg\min_{\ttheta} \mathcal{L}_\lambda (\ttheta),\,
    \quad \text{subject to } \sum_{j=1}^p J_\lambda(\hat\sigma_j\theta_j) \leq \mu,
\end{equation}
with $\mu > 0$ being another tuning parameter that measures the level of overall sparsity. 
We will discuss the tuning parameter selection in Section \ref{sec: simulations} regarding the Monte Carlo simulations. Intuitively, the constraint on the total penalty level restricts the parameter space to a bounded-penalty region that favors sparser solutions, and avoids overfitting caused by many small coefficients. The constrained regularization problem has been studied by \citet{shen2012likelihood} and \citet{chen2023sparse}, while this paper is the first to consider this problem under cointegration.

\begin{remark}
For the LASSO penalty, the duality of convex optimization implies that the constrained formulation \eqref{eq: constrained regularized estimation} is equivalent to the unconstrained analog \eqref{eq: regularized estimation} for an appropriately chosen tuning parameter. Such an equivalence breaks down for nonconvex penalties such as SCAD and MCP. Interested readers can refer to supplementary Section \ref{subsec: inequivalence} for detailed discussions.
\end{remark}

\subsection{Inconsistency of LASSO and Its Remedy via FCPs}\label{subsec: inconsistency}

This section illustrates the inconsistency of LASSO under cointegration, and why FCPs are valid remedies. While LASSO is a representative example of convex penalties, with appropriate adaptations parallel arguments apply to other convex penalties such as the ridge regression and the elastic net.

The key difficulty in penalizing the cointegrated regression lies in the fact that, individually, each nonstationary regressor appears to be of order $O_p(\sqrt{n})$, but collectively they may contribute a stationary component of order $O_p(1)$. Since the standardized LASSO penalizes coefficients on an individual basis, it rescales the penalty according to the stochastic order of each regressor. The $O_p(\sqrt{n})$ penalty is disproportionately large relative to the $O_p(1)$ signal that should be retained.
This is the mismatch that causes the failure of LASSO.

To elaborate this issue, we reiterate a bivariate toy model used in \citetalias{mei2024lasso}, and then extend the arguments to general predictive regressions.
Throughout this section, we consider the LASSO tuning parameter satisfying $\lambda \to 0$ and $\sqrt{n}\lambda \to \infty$, which is necessary for balancing the LURs and the stationary ones.

\begin{example}\label{ex: bivariate toy}
Suppose that we observe the cointegrated pair $x^{\rm a}_t$ and $x^{\rm b}_t$  satisfying $x^{\rm a}_t = x^{\rm b}_t + v_t^{\rm a}$ and $x^{\rm b}_t = x^{\rm b}_{t-1} + e_t^{\rm b}$. Furthermore, a stationary outcome $y_t$ follows the predictive regression on the underlying cointegrating residual $v_t^{\rm a}$ with coefficient  $\phi^{\rm a*}\neq 0$:
\begin{equation*} 
    y_t  = v_{t-1}^{\rm a}\phi^{\rm a*} + u_t
         = x^{\rm a}_{t-1}\phi^{\rm a*} + x^{\rm b}_{t-1}\phi^{\rm b*} + u_t.
\end{equation*}
The rightmost expression, with $\phi^{\rm b*} = -\phi^{\rm a*}$,  is the feasible version.
Consider the (standardized) LASSO estimator $\hat\phi = \{\hat\phi^{\rm a},\hat\phi^{\rm b}\}$ that minimizes 
\begin{equation*}
\mathcal{L}_\lambda^{\rm LASSO}(\phi^{\rm a}, \phi^{\rm b})  =
\frac{1}{n}\sum_{t=1}^n (y_t - x^{\rm a}_{t-1}\phi^{\rm a} - x^{\rm b}_{t-1}\phi^{\rm b})^2 + \lambda (\hat\sigma_{\rm a}|\phi^{\rm a}| + \hat\sigma_{\rm b}|\phi^{\rm b}|),
\end{equation*}
with $\hat\sigma_j$ being the sample s.d.~of $\{x^{j}_t\}$ for $j \in\{{\rm a}, {\rm b}\}$, whose orders are $\sqrt{n}$ due to nonstationarity.
If $\hat\phi^{\rm a}$ and $\hat\phi^{\rm b}$ converge to any nonzero true coefficients, say $c^{\rm a}$ and $c^{\rm b}$, the LASSO criterion function
$
\mathcal{L}_\lambda^{\rm LASSO}( \hat \phi^{\rm a}, \hat \phi^{\rm b}) \geq
\lambda (\hat\sigma_{\rm a}|\hat\phi^{\rm a}| + \hat\sigma_{\rm b}|\hat\phi^{\rm b}|) \stackrel{\mathrm{p}}{\asymp} \sqrt{n}\lambda ( | c^{\rm a}| +  |c^{\rm b}|) \to \infty $ given $\sqrt{n} \lambda \to \infty$.
This is much larger than
$\mathcal{L}_\lambda^{\rm LASSO}( \hat \phi^{\rm a} = 0, \hat \phi^{\rm b} = 0 ) =  n^{-1}\sum_{t=1}^n y_t^2  + 0 \stackrel{p}{\rightarrow} \mathbb{E}(y_t^2) < \infty$, where the last step applies the fact that $y_t$ is stationary. This contradiction implies that LASSO is inconsistent.
\end{example}

LASSO's inconsistency arises in a wider class of cointegrated models beyond the simple bivariate example, which is formulated in Section \ref{subsec:InconsistencyLASSO} of the Supplementary Materials.
 We highlight that LASSO's inconsistency under cointegration is more severe than overly selecting redundant variables in the i.i.d.~setting. Instead, this inconsistency leads to a failure in guaranteeing vanishing $\ell_1$ estimation error, which fundamentally undermines parameter estimation and impairs the predictive performance.

How do FCPs fix the inconsistency of LASSO? Recall that in the toy example, the $\ell_1$ penalties explode whenever the LASSO estimates for the coefficients of the cointegrated regressors converge to nonzero values. In contrast, FCPs become flat when the estimates are sufficiently large, which prevents the penalty from growing excessively and allows the estimates to converge to the true nonzero coefficients without being overly penalized. At the same time, FCPs behave like LASSO around 0, thereby remaining effective in variable selection.

Figure \ref{fig:illustration} provides an illustration of how LASSO and SCAD penalties behave for cointegrated regressors. The leftmost panel shows the sum of squared residuals (SSR) function and penalties for a stationary regressor, and the middle panel presents the analogs for an LUR regressor. For stationary regressors, both the SSR and the penalty operate on a relatively small scale. For LUR regressors, the SSR becomes much steeper, which is matched by a larger penalty. The rightmost panel exhibits a misaligned pattern under a cointegrated regressor.
Specifically, the SSR there resembles that of a stationary component, whereas the penalty level follows that of the LUR according to the behavior of an individual time series. This mismatch causes LASSO to penalize the cointegrating coefficients excessively. SCAD remedies the over-penalization of LASSO because its penalty becomes flat for sufficiently large coefficients. Furthermore, when the coefficient is near zero, SCAD behaves like LASSO and therefore retains its variable selection.

\begin{figure}[htbp]
\centering
\includegraphics[width=1\textwidth]{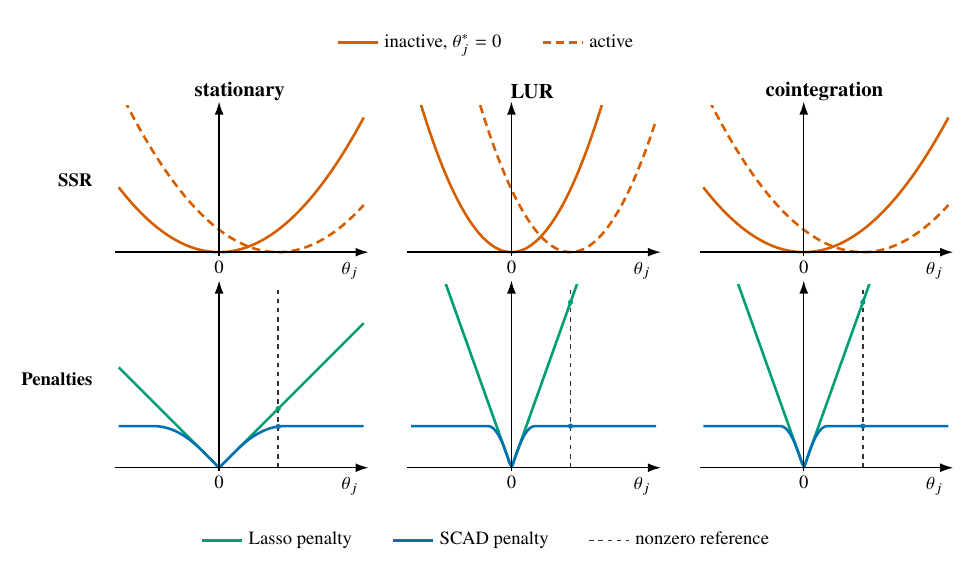}
\caption{Illustration of LASSO and SCAD penalties for cointegration}
{{\begin{flushleft}
\singlespacing\footnotesize \textit{Notes}: The upper panel plots the SSR as a function of the coefficient $\theta_j$ when the regressor is stationary, LUR, and cointegrated, respectively. The solid line represents the SSR function for an inactive regressor, while the dashed line represents that for an active regressor. The lower panel shows the corresponding penalty functions for LASSO and SCAD, respectively.
\end{flushleft}}}
\label{fig:illustration}
\end{figure}

\begin{remark}\label{rem: Plasso}
The inconsistency of LASSO cannot be resolved by removing standardization. \citetalias{mei2024lasso} show that the Plain LASSO without standardization requires $\lambda\to\infty$, rather than $\sqrt{n}\lambda\to\infty$, to handle LUR regressors. By the same argument, the estimated coefficients of the cointegrated regressors in the toy example still converge to zero. Moreover, such a large penalty also over-penalizes stationary regressors and shrinks their coefficients to zero. This is why \citetalias{mei2024lasso} emphasize that standardization is essential when stationary and persistent regressors coexist.
\end{remark}

\subsection{Illustrative Simulation}\label{subsec:illustrative}
Before diving into asymptotic theory, we use an illustrative simulation to highlight the necessity of FCPs under cointegration.
Under a predictive regression model specified in Section \ref{sec: simulations} with $(n,p)=(300,300)$, we examine the one-period OOS prediction errors and coefficient estimation errors by LASSO, SCAD, MCP, and the constrained counterparts of the latter two defined in (\ref{eq: constrained regularized estimation}), which we call the ``SCAD-C'' and ``MCP-C''.

\begin{figure}[htbp]
\centering
\includegraphics[width=0.9\textwidth]{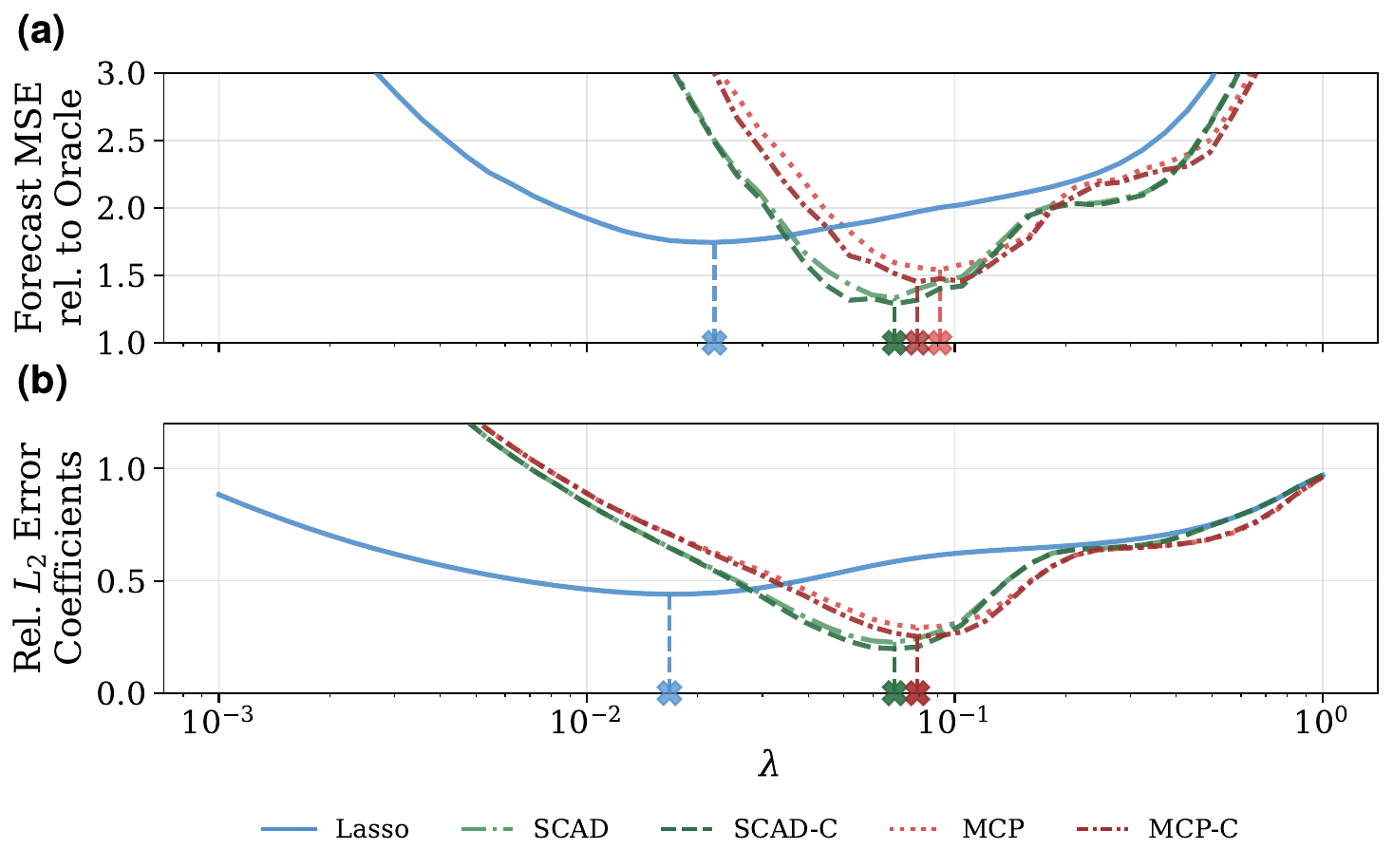}
\caption{Forecast performance and coefficient estimation error across $\lambda$}
\label{fig:solution-path}
{\begin{flushleft}
\singlespacing\footnotesize \textit{Notes}: Panel (a) exhibits the average OOS forecast MSE relative to that of the oracle OLS estimator over 1000 replications, and panel (b) exhibits the average relative $\ell_2$ error $\|\hat\ttheta-\ttheta^*\|_2/\|\ttheta^*\|_2$. The horizontal axis is the tuning parameter $\lambda$ from $10^{-3}$ to 1. The vertical dashed lines indicate the best tuning parameter $\lambda$ that minimizes the average forecast or estimation error. 
\end{flushleft}}
\end{figure}

Figure \ref{fig:solution-path}(a) exhibits the OOS forecast mean squared error (MSE) relative to that of the oracle OLS estimator over 1000 replications, where the oracle directly uses the infeasible knowledge of the active regressors. The graph shows the result over all penalty parameters $\lambda$ in (\ref{eq: regularized estimation}) over a grid system from $10^{-3}$ to 1. The result shows that the best average forecast MSE relative to the oracle OLS along the SCAD and MCP paths is close to 1, while the best average forecast MSE along the LASSO path is always above 1.7. Figure \ref{fig:solution-path}(b) exhibits the average relative $\ell_2$ error $\|\hat\ttheta-\ttheta^*\|_2/\|\ttheta^*\|_2$ for a generic estimator $\hat\ttheta$. The minimum errors along the SCAD and MCP paths are close to zero, in contrast to the minimum along the LASSO path that is around 0.5.
This illustration provides evidence that LASSO is missing useful components when cointegration is present, and it cannot be fixed by adjusting the tuning parameter selection.
SCAD and MCP clearly outperform LASSO, and the constrained regularization further improves the finite-sample performance slightly. 

\section{Theory}\label{sec: theory}

This section establishes asymptotic theory for FCPs under cointegration. Section \ref{subsec: conditions} states the conditions. Section \ref{subsec: OLS} characterizes the convergence rate of the oracle OLS estimator as a benchmark for subsequent analysis.
Section \ref{subsec: global} shows that the global minimizer of the constrained FCP (\ref{eq: constrained regularized estimation}) achieves consistency in both parameter estimation and OOS prediction. Section \ref{subsec: oracle} establishes the oracle property of FCPs under an additional beta-min condition.

\subsection{Conditions}\label{subsec: conditions}

The LASSO error bounds have been deduced by \citetalias{mei2024lasso} when stationary and unit root regressors are mixed in the predictive model without cointegration. As this paper is built on this theoretical foundation, understandably the conditions largely overlap. 

Suppose that all stationary components $\vv_t = (u_t,\zz_t^\top,\vv_t^{\rm a\top},\ee_t^{\rm b\top},\ee_t^\top)^\top$ are generated from a linear transformation of independent shocks, given as
\begin{equation}\label{eq: v=Phiepsilon}
    \vv_t = \PPhi \vareepsilon_t,
\end{equation}
where $\ensuremath{\vareepsilon_{t}=(\varepsilon_{j,t})_{j\in[p+1]}}$
is a $(p+1)\times1$ random vector and $\PPhi$ is a $(p+1)\times(p+1)$
deterministic matrix. For each $j\in[p+1]$, the shock
\begin{equation}
\varepsilon_{j,t}=\sum_{d=0}^{\infty}\psi_{j,d}\eta_{j,t-d}\label{eq:linrProc}
\end{equation}
is a linear process. Assumptions \ref{assu:tail}--\ref{assu:covMat} characterize the stationary innovations.

\begin{assumption}
\label{assu:tail} Suppose $(\eta_{j,t})$ is i.i.d.~over the cross
section $j\in[p+1]$ and time $t\in\mathbb{Z}$ with $\mathbb{E}\eta_{j,t}=0$
and $\mathbb{E}\eta_{j,t}^{2}=1$. There exists an absolute constant
$C_{\mathrm{f}}$ such that
\begin{equation}
\int_{-\infty}^{\infty}|f(x+a)-f(x)|dx\leq C_{\mathrm{f}}|a|\quad\text{for all } a\in\mathbb{R},\label{eq:density}
\end{equation}
where $f$ is the density function of $\eta_{j,t}$. For all $t\in\mathbb{Z}$
and $\mu>0$, there exist absolute constants $C_{\eta}$ and $b_{\eta}$
such that
\begin{equation}
\Pr\left\{ |\eta_{j,t}|>\mu\right\} \leq C_{\eta}\exp(-\mu/b_{\eta}).\label{eq:innovtail}
\end{equation}

\end{assumption}

\begin{assumption}
\label{assu:alpha} For all $j\in[p+1]$, there exist absolute
constants $C_{\psi}$, $c_{\psi}$, and $r$ such that the coefficients
of the linear processes satisfy
\begin{equation*}
|\psi_{j,d}|\leq C_{\psi}\exp\left(-c_{\psi}d^{r}\right)\quad\text{for all } d\in\mathbb{N}.
\end{equation*}
Furthermore, $\psi_{j}(z)=\sum_{d=0}^{\infty}\text{\ensuremath{\psi_{jd}\cdot z^{d}}}$,
where we normalize $\psi_{j0}=1$ without loss of generality. Assume $\left|\psi_{j}(z)\right|>c_{\psi}>0$ for any $z\in\left\{ a\in\mathbb{C}:|a|\leq1\right\} $.
\end{assumption}

Assumption \ref{assu:tail} presumes that the innovations $\eta_{j,t}$ are i.i.d.\ with a common density function $f(\cdot)$ and a sub-exponential tail. Assumption \ref{assu:alpha} characterizes the dependence structure of the stationary components, which together with (\ref{eq:density}) implies the $\alpha$-mixing property of the linear processes (\ref{eq:linrProc}). For the next assumption, define $\OOmega = \PPhi\PPhi^\top$.

\begin{assumption}
\label{assu:covMat}   There are absolute constants $c_{\OOmega}$, $C_{\OOmega}$
and $C_{L}$ such that: (a) $c_{\OOmega}\leq\lambda_{\min}(\OOmega)\leq\lambda_{\max}(\OOmega)\leq C_{\OOmega}$;
(b) $\max_{j\in[p+1]}\|\PPhi_{j,\cdot}\|_1\leq C_{L}$.
\end{assumption}

Assumption \ref{assu:covMat}(a) controls the magnitude of cross-sectional correlation. 
Assumption \ref{assu:covMat}(b), together with
(\ref{eq:innovtail}), guarantees the sub-exponential tail of $\vv_t$ in (\ref{eq: v=Phiepsilon}), and validates the use of concentration inequalities in the proofs.

To facilitate theoretical analysis, we rewrite (\ref{eq: feasible}) into two equivalent forms. Define $\ttheta^* = (\ggamma^{*\top}, \pphi^{*\top}, \bbeta^{*\top})^\top$ and $\pphi^{*} = (\pphi^{\rm a*\top}, \pphi^{\rm b*\top})^\top$. We have
\begin{align}
  y_t &=  \ww_{t-1} ^\top \ttheta^* + u_{t} \nonumber\\
  & =  \zz_{t-1}^\top\ggamma^* + \xx_{t-1}^{{\rm ab}\top}\pphi^{*}+ \xx_{t-1}^\top\bbeta^* + u_t
  \label{eq:equivalent1}\\
  &= \zz_{t-1}^\top\ggamma^* + \vv_{t-1}^{\rm a\top}\pphi^{\rm a*} + \xx_{t-1}^{\rm b\top}(\AA^\top\pphi^{\rm a*}+\pphi^{\rm b*}) + \xx_{t-1}^\top\bbeta^* + u_t
  \label{eq:equivalent2}\\
 &= \ww_{t-1}^{\sharp\top} \boldsymbol{\vartheta}^{\sharp*} + u_{t}. \nonumber
\end{align}
The first form (\ref{eq:equivalent1}) is the direct expansion of (\ref{eq: feasible}), and the second form (\ref{eq:equivalent2}) is more convenient for theoretical analysis, since it only involves stationary components and the LURs without cointegrating relationships. The vector $\ww_{t}^\sharp = (\zz_{t}^\top, \vv_{t}^{\rm a\top}, \xx_{t}^{\rm b\top}, \xx_{t}^\top)^\top$ collects the \emph{infeasible} regressors with known  cointegration relationship to transform the cointegrated regressors into stationary components, and $\boldsymbol{\vartheta}^{\sharp*} = (\ggamma^{*\top}, \pphi^{\rm a*\top}, (\AA^\top\pphi^{\rm a*}+\pphi^{\rm b*})^\top, \bbeta^{*\top})^\top$ is the corresponding infeasible parameter,
by substituting $\xx_{t-1}^{\rm a}$ with $\AA \xx_{t-1}^{\rm b} + \vv_{t-1}^{\rm a}$ as in (\ref{eq: triangular}). Note that the feasible $\ww_t$ is a rotation of the infeasible $\ww_{t}^\sharp$, given as
\begin{equation}\label{eq: rotation}
    \ww_t = \QQ \ww_{t}^\sharp,\text{ where }\QQ = \left(\begin{array}{cccc}
       \bfI_{p_z}  &  &  & \\
         & \bfI_{p_{\rm a}} & \AA  &  \\
         &  & \bfI_{p_{\rm b}} & \\
         &  &  & \bfI_{p_x}
    \end{array}\right).
\end{equation}
Correspondingly, the coefficients satisfy $\ttheta^* = (\QQ^{\top})^{-1} \boldsymbol{\vartheta}^{\sharp*}.$

Assumptions \ref{assu:EZu}--\ref{assu:cointegration} below impose regularity conditions on the regression models (\ref{eq:equivalent1}) and (\ref{eq:equivalent2}). In our context, we must characterize both the feasible and the infeasible models to address the complexity of cointegration, presented in the following assumptions.

\begin{assumption}
\label{assu:EZu} $\mathbb{E}(\zz_{t-1}u_{t})={\boldsymbol 0}_{p_{z}}$ and $\mathbb{E}(\vv^{\rm a}_{t-1}u_{t})=\boldsymbol{0}_{p_{\rm a}}$ for all $t\in\mathbb{Z}$.
\end{assumption}

\begin{assumption}
\label{assu:cointegration} (a)
$\| \ttheta^* \|_0
\leq s$; (b) $\max\limits_{\mathcal{I}\subset [p_{\rm a}],|\mathcal{I}|\leq Ks}\sum\limits_{j\in[p_{\rm b}]}\boldsymbol{1}\{\AA_{\mathcal{I},j}\neq \boldsymbol{0}_{|\mathcal{I}|}\} = O(s)$ for some absolute constant $K\geq 1$; (c) $\max\limits_{i\in[p_{\rm a}]}\|\AA_{i,\cdot}\|_{1}+\left[\min\limits_{i\in[p_{\rm a}]}\|\AA_{i,\cdot}\|_2\right]^{-1}=O(1)$; (d) $
\max\limits_{j\in[p_x]}|c_j| + \max\limits_{j\in[p_{\rm b}]} |c_j^{\rm b}| = O(1)$.
\end{assumption} 

Assumption \ref{assu:EZu} imposes moment conditions that the error term and the stationary regressors in (\ref{eq:equivalent2}) are uncorrelated, ensuring that the true parameters are identified. Assumption \ref{assu:cointegration}(a)
entails the sparsity of the feasible model  (\ref{eq:equivalent1}). 
Assumption \ref{assu:cointegration}(b) imposes sparsity of the cointegration system (\ref{eq: triangular}). It means that any submatrix of  $\AA$ comprising $O(s)$ rows has at most $O(s)$ nonzero columns. Consequently,  $\AA^\top\aalpha$ is sparse for any sparse vector $\aalpha$, thereby ensuring the sparsity of the infeasible model (\ref{eq:equivalent2}). Assumption \ref{assu:cointegration}(c) regularizes the cointegration matrix $\AA$ by bounding its norms and singular values.  Assumption \ref{assu:cointegration}(d) controls the magnitudes of the local unit root coefficients for $\xx_t$ and $\xx^{\rm b}_t$.

Finally, Assumption \ref{assu:asym_n} below specifies the relative size of $n$, $p$, and $s$.

\begin{assumption}
\label{assu:asym_n} (a) $p=O(n^{\nu})$ for an absolute
constant $\nu\in\left(0,\infty\right)$; (b)\textcolor{red}{{} }$s=O\left(n^{1/4-\zeta}\wedge p^{1-\zeta}\right)$
for an arbitrarily small constant $\zeta>0$.
\end{assumption}

Assumption \ref{assu:asym_n}(a) allows for an arbitrary polynomial growth of $p$. Proofs can still go through if we relax $p$ to grow at some exponential rate of $n$. But the speed of such rates
will be peculiar to each of our convergence statements, thereby complicating the notations. Assumption \ref{assu:asym_n}(b) requires that the sparsity level $s$ cannot grow faster than $n^{1/4}$ or $p$.

\subsection{The Oracle OLS Estimator} \label{subsec: OLS}
We first establish the convergence rate of the oracle OLS estimator as a benchmark for subsequent analysis, whose theoretical property under cointegration is unexplored in the literature. The oracle estimator $\hat\ttheta^{\rm or}$ is defined as
\begin{equation}\label{eq:oracle-ols}
  \hat\ttheta^{\rm or}_{\mathcal{S}} = \left(\sum_{t=1}^n \ww_{\mathcal{S},t-1}\ww_{\mathcal{S},t-1}^\top\right)^{-1}\sum_{t=1}^n \ww_{\mathcal{S},t-1}y_t \quad \text{ and } \quad \hat\ttheta^{\rm or}_{\mathcal{S}^{c}} = \mathbf{0}_{|\mathcal{S}^{c}|},\end{equation}
where $\ww_{\mathcal{S},t-1} = (w_{j,t-1})_{j\in\mathcal{S}}$. The following proposition establishes the convergence rate of the oracle estimator $\hat\ttheta^{\rm or}$. Define $\mathcal{F}_n = \sigma(\{\vareepsilon_t\}_{t\leq n})$ as the $\sigma$-field generated by the innovations $(\vareepsilon_t)_{t\leq n}$ in (\ref{eq: v=Phiepsilon}).
 Recall that the sparsity level $s$ is defined in Assumption \ref{assu:cointegration}, and the rotation matrix $\QQ$ is defined in (\ref{eq: rotation}). 
\begin{proposition}\label{prop: OLS}
  Suppose that Assumptions \ref{assu:tail}--\ref{assu:asym_n} hold. Then for $q=1,2$,
  \begin{equation}\label{eq: OLS1}\|\RR_n \QQ^{\top} (\hat\ttheta^{\rm or}-\ttheta^*)\|_q
  \stackrel{\mathrm{p}}{\preccurlyeq} s^{(5-q)/2}(\log p)^{2+1/r},
  \end{equation}
  where $\RR_n = {\rm diag}(\sqrt{n}\bfI_{p_z+p_{\rm a}}, n\bfI_{p_{\rm b}+p_x})$.
  As a result, the feasible coefficients satisfy
  \begin{equation}\label{eq: OLS2}\|\hat\ggamma^{\rm or}-\ggamma^*\|_q+\|\hat\pphi^{\rm or}-\pphi^{*}\|_q+\sqrt{n}\|\hat\bbeta^{\rm or}-\bbeta^*\|_q \stackrel{\mathrm{p}}{\preccurlyeq}  
s^{(5-q)/2}(\log p)^{2+1/r} / \sqrt{n},
  \end{equation}
  and the OOS prediction error satisfies
  \begin{equation}\label{eq: oracle prediction}
    \mathbb{E}\left[(y_{n+1} - \ww_n^\top \hat\ttheta^{\rm or})^2|\mathcal{F}_n\right]  - \mathbb{E}\left[u_{n+1}^2|\mathcal{F}_n\right] \overset{\rm p}{\to} 0.
  \end{equation}
\end{proposition}

Given the infeasible rotation matrix $\QQ$, the coefficient
$ \boldsymbol{\vartheta}^{\sharp*} = \QQ^{\top} \ttheta^* $ in (\ref{eq: OLS1})
is separated by two rates in $\RR_n$:  the stationary components converge at the standard $1/\sqrt{n}$ rate (up to a factor involving $s$ and $\log p$), and the local unit roots enjoy
\emph{super-consistency} with a faster convergence rate $1/n$.
Without $\QQ$, the feasible OLS coefficients in \eqref{eq: OLS2} also have two rates: the components tied to the stationary regressors and cointegrating system converge at the $1/\sqrt{n}$ rate, and only those associated with the non-cointegrated LURs converge at the faster $1/n$ rate. Though the cointegrated regressors individually appear as LUR processes, they comove closely together so that their convergence rate falls from $1/n$ to $1/\sqrt{n}$. This phenomenon is similar to the slowdown of convergence rates in (low-dimensional) OLS estimation when the regressors are nearly perfectly collinear.
Finally, prediction consistency is delivered by (\ref{eq: oracle prediction}),  showing that conditionally on the information set $\mathcal{F}_n$ the OOS mean squared prediction error of the oracle estimator converges to the conditional second moment of the innovation $u_{n+1}$.

With the oracle OLS convergence rate as a benchmark, we are ready to investigate the properties of the feasible FCP estimators under cointegration.

\subsection{Restricted Global Optimality}\label{subsec: global}

In this section we establish the error bound for the global optimizer for the constrained FCP (\ref{eq: constrained regularized estimation}) without the beta-min condition. As we have mentioned in the Introduction, the literature has explored alternative paths to address the challenge in characterizing the global optimizer, while these solutions depend on either the RSC condition \citep{loh2015regularized} or a consistent LASSO estimator \citep{fan2014strong,wang2013calibrating}, which are not available under cointegration.
We therefore turn to the constrained optimization problem (\ref{eq: constrained regularized estimation}) that restricts the search space to a region with a bounded penalty which encourages sparser solutions and helps establish the consistency of the global optimizer.

We impose the following condition on the penalty function $J_\lambda(\cdot)$.
\begin{condition}\label{cond:penalty}  Suppose that (a) the continuous function $J_\lambda(t)$ is symmetric around zero (i.e.~$J_\lambda(-t) = J_\lambda(t)$ for all $t \geq 0$), nondecreasing and concave on $[0,\infty)$ with $J_\lambda(0) = 0$; (b) $J_\lambda(t)$ is  differentiable on $(0,\infty)$, with $\lim_{t\to 0^+}J_\lambda'(t) \geq  a_1\lambda$ for some absolute constant $a_1 > 0$; (c) there exist absolute constants  $b > a > 0$ and $a_2\geq a_1$ such that $a_1\lambda \leq J_\lambda'(t) \leq  a_2 \lambda$ for $0 < t \leq a \lambda$, and  $J_\lambda'(t) = 0$ for $t > b \lambda$.
\end{condition}
Condition \ref{cond:penalty} follows the definition of FCPs in \citet{wang2013calibrating} that covers the widely used SCAD and MCP as special cases; see also \citet{fan2011nonconcave} and \citet{fan2014strong}  for alternative specifications of FCPs. The following theorem characterizes the global optimizer. 
\begin{theorem}\label{thm: global}
  Suppose that the penalty function $J_\lambda(\cdot)$ satisfies Condition \ref{cond:penalty}. In addition, there exist sufficiently large absolute constants $(C_1, C_2)$ such that $\lambda = C_1 (\log p)^{1+1/r}/\sqrt{n}$ and $\mu = C_2\cdot s \lambda^2$. Then
  under Assumptions \ref{assu:tail}--\ref{assu:asym_n}, the global minimizer
   $\tilde\ttheta = (\tilde\ggamma^\top, \tilde\pphi^\top,  \tilde\bbeta^\top)^\top$ of the constrained optimization problem \eqref{eq: constrained regularized estimation} satisfies the following error bound: for $q=1,2$,
  \begin{equation*}\left\|\RR_n \QQ^{\top}\bigl(\tilde\ttheta-\ttheta^*\bigr)\right\|_q
  \stackrel{\mathrm{p}}{\preccurlyeq} s^{(5-q)/2}(\log p)^{2+1/r},
  \end{equation*}
  and the feasible coefficients satisfy
    \begin{equation*}\|\tilde\ggamma-\ggamma^*\|_q+\|\tilde\pphi-\pphi^{*}\|_q+\sqrt{n}\|\tilde\bbeta-\bbeta^*\|_q \stackrel{\mathrm{p}}{\preccurlyeq} 
      s^{(5-q)/2}(\log p)^{2+1/r} / \sqrt{n}.
  \end{equation*}
Furthermore, the OOS prediction error satisfies
  \begin{equation*}
    \mathbb{E}\left[(y_{n+1} - \ww_n^\top \tilde\ttheta)^2\middle|\mathcal{F}_n\right]  - \mathbb{E}\left[u_{n+1}^2|\mathcal{F}_n\right] \overset{\rm p}{\to} 0.
  \end{equation*}

\end{theorem}

Theorem \ref{thm: global} provides a theoretical guarantee for the global optimizer of the constrained problem. The global minimizer of the constrained optimization problem (\ref{eq: constrained regularized estimation}) enjoys both the parameter estimation consistency and OOS prediction consistency, with the same convergence rates as the oracle OLS estimator in Proposition \ref{prop: OLS}. This result builds the foundation for using the constrained FCPs for OOS prediction in high-dimensional cointegrated regressions.

Intuitively, the upper bound of the total penalty level in (\ref{eq: constrained regularized estimation}) avoids potential overfitting caused by a large number of small active coefficients in the estimator, thereby safeguarding consistency of the global optimizer. Heuristic discussions of the technical proofs are provided in Remark \ref{rem: RE} below.

To characterize variable selection consistency, the next section establishes the oracle properties of a local minimizer of the vanilla FCP (\ref{eq: regularized estimation}) and the global minimizer of the constrained FCP (\ref{eq: constrained regularized estimation}) under an additional beta-min condition.

\subsection{Oracle Properties}\label{subsec: oracle}
We impose the following condition on the minimum gap of the nonzero
coefficients from zero. Define $\mathcal{M}_x$ as the index set of the LUR regressors $\xx_t$. 

\begin{assumption}\label{assu:beta-min}
 $\left(\sqrt{n}\min_{j\in\mathcal{S}\cap\mathcal{M}_x}|\theta^*_j|\right)\wedge  \left(\min_{j\in\mathcal{S}\backslash\mathcal{M}_x}|\theta^*_j|\right)
 \gg s^{3/2}(\log p)^{3+1/r} / \sqrt{n}.$
\end{assumption}

This assumption requires the nonzero coefficients to be sufficiently large and distinguishable from zero, which is necessary for consistent variable selection. For nonzero coefficients of the active LUR regressors that are not cointegrated, the minimum gap from zero has an order larger than $s^{3/2}(\log p)^{3+1/r}/n$. In contrast, for the nonzero coefficients of the active regressors that are stationary or cointegrated, the minimum gap from zero has an order larger than $s^{3/2}(\log p)^{3+1/r}/\sqrt{n}$. This is reasonable because the LUR regressors are much more persistent than the stationary regressors as well as the cointegrating residuals, and their coefficient estimates enjoy super-consistency at the rate of $1/n$.
Unlike the standard beta-min condition for i.i.d.~data, the lower bound here involves the sparsity level $s$, because the RE of high-dimensional LURs is not a constant but dependent on $s$; see Proposition 3 in \citetalias{mei2024lasso}.

The following theorem establishes the oracle property of FCPs under cointegration.

\begin{theorem}\label{thm: oracle} Suppose that the penalty function $J_\lambda(\cdot)$ satisfies Condition \ref{cond:penalty}, and Assumptions \ref{assu:tail}--\ref{assu:beta-min} hold.
  \begin{enumerate}
  \item[(a)] Let $\mathcal{A}_n(\lambda)$ be the set of local minimizers of the criterion function $\mathcal{L}_\lambda(\ttheta)$ in \eqref{eq: regularized estimation}. Then there exists a sufficiently large absolute constant $C_1$ such that if we set $\lambda = C_1\cdot s(\log p)^{5/2+1/r}/\sqrt{n}$, then $\Pr\{\hat\ttheta^{\rm or} \in \mathcal{A}_n(\lambda)\} \to 1$. 
    \item[(b)] With $\lambda$ specified in Part (a), there exists an absolute constant $C_2$ such that when $\mu = C_2s\lambda^2$, the global minimizer $\tilde\ttheta$ in \eqref{eq: constrained regularized estimation} has the oracle property
   $\Pr\{\tilde\ttheta = \hat\ttheta^{\rm or}\} \to 1$.
   \end{enumerate}
\end{theorem}

Theorem \ref{thm: oracle}(a) shows that the oracle estimator is one of the local minimizers of the unconstrained optimization problem (\ref{eq: regularized estimation}) with FCPs, and therefore distinguishable from other candidate parameters in a neighborhood of the true coefficients. Nevertheless, in practical use it is challenging to determine which local minimizer corresponds to the oracle estimator.
In contrast, Theorem \ref{thm: oracle}(b) guarantees that the global minimizer of the constrained FCP (\ref{eq: constrained regularized estimation}) coincides with the oracle estimator with high probability, thereby providing a more practical solution.
To ensure the oracle property, Theorem \ref{thm: oracle} requires a larger order of $\lambda$ than that in Theorem \ref{thm: global} to distinguish the nonzero coefficients from zero. We adopt data-driven tuning parameter selection for practical implementation; see Section \ref{sec: simulations} for details.

The following remark explains why it is challenging to establish the consistency of the global minimizer of the vanilla FCPs (\ref{eq: regularized estimation}), and how  the additional constraint on the penalty level helps with the technical proofs to verify the consistency of the global minimizer.

\begin{remark}\label{rem: RE}
In the high-dimensional setting with $p>n$, even for i.i.d.\ data, \citet{fan2011nonconcave} note that establishing consistency of the global minimizer under folded-concave penalties (FCPs) is challenging because the estimation error for the inactive regressors may be poorly controlled. Since SCAD and MCP are uniformly bounded at order $O(\lambda^2)$, an estimator with $\hat s=O(\lambda^{-2})$ nonzero coefficients can still incur only $O(1)$ total penalty that may be dominated by the improvement in the mean squared error, increasing the risk of overfitting. This also prevents a direct use of the RE condition, which relies on the cone condition
\begin{equation}\label{eq: inequality RE}
\|\hat\ttheta_{\mathcal{S}^{c}}\|_1
\leq
3\|\hat\ttheta_{\mathcal{S}}-\ttheta^*_{\mathcal{S}}\|_1.
\end{equation}
Such a condition automatically holds for LASSO but does not generally follow under FCPs.
Therefore, the additional constraint in (\ref{eq: constrained regularized estimation}) is essential in bounding the total penalty, which in turn controls the estimation error for the inactive regressors and validates a condition parallel to (\ref{eq: inequality RE}). This allows us to apply the RE condition for LURs in \citetalias{gao2026lasso} and establish consistency of the constrained global minimizer.
\end{remark}

To summarize, our theoretical statements clarify that FCPs can remedy the inconsistency of LASSO under cointegration. The additional constraint on the penalty level further guarantees the consistency and oracle property of the global minimizer of the penalized loss function under regularity conditions. The Monte Carlo simulations in the next section will provide supportive evidence for the theoretical predictions.

\section{Monte Carlo Simulations}\label{sec: simulations}
In this section, we compare the finite-sample performance of LASSO, MCP, SCAD, and the constrained versions MCP-C and SCAD-C. We evaluate estimation accuracy and variable selection under various DGPs that incorporate different types of persistence and cointegration structures. To save space,  in the main text we only display simulations of a DGP with the triangular representation following (\ref{eq: triangular}). Additional simulation designs of stochastic common trends and pairwise cointegration are relegated to the supplementary Section \ref{sec:add-simul-results}.

\subsection{Design}

We generate data from (\ref{eq: feasible}), with $n\in\{200,300,500\}$ and $p\in\{100,300,500\}$. The error term $u_t \overset{\rm i.i.d}{\sim}  \mathcal{N}(0,1)$. The regressors always include $p_z$ stationary regressors $\zz_t \overset{\rm i.i.d}{\sim} \mathcal{N}(\boldsymbol{0}_{p_z},\bfI_{p_z})$.
The other $(p-p_z)$ regressors are nonstationary, either with or without cointegration.  

There are $p_x$ LUR regressors satisfying $\xx_t = \rho_n\xx_{t-1} + \ee_t$, where $\ee_t\overset{\rm i.i.d}{\sim} \mathcal{N}(\bm{0}_{p_x},\bfI_{p_x})$, and the AR coefficient $\rho_n = 1$ or $1-1/n$ to characterize unity or local to unity. These two settings apply to all persistent AR processes in this section. For the cointegrated regressors $\xx_t^{\rm ab}$, we specify the triangular representation in (\ref{eq: triangular}) with $p_{\rm a}= 2$ and $p_{\rm b} = p_{\rm ab} - 2$ accordingly, where the cointegrating residual $\vv_t^{\rm a}\overset{\rm i.i.d}{\sim}  \mathcal{N}(\boldsymbol{0}_{p_{\rm a}},\bfI_{p_{\rm a}})$, the LURs $\xx_t^{\rm b} = \rho_n\xx_{t-1}^{\rm b} + \ee_t^{\rm b}$ with $\ee_t^{\rm b}\overset{\rm i.i.d}{\sim} \mathcal{N}(\bm{0}_{p_{\mathrm{b}}},\bfI_{p_{\rm b}})$. We set the cointegrating matrix as $\AA = (\AA_1,\AA_2)^{\top}$  with the sparse vectors $\AA_1=(-0.5,-0.5,{\bf 0}_{p_{\rm b}-2}^\top)^\top$ and $\AA_2=(0,0,1,{\bf 0}_{p_{\rm b}-3}^\top)^\top$. The true slope coefficients, as in (\ref{eq:equivalent1}), are $\ggamma^*=(1,1,\boldsymbol{0}_{p_z-2}^\top)^\top$, $\pphi^{\rm a*}=(1,0)^\top$, $\pphi^{\rm b*}=(0.5,0.5,\boldsymbol{0}_{p_{\rm b}-2}^\top)^\top$, and $\bbeta^*=(5/\sqrt{n}, \boldsymbol{0}_{p_x-1}^\top)^\top$. The local-to-zero coefficient $5/\sqrt{n}$ in $\bbeta^*$ characterizes weak signals of nonstationary time series in economic and financial applications. The numbers of regressors are $p_x =p/4$, $p_{\rm ab} = 2+p/4$, and thus $p_z = p - p_x - p_{\rm ab}$.

For the FCPs, we set the constant $a=3.7$ for SCAD and $a=3$ for MCP, following prevalent practice \citep{fan2001variable,zhang2010nearly}. Unlike the illustrative simulations in Section \ref{subsec:illustrative} based on a fixed grid over all candidate values of $\lambda$ in (\ref{eq: regularized estimation}), in this section we select $\lambda$ by 10-fold contiguous block cross-validation, which divides the time series into 10 sequentially consecutive blocks. To save space, we relegate the details of data-driven tuning parameter selection to the supplementary Section \ref{subsec: tuning selection}.

\subsection{Results}

\begin{figure}[htbp]
\centering
\includegraphics[width=1\textwidth]{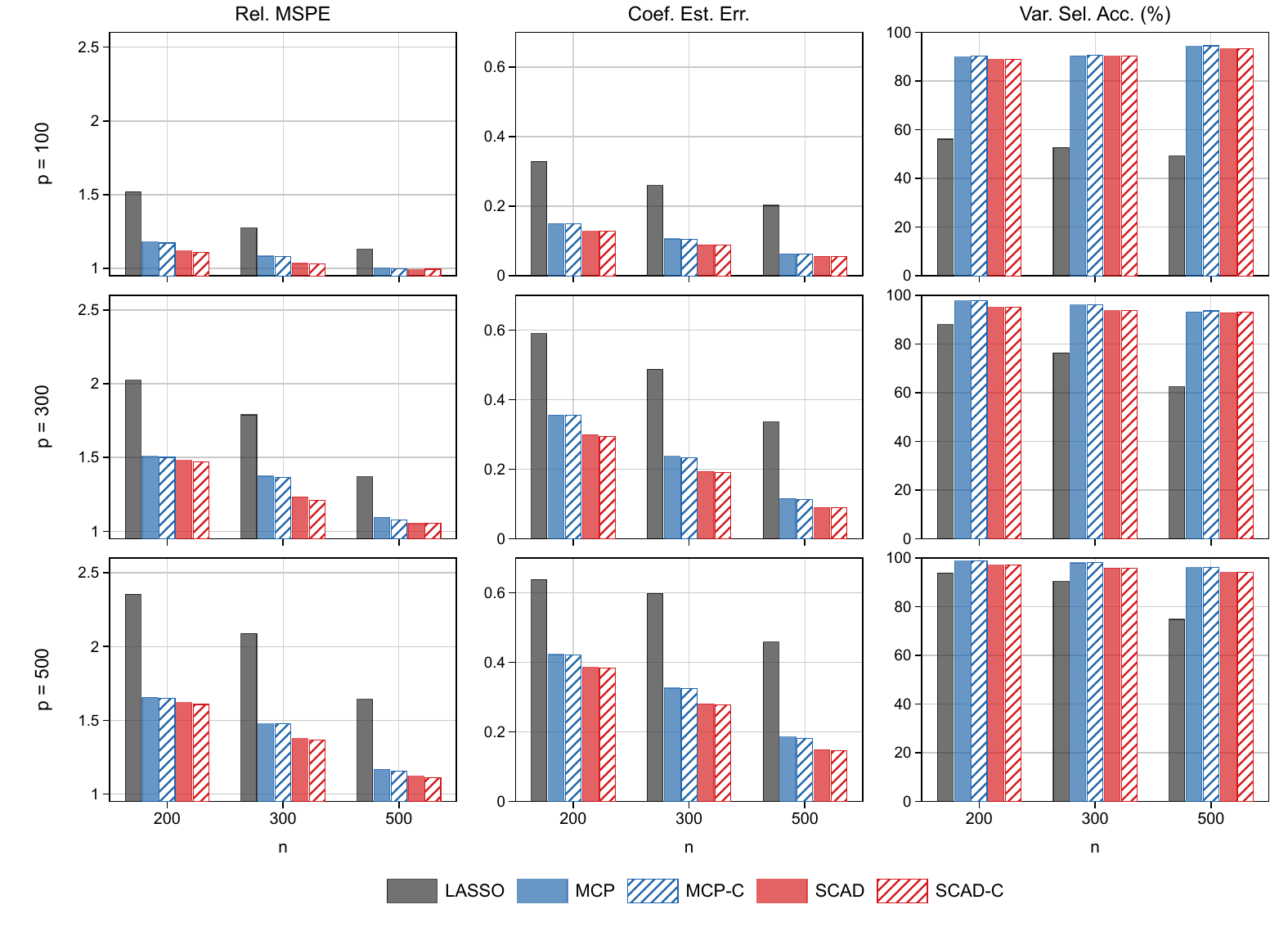}
\caption{Simulation results for unit roots ($\rho_n = 1$)}
{{\begin{flushleft}
\singlespacing\footnotesize \textit{Notes}: The $x$-axis represents the sample size $n$. The left column \texttt{Rel.MSPE} exhibits the OOS mean squared prediction error (MSPE) relative to the oracle OLS, the middle column \texttt{Coef.Est.Err.}~shows the relative $\ell_2$-norm of the coefficient estimation error, and the right column \texttt{Var.Sel.Acc.}~displays the variable selection accuracy.
\end{flushleft}}}
\label{fig:dgp1_ur}
\end{figure}

To save space,  in the main text we report the results for unit roots with $\rho_n = 1$ only. Those for LUR $\rho_n = 1 - 1/n$ are similar and relegated to the supplementary Section \ref{sec:add-simul-results}.

We first evaluate OOS prediction performance, which is the essential task of predictive regressions. We benchmark the regularization methods against the oracle OLS estimator defined in \eqref{eq:oracle-ols}. The performance indicator is the one-step-ahead mean squared prediction error (MSPE) relative to the oracle OLS estimator, defined as $\sum_{m=1}^M(\hat y_{n+1}^{(m)} - y_{n+1}^{(m)})^2/\sum_{m=1}^M( \hat{y}_{n+1}^{\text{or}(m)} - y_{n+1}^{(m)})^2$, where for the $m$-th replication $y^{(m)}_{n+1}$ is the OOS true value of the outcome,  $\hat y^{(m)}_{n+1}$ is the predicted value from a regularization method, and $\hat{y}^{\text{or}(m)}_{n+1}$ is the predicted value from the oracle OLS, respectively. The Monte Carlo is repeated $M = 1000$ times.

The left column \texttt{Rel.MSPE} of Figure \ref{fig:dgp1_ur} reports the MSPEs relative to oracle OLS. Overall, the prediction error decreases as the sample size $n$ increases, while it grows as the number of regressors $p$ increases.  The MSPE of LASSO is obviously larger than those of the regularization methods with FCPs. When the sample size gets big ($n=500$), the relative MSPEs by FCPs are close to one. These results indicate that FCPs are approaching the predictive performance of the oracle OLS, which is not achievable by LASSO.
SCAD-C and MCP-C produce slightly smaller prediction errors than their respective unconstrained counterparts in almost all settings. These findings align with the theoretical results in Section \ref{sec: theory}.

We then examine the quality of the estimated coefficients. We use the relative $\ell_2$-norm of the coefficient estimation error, defined as $\|\hat\ttheta - \ttheta^*\|_2/\|\ttheta^*\|_2$, to measure the accuracy of the estimated coefficients. The variable selection accuracy is measured by the proportion of correctly classified regressors, defined as $p^{-1}\sum_{j=1}^p \textbf{1}\{ {\rm sgn}(|\hat\theta_{j}|) = {\rm sgn}(|\theta_j^*|)\}$, where $\hat\theta_j$ denotes a generic estimator of $\theta_j^*$.

The middle column \texttt{Coef.Est.Err.}~of Figure \ref{fig:dgp1_ur} reports the relative $\ell_2$ error of all five methods, and the right column \texttt{Var.Sel.Acc.}~reports the average variable selection accuracy. The results are largely consistent with the prediction performance: LASSO has the worst parameter estimation and variable selection performance, while the FCPs can achieve substantial improvements with small parameter estimation errors and high variable selection accuracy close to one. These simulation results highlight the importance of using FCPs in high-dimensional predictive regressions.

\section{Predicting the Euro Area Unemployment Rate}\label{sec: empirical}

\defcitealias{smeekes2020unit}{Smeekes and Wijler, 2020}

Applications of machine learning to forecasting real-world time series have been widely documented in the literature
\citep{smeekes2018macroeconomic,medeiros2021forecasting,giannone2021economic}.
While researchers used to transform the data into stationary time series before implementing estimation and prediction, more recently they employ machine learning for nonstationary data and provide empirical evidence for the benefits of using raw data, in particular when the predicted outcome is highly persistent \citepalias{smeekes2020unit,mei2024lasso}. Our empirical application examines the prediction of the euro area unemployment rate using high-dimensional macroeconomic time series.

\subsection{Data}

We use the database \texttt{EA-MD-QD},\footnote{The dataset is available at \url{https://www.barigozzi.eu/EA.html}.}
a euro area (EA) macroeconomic dataset contributed by \citet{barigozzi2026large}.
It contains monthly and quarterly
time series data of aggregate euro area (EA) macroeconomic variables, along with country-level data for the ten largest member
countries in EA, including  Austria, Belgium,
Germany, Greece, Spain, France, Ireland, Italy,
the Netherlands, and Portugal. Our target variable is the monthly euro area total
unemployment rate. We use the monthly data with seasonal adjustment as provided by \citet{barigozzi2026large}, covering the period from January 2000 to September 2025.

\begin{figure}[htbp]
\centering
\includegraphics[width=0.65\textwidth]{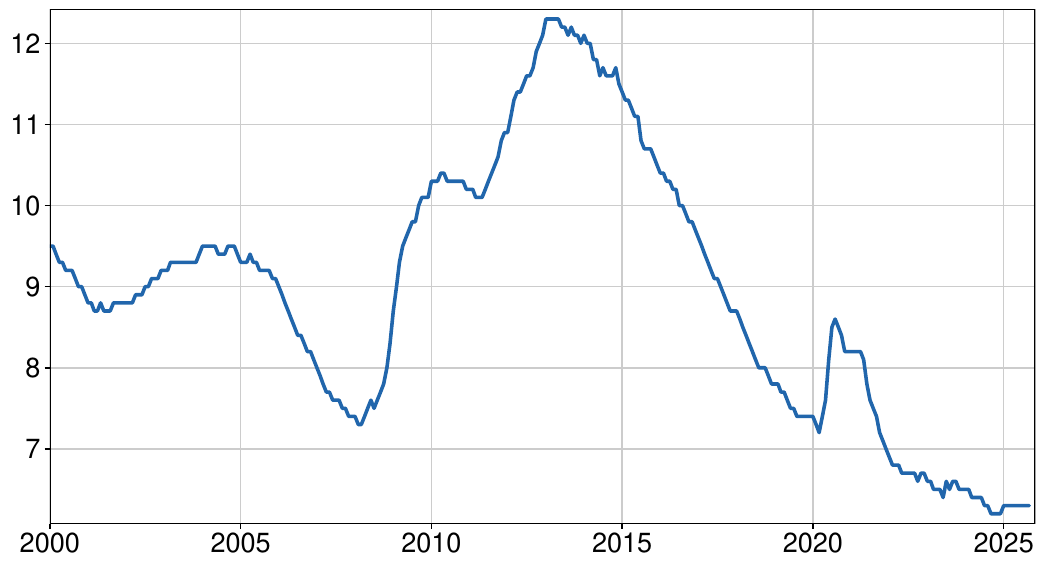}
\caption{Euro Area Unemployment Rate}
\label{fig: unrate}
\end{figure}

Figure \ref{fig: unrate} plots the time series of the euro area unemployment rate. It exhibits a clear downward trend from 2000 to 2008, followed by a sharp increase during the global financial crisis in 2009. The unemployment rate then peaks in 2013 and gradually declines until 2019, before rising again during the COVID-19 pandemic in 2020. The AR(1) coefficient estimated by OLS is 1.0022, which is virtually unity, and the augmented Dickey-Fuller (ADF) test cannot reject the null hypothesis of a unit root ($p$-value 0.7043). These observations indicate that the euro area unemployment rate is highly persistent.

We consider two sets of regressors. The first specification contains $p=38$ regressors of EA-level aggregate data: the lagged EA unemployment rate and 37 other variables with full observations. The second specification contains $p=368$ regressors, including both EA aggregate indicators and country-level variables. This augmentation is motivated by recent advancements in machine learning for macroeconomic prediction \citep{cui2026asset,cui2026forecasting}, showing that using individual level data rather than aggregate data can improve the prediction of essential macroeconomic time series, such as inflation and GDP.

\begin{figure}[htbp]
\centering
\includegraphics[width=0.86\textwidth]{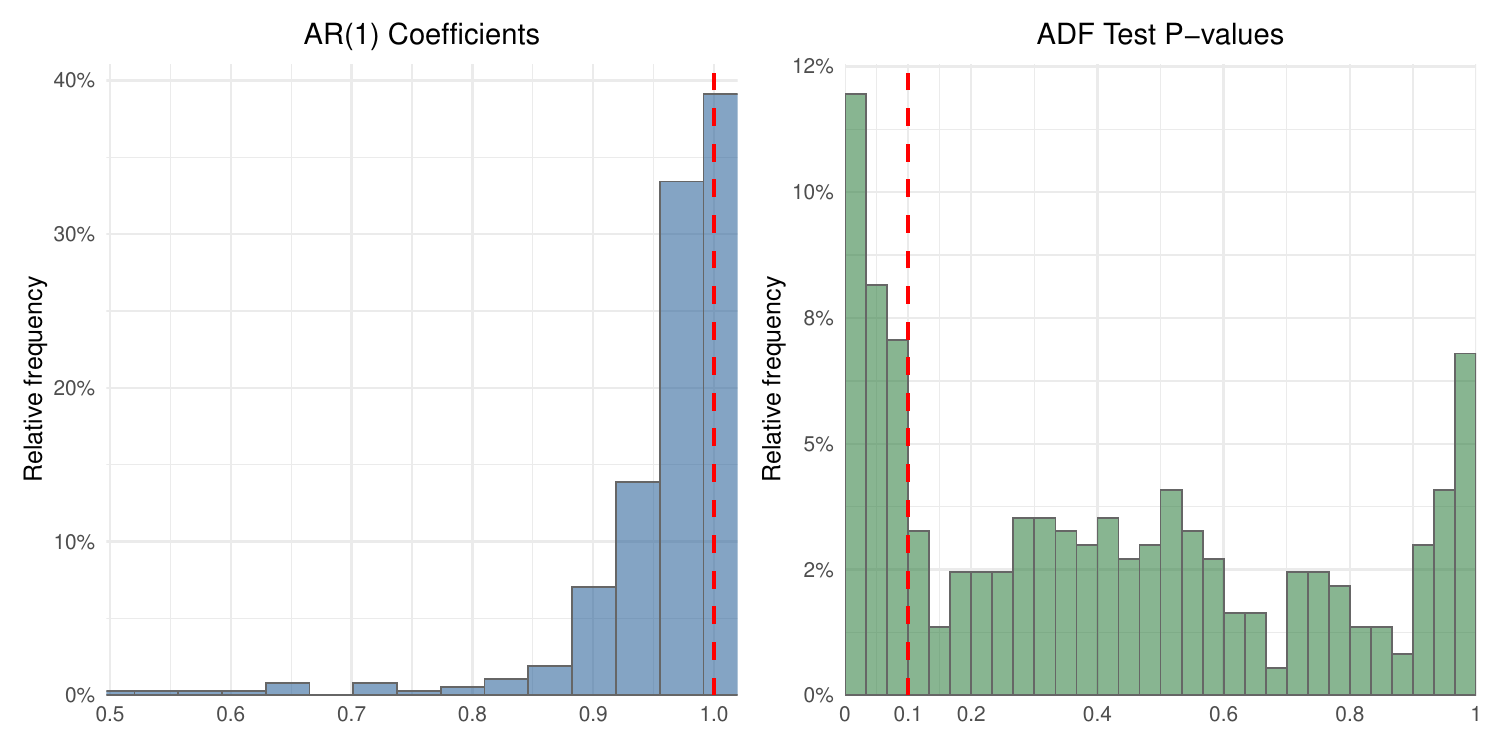}
\caption{Histograms of AR(1) Coefficients and ADF Test $p$-values of Predictors} 
\label{fig: ar1_adf}
\end{figure}

Many macroeconomic series are highly persistent. Figure \ref{fig: ar1_adf} exhibits histograms of their AR(1) coefficients estimated by OLS and their ADF test $p$-values. The left panel shows that a majority of the AR(1) coefficients are close to unity. Specifically, $91.6\%$ of the series have an AR(1) coefficient of at least $0.90$. The right panel shows that approximately $72.8\%$ of the series have ADF test $p$-values above $10\%$, indicating that the unit root null cannot be rejected for most series.

 Cointegration is another important feature of the data.
According to \citet{barigozzi2026large}, the 37 external EA-only predictors can be classified by their macroeconomic content into financial markets and interest rates, industrial production and turnover, prices, and confidence. The time series within each category are closely related and may be cointegrated. Cointegration may also exist across categories of nonstationary time series. We perform the cointegration test for large VAR models proposed by \citet{bykhovskaya2024asymptotics} on each category of EA-only predictors.
Table \ref{tab:cointegration_test} reports the test statistics and $p$-values. The null hypothesis of no cointegration is rejected at the 1\% level for industrial production and turnover, prices, and confidence indicators, which shows statistical evidence of cointegration.  

\begin{table}[t]
    \centering
    \caption{Cointegration Tests for the EA-only Predictor Set}
    \label{tab:cointegration_test}
    \begin{tabular}{lccc}
        \toprule
        Variables & $N$ & Test Statistic & $p$-value \\
        \midrule
        Financial Markets and Interest Rates & 3  & $-1.65$ & 0.63 \\
        Industrial Production and Turnover   & 14 & 23.37 & $\leq 0.01$ \\
        Prices                                & 12 & 59.52 & $\leq 0.01$ \\
        Confidence Indicators                 & 8  & 47.92 & $\leq 0.01$ \\
        \bottomrule
    \end{tabular}
{{\begin{flushleft}
\singlespacing\footnotesize \textit{Notes}: The cointegration test is implemented using the R package \texttt{Largevars} \citep{bykhovskaya2025largevarsrpackagetesting}, with lag order $k=1$ and maximum cointegration rank $r=1$. $N$ is the number of variables included in each test. A statistic above the 1\% critical value of 2.02 rejects the null hypothesis of no cointegration.
\end{flushleft}}}
\end{table}

\subsection{Variable Selection and Prediction}
We conduct $h$-month ahead rolling-window predictions, with $h = 1, 2, 3$. We consider three testing periods: 2010--2025, 2015--2025, and
2020--2025. For 2010--2025, the window length is set to 10 years with $n=120$ months, since only 10 years of data are available before 2010. We consider $n=120$ and 180 for the second testing period 2015--2025, and $n=120,180,240$ for the third period 2020--2025. The implementation of LASSO, SCAD, MCP, SCAD-C and MCP-C follows Section \ref{sec: simulations}.

\begin{table}[!t]
\centering
\caption{Average Numbers of Selected Variables by Category with EA-only Data}
\label{tab:emp-variable-selection}
\setlength{\tabcolsep}{2.5pt}
\begin{tabular}{lccccc}
\toprule
\cmidrule{1-6}
Method & \makecell{Finance\\(3)} & \makecell{Industry\\(14)}
& \makecell{Price\\(12)} & \makecell{Confidence\\(8)}
& \makecell{Total\\(37)} \\
\midrule
LASSO   & 0.037 & 0.048 & 0.041 & 2.051 & 2.18 \\
SCAD    & 1.554 & 2.891 & 1.212 & 2.276 & 7.93 \\
SCAD-C  & 1.580 & 2.935 & 1.179 & 2.299 & 7.99 \\
MCP     & 1.560 & 2.735 & 1.152 & 2.186 & 7.63 \\
MCP-C   & 1.594 & 2.758 & 1.156 & 2.227 & 7.74 \\
\bottomrule
\end{tabular}
{{\begin{flushleft}
\singlespacing\footnotesize \textit{Notes}: The entries are the average numbers of selected predictors across all windows and forecasting scenarios. Parentheses represent the number of candidate external predictors in each category. ``Finance'' refers to financial markets and interest rates; ``Industry'' is industrial production and turnover; ``Price'' denotes price indicators; ``Confidence'' represents confidence indicators.
\end{flushleft}}}
\end{table}

We first examine variable selection. Table \ref{tab:emp-variable-selection} reports the average number of external predictors selected from each economic category with EA-only data, using the same classification as Table \ref{tab:cointegration_test}. All five methods select sparse models. LASSO selects only 2.18 predictors on average, 2.05 of which are confidence indicators. In contrast, the four FCP methods select between 7.63 and 7.99 predictors on average, spanning a broader range of categories including financial markets, industrial activity, confidence indicators, and prices. Given the cointegration documented in Table \ref{tab:cointegration_test},
the observation that LASSO selects a sparser model is consistent with its tendency to overly shrink the coefficients of cointegrated predictors.

\begin{table}[!t]
\caption{Forecast RMSEs for Euro Area Unemployment Rate}
\label{tab:emp-level-rmse}
\small
\setlength{\tabcolsep}{1.5pt}
\onehalfspacing
\begin{center}
\begin{tabular}{cc|c|ccccc|ccccc}
\hline\hline
\multirow{2}{*}{$n$}
& \multirow{2}{*}{$h$}
& \multicolumn{1}{c|}{\multirow{2}{*}{\makecell[c]{Best\\BM}}}
& \multicolumn{5}{c|}{EA data only}
& \multicolumn{5}{c}{EA and country level data}
\\ \cline{4-13}
& & \multicolumn{1}{c|}{}
& LASSO & SCAD & SCAD-C & MCP & \multicolumn{1}{c|}{MCP-C}
& LASSO & SCAD & SCAD-C & MCP & MCP-C
\\
\hline
\multicolumn{13}{c}{\textbf{Testing Sample: 2010--2025}} \\
\hline
\multirow{3}{*}{120} & 1 & 0.0969 & 0.1049 & 0.0870 & \textit{\textbf{0.0862}} & 0.0864 & \textit{\textbf{0.0862}} & 0.1064 & 0.0899 & \textit{0.0884} & 0.0898 & 0.0900 \\
 & 2 & 0.1657 & 0.1777 & 0.1395 & \textit{0.1394} & 0.1435 & 0.1398 & 0.1817 & 0.1355 & \textit{\textbf{0.1332}} & 0.1356 & 0.1371 \\
 & 3 & 0.2494 & 0.2490 & 0.2210 & 0.2151 & \textit{0.2144} & 0.2178 & 0.2561 & 0.2069 & \textit{\textbf{0.2067}} & 0.2141 & 0.2136 \\
\hline
\noalign{\vskip 5pt}
\multicolumn{13}{c}{\textbf{Testing Sample: 2015--2025}} \\
\hline
\multirow{3}{*}{120} & 1 & 0.0997 & 0.1073 & 0.0870 & 0.0860 & 0.0848 & \textit{\textbf{0.0845}} & 0.1093 & 0.0885 & \textit{0.0861} & 0.0870 & 0.0875 \\
 & 2 & 0.1793 & 0.1898 & \textit{0.1459} & 0.1463 & 0.1526 & 0.1489 & 0.1947 & 0.1365 & \textit{\textbf{0.1333}} & 0.1372 & 0.1391 \\
 & 3 & 0.2622 & 0.2673 & 0.2334 & 0.2267 & \textit{0.2249} & 0.2285 & 0.2739 & \textit{\textbf{0.2198}} & 0.2199 & 0.2243 & 0.2213 \\
\hline
\multirow{3}{*}{180} & 1 & 0.0959 & 0.1003 & \textit{0.0823} & 0.0835 & 0.0877 & 0.0869 & 0.1058 & 0.0792 & 0.0819 & 0.0789 & \textit{\textbf{0.0782}} \\
 & 2 & 0.1783 & 0.1736 & 0.1432 & 0.1418 & \textit{0.1322} & 0.1328 & 0.1855 & 0.1245 & 0.1255 & 0.1191 & \textit{\textbf{0.1153}} \\
 & 3 & 0.2537 & 0.2443 & 0.2075 & 0.2145 & 0.2035 & \textit{0.1979} & 0.2610 & 0.1836 & \textit{\textbf{0.1817}} & 0.1920 & 0.1913 \\
\hline
\noalign{\vskip 5pt}
\multicolumn{13}{c}{\textbf{Testing Sample: 2020--2025}} \\
\hline
\multirow{3}{*}{120} & 1 & 0.1161 & 0.1195 & 0.0989 & 0.0974 & 0.0981 & \textit{\textbf{0.0961}} & 0.1198 & 0.1046 & \textit{0.1003} & 0.1027 & 0.1031 \\
 & 2 & 0.2071 & 0.2091 & \textit{0.1638} & 0.1644 & 0.1775 & 0.1707 & 0.2093 & 0.1642 & \textit{\textbf{0.1587}} & 0.1655 & 0.1666 \\
 & 3 & 0.2842 & 0.2874 & 0.2780 & 0.2780 & \textit{\textbf{0.2741}} & 0.2787 & 0.2883 & 0.2796 & 0.2794 & 0.2829 & \textit{0.2782} \\
\hline
\multirow{3}{*}{180} & 1 & 0.1150 & 0.1139 & \textit{0.0929} & 0.0946 & 0.1014 & 0.1001 & 0.1182 & 0.0924 & 0.0957 & 0.0924 & \textit{\textbf{0.0912}} \\
 & 2 & 0.2092 & 0.2002 & 0.1625 & 0.1610 & 0.1451 & \textit{0.1424} & 0.2047 & 0.1502 & 0.1506 & 0.1408 & \textit{\textbf{0.1343}} \\
 & 3 & 0.2877 & 0.2814 & 0.2453 & 0.2560 & 0.2255 & \textit{\textbf{0.2176}} & 0.2845 & 0.2247 & \textit{0.2233} & 0.2336 & 0.2312 \\
\hline
\multirow{3}{*}{240} & 1 & 0.1115 & 0.1149 & 0.0900 & 0.0891 & \textit{\textbf{0.0869}} & 0.0884 & 0.1178 & 0.0950 & 0.0939 & 0.0922 & \textit{0.0891} \\
 & 2 & 0.2086 & 0.1967 & 0.1622 & 0.1646 & 0.1637 & \textit{0.1580} & 0.2049 & 0.1313 & \textit{\textbf{0.1307}} & 0.1429 & 0.1426 \\
 & 3 & 0.2868 & 0.2649 & 0.2405 & 0.2428 & 0.2380 & \textit{0.2378} & 0.2796 & \textit{\textbf{0.1919}} & 0.2083 & 0.2118 & 0.2050 \\
\hline
\hline
\end{tabular}
\end{center}
\footnotesize{
\textit{Notes:}
The column ``Best BM'' reports the lowest RMSE among the random walk (RW), random walk with drift (RWwD), and autoregression (AR) benchmarks for each scenario.
Italic numbers indicate the best performance among the five regularization methods using the same predictor set.
Bold numbers highlight the lowest RMSE in each row, including the benchmark.}
\end{table}

We then turn to the OOS prediction. We compare these regularization methods to three benchmarks: (1) the random walk (RW), which uses the prediction
$\widehat y_{t+h\mid t}^{\mathrm{RW}}=y_t$; (2) the random walk with drift (RWwD)
$\widehat y_{t+h\mid t}^{\mathrm{RWwD}}
=y_t+\{h/n\}(y_t-y_{\rm init})$, where $y_{\rm init}$ is the outcome of the initial period within the rolling window; (3) the autoregression (AR) with the lag order selected by the Bayesian information criterion (BIC). For each scenario we report the result of the best-performing benchmark method with the lowest prediction error, which we call the best benchmark (Best BM).

Table \ref{tab:emp-level-rmse} reports the OOS root mean squared error
(RMSE). In all scenarios, the FCPs outperform the best benchmark and also dominate LASSO. In addition, in 14 out of the 18 scenarios, the lowest RMSE is achieved by a constrained FCP method. These results echo the theoretical findings and simulation evidence that the FCPs are robust to cointegration under which LASSO can be inconsistent, and the additional constraint on the penalty level can further guarantee the consistency of the global optimizers of the criterion functions regularized by FCPs.

Comparison of the results using the predictors of EA level only to those using both the aggregate and country level predictors is informative. Comparing the lowest RMSE achieved under each predictor set, we find that adding country level predictors improves the prediction accuracy in 12 out of the 18 scenarios. The improvement is evident when the horizon $h=2$ and 3, and when the sample size $n$ is relatively large. These results are consistent with the recent findings of the value of individual-level data beyond aggregate data for forecasting \citep{cui2026asset,cui2026forecasting}.
Our FCP-based estimation allows the use of cointegrated data, opening the door for the use of many individual-level time series in predictive exercises.

\section{Conclusion}
While LASSO is inconsistent in high-dimensional predictive regressions with cointegration, we recommend using the FCPs in this environment with a constraint to bound the total penalty level. We establish coefficient estimation consistency and the oracle property for the recommended procedure
when stationary, nonstationary, and cointegrated time series are present.
Our numerical results confirm that FCPs can be a robust alternative to LASSO with regressors of diverse dynamic behaviors.

\onehalfspacing \bibliographystyle{chicagoa}
\small
\bibliography{ref}

\doublespacing
\normalsize

\newpage{}
 
\begin{appendices}
\begin{center}
{\LARGE Supplementary Materials for ``Cointegrated Regressors: Where Folded Concave Penalties Go Beyond LASSO''}
\end{center}

\begin{center}
{\onehalfspacing\large Hongqi Chen$^a$, Chengwang Liao$^b$, Ziwei Mei$^c$, Zhentao Shi$^b$ \\

\vspace{1em}

$^a$Hunan University\\ 
$^b$The Chinese University of Hong Kong\\ 
$^c$University of Macau\par}
\end{center}

\setcounter{footnote}{0}
\setcounter{table}{0}
\setcounter{figure}{0}
\setcounter{equation}{0}
\renewcommand{\thefootnote}{\thesection.\arabic{footnote}}
\renewcommand{\theequation}{\thesection.\arabic{equation}}
\renewcommand{\thefigure}{\thesection.\arabic{figure}}
\renewcommand{\thetable}{\thesection.\arabic{table}}
\renewcommand{\thetheorem}{\thesection.\arabic{theorem}}
\renewcommand{\thelemma}{\thesection.\arabic{lemma}}
\renewcommand{\theremark}{\thesection.\arabic{remark}} 
\renewcommand{\theproposition}{\thesection.\arabic{proposition}}

\defcitealias{mei2024lasso}{MS24}
\defcitealias{gao2026lasso}{GLMS26}

Section \ref{sec:AdditionalDiscussions} provides additional discussions on the theoretical results presented in the main text.
Section \ref{sec:Proofs} includes the technical proofs of all theoretical statements. Section \ref{sec:add-simul} collects additional simulation
results omitted from the main text.

\section{Additional Discussions\protect\label{sec:AdditionalDiscussions}}
\subsection{Inconsistency of LASSO for General Cointegrated Regressors}
\label{subsec:InconsistencyLASSO}

This section formulates the inconsistency of LASSO in general high-dimensional cointegrated predictive regression beyond the bivariate toy model in Section \ref{subsec: inconsistency} of the main text.

In view of the rotation from (\ref{eq:equivalent1}) to (\ref{eq:equivalent2}),
let
$\breve\ttheta = (\ggamma^{*\top}, \breve\pphi^{\top}, \bbeta^{*\top})^\top$ with $\breve\pphi = (\breve\pphi^{\rm a\top}, \breve\pphi^{\rm b\top})^\top$ be an
\emph{alternative coefficient}, such that
\begin{equation}\label{eq: Aphi1+phi2}
  \AA^\top\breve\pphi^{\rm a} + \breve\pphi^{\rm b} = \AA^\top\pphi^{\rm a*} + \pphi^{\rm b*}.
\end{equation}
The average of squared residuals under $\breve\ttheta$ is
\begin{align*}
    \dfrac{1}{n}\sum_{t=1}^n \left(y_t - \zz_{t-1}^\top\ggamma^* - \vv_{t-1}^{\rm a\top}\breve\pphi^{\rm a} - \xx_{t-1}^{\rm b\top}(\AA^\top\breve\pphi^{\rm a}+\breve\pphi^{\rm b}) - \xx_{t-1}^\top\bbeta^*\right)^2 =
     \dfrac{1}{n}\sum_{t=1}^n \left(u_t - \vv_{t-1}^{\rm a\top}(\breve\pphi^{\rm a} - \pphi^{\rm a*})\right)^2,
\end{align*}
which is $O_p(1)$ if $\|\breve\pphi^{\rm a} - \pphi^{\rm a*}\|_1$ is bounded, since $\vv_t^{\rm a}$ and $u_t$ are stationary. Therefore, if the LASSO penalty $J_\lambda(\hat\sigma_j\theta_j) = \lambda\hat\sigma_j|\theta_j|$ satisfies
\begin{equation}
  \lambda \sum_{j=1}^p   \hat\sigma_j|\theta^*_j| - \lambda \sum_{j=1}^p  \hat\sigma_j|\breve\theta_j|\to\infty,\label{eq: penalty difference}
\end{equation}
the criterion function in (\ref{eq: regularized estimation}) satisfies $\mathcal{L}_\lambda(\breve\ttheta) < \mathcal{L}_\lambda(\ttheta^*)$ w.p.a.1, which implies that LASSO prefers the alternative coefficient $\breve\ttheta$ to the true $\ttheta^*$, suggesting its inconsistency. The following Proposition \ref{prop: inconsistent LASSO} shows that coefficients $\breve\pphi^{\rm a}$ and $\breve\pphi^{\rm b}$ satisfying (\ref{eq: Aphi1+phi2}) and (\ref{eq: penalty difference}) exist in a wide class of models, 
indicating that LASSO is generally inconsistent.

For two vectors $\aa$ and $\bb$ of the same dimension, we write $|\aa| \preceq |\bb|$ if $|a_j| \leq |b_j|$ elementwise for all $j$.

\begin{proposition} \label{prop: inconsistent LASSO}
Suppose that Assumptions \ref{assu:tail}--\ref{assu:asym_n} hold, and $J_\lambda(t) = \lambda|t|$ is the LASSO penalty, where the tuning parameter $\lambda$ satisfies
$\sqrt{n}\lambda \to \infty$ as $n\to\infty$. In addition, suppose that $\pphi^*\neq\boldsymbol{0}$, $\max_j|\phi_j^*| = O(1)$, $\min_{j:\phi_j^*\neq 0} |\phi_j^*| \geq c_\phi$ for some absolute constant $c_\phi$, and the sparsity index $s$ is fixed as $n\to\infty$.

  (a) Assume that there exists an alternative coefficient $\breve{\pphi}\in\mathbb{R}^{p_{\rm ab}}$  satisfying \eqref{eq: Aphi1+phi2} and $|\breve{\pphi}| \preceq  c_1|\pphi^{*}|$ for some absolute constant $c_1\in(0,1)$. Then the alternative coefficient $\breve\ttheta = (\ggamma^{*\top}, \breve{\pphi}^{\top}, \bbeta^{*\top})^\top$ satisfies
  \begin{equation}\label{eq: LASSO inconsistency}
    \mathcal{L}_\lambda(\ttheta^*) - \mathcal{L}_\lambda(\breve\ttheta) \stackrel{\mathrm{p}}{\to} \infty.
  \end{equation}
In this case, $\|\hat\ttheta - \ttheta^*\|_\infty \stackrel{\mathrm{p}}{\not\to} 0$.

(b) Let $\odot$ denote the Hadamard (elementwise) product.
If $\pphi^{\rm b*} \odot  (\AA^\top \pphi^{\rm a*}) \preceq \mathbf 0 $ and $(1-c_2)|\pphi^{\rm b*}| \preceq |\AA^\top \pphi^{\rm a*}| \preceq (1+c_2)|\pphi^{\rm b*}|$ with $c_2\in(0,1)$, then
there exists $\breve{\pphi}$ satisfying the conditions in Part (a).

\end{proposition}

Proposition \ref{prop: inconsistent LASSO} illustrates that inconsistency of LASSO is a generic phenomenon in high-dimensional predictive regressions with cointegration, when the underlying cointegrating relationship is unknown. The condition $\sqrt{n}\lambda\to\infty$ is a commonly used condition and necessary for the consistency of LURs $\xx_t$ and stationary regressors $\zz_t$. The lower bound on the nonzero coefficients in $\pphi^*$ and the fixed sparsity index $s$ simplify the technical proofs. If LASSO is inconsistent in this restrictive setting, it cannot go through in more general settings with diverging $s$.

 Part (a) shows that if there exists an alternative  vector $\breve\pphi$ satisfying (\ref{eq: Aphi1+phi2}) and $|\breve\pphi| \preceq c_1 |\pphi^{*}|$, the criterion function evaluated under any consistent estimator of the true $\ttheta^*$ is larger than that evaluated under any estimator converging to the alternative coefficient $\breve\ttheta$. Therefore, LASSO always prefers an estimator converging to the alternative coefficient to any consistent estimator of the truth, implying the inconsistency of LASSO.

 Part (b) shows that such an alternative vector $\breve\pphi$ exists when the two vectors $\pphi^{\rm b*}$ and $\AA^\top \pphi^{\rm a*}$ have opposite signs element-wise. A typical scenario is when $\AA^\top\pphi^{\rm a*} + \pphi^{\rm b*} = \mathbf{0}$, so that $\xx_t^{\rm b}$ in (\ref{eq:equivalent2}) has zero coefficients. In this case, the cointegrated regressors $\xx_t^{\rm a}$ and $\xx_t^{\rm b}$ are predictive of $y_t$ only through the cointegrating residual $\vv_t^{\rm a}$, under which the predictive relationship is effectively driven by stationary variation rather than by the common stochastic trend. Therefore, this scenario accommodates the situation where a stationary outcome is predictable using nonstationary regressors, and our theoretical result indicates that LASSO suffers serious drawbacks in this important setting. The bivariate model in Example \ref{ex: bivariate toy} is a special case of this scenario, where $\AA = 1$ and $\pphi^{\rm b*} = -\pphi^{\rm a*}$.  In addition,
 if LASSO is used to provide the initial estimator in \citet{zou2006adaptive}'s adaptive LASSO, the latter becomes inconsistent as well.

\subsection{Inequivalence of the Constrained and Unconstrained FCPs}\label{subsec: inequivalence}

This section discusses why the constrained and unconstrained formulations of FCPs are generally not equivalent, highlighting that the additional constraint makes a genuine difference to the vanilla FCP formulation.

Recall that we consider two penalized estimators in Section \ref{subsec: regularization}, including the vanilla penalized regression
\begin{equation}\label{eq: regularized estimation app}
\min_{\ttheta}  \mathcal{L}_\lambda(\ttheta),\text{ where }\mathcal{L}_\lambda(\ttheta)=\dfrac{1}{n}\sum_{t=1}^n(y_t - \ww_{t-1}^\top \ttheta)^2 + \sum_{j=1}^p J_\lambda(\hat\sigma_j\theta_j),
\end{equation} 
and the constrained regularization
\begin{equation}\label{eq: constrained regularized estimation app}
\min_{\ttheta} \mathcal{L}_\lambda (\ttheta)
    \quad \text{subject to } \sum_{j=1}^p J_\lambda(\hat\sigma_j\theta_j) \leq \mu.
\end{equation}
The solutions to \eqref{eq: regularized estimation app} and \eqref{eq: constrained regularized estimation app} are denoted by
$\hat\ttheta $ and $\tilde\ttheta $, respectively, as in the main text.

The two formulations are not equivalent in general. For LASSO, by the duality, a constrained solution can be reproduced by an unconstrained LASSO with an appropriately adjusted tuning parameter. For FCPs like SCAD and MCP, however, there is generally no guarantee of a value of the tuning parameter in the unconstrained problem that reproduces a given constrained solution. Details are available in the following discussions.
\subsubsection{LASSO} For LASSO with $J_\lambda(t) = \lambda |t|$,  both the objective and the feasible set of the constrained problem are convex. Moreover, because $\mu>0$, the point $\boldsymbol{\theta}=\boldsymbol{0}$ is an interior point of the feasible set. Let $\eta\geq0$ denote the Lagrange multiplier associated with the constraint. The Lagrangian function is
\begin{align*}
\mathcal{L}_\lambda(\ttheta,\eta)
&=
\dfrac{1}{n}\sum_{t=1}^n(y_t - \ww_{t-1}^\top \ttheta)^2
+\sum_{j=1}^p J_\lambda(\hat\sigma_j\theta_j)
+\eta\left\{
\sum_{j=1}^p J_\lambda(\hat\sigma_j\theta_j)-\mu
\right\}\\
&=
\dfrac{1}{n}\sum_{t=1}^n(y_t - \ww_{t-1}^\top \ttheta)^2
+(1+\eta)\lambda
\sum_{j=1}^p\widehat{\sigma}_j|\theta_j|
-\eta\mu.
\end{align*}
The term $-\eta\mu$ does not depend on $\boldsymbol{\theta}$. Therefore, the constrained solution also minimizes an unconstrained LASSO objective by re-defining a tuning parameter
$
\widetilde{\lambda}=(1+\eta)\lambda.
$
If the constraint is inactive, a zero Lagrange multiplier $\eta = 0$ gives $\widetilde{\lambda}=\lambda$. If the constraint is active, then typically $\eta>0$ and $\widetilde{\lambda}>\lambda$.


\subsubsection{FCPs}

For FCPs like SCAD and MCP, the preceding argument does not establish equivalence. First, the feasible set
$
\left\{
\ttheta:
\sum_{j=1}^p J_\lambda(\hat\sigma_j\theta_j) \leq\mu
\right\}
$
is generally nonconvex. Strong duality therefore may be invalid, and thus the KKT conditions do not necessarily characterize a global minimizer.

Second, even if a constrained solution satisfies the KKT conditions with a multiplier $\eta$, there is no guarantee that it corresponds to an unconstrained solution with some tuning parameter $\widetilde{\lambda}$. Specifically, the stationarity condition of the constrained problem involves the Lagrangian
\begin{align*}
\mathcal{L}_\lambda(\ttheta,\eta) =
\dfrac{1}{n}\sum_{t=1}^n(y_t - \ww_{t-1}^\top \ttheta)^2
+(1+\eta) \sum_{j=1}^p J_\lambda(\hat\sigma_j\theta_j)
-\eta\mu.
\end{align*}
For equivalence between the unconstrained and constrained optimization problems, we need a value of tuning parameter $\widetilde{\lambda}$ such that
\begin{equation}\label{eq: equivalence lambda}
  \sum_{j=1}^p J_{\widetilde{\lambda}}(\hat\sigma_j\theta_j)
=
(1+\eta)\sum_{j=1}^p J_\lambda(\hat\sigma_j\theta_j).
\end{equation}
This equality generally does not hold for all $\boldsymbol{\theta}$ when $\eta>0$.

We use MCP as an example for illustration. The derivative of MCP for $t>0$ is
$
J_\lambda'(t)
=
\left(\lambda-t/a \right)_+
$.
Suppose two distinct standardized active coefficients, $t_j\neq t_k$, lie in the concave region of MCP. For the condition (\ref{eq: equivalence lambda}) for equivalence to hold, we must have
\[
J_{\widetilde{\lambda}}'(t_\ell)
=
(1+\eta)J_{\lambda}'(t_\ell),
\qquad
\ell\in\{j,k\}.
\]
If both coefficients lie in the interior concave regions, this becomes
\[
\widetilde{\lambda}-\frac{t_\ell}{a}
=
(1+\eta)
\left(\lambda-\frac{t_\ell}{a}\right),
\qquad
\ell\in\{j,k\}.
\]
Subtracting the equations for $j$ and $k$ yields
$
-(t_j-t_k) =
-(1+\eta) (t_j-t_k).
$
Since $t_j\neq t_k$, this forces $\eta=0$. Consequently, when the constraint is active and $\eta>0$, (\ref{eq: equivalence lambda}) does not hold, and therefore a single value of $\widetilde{\lambda}$ generally cannot reproduce the constrained solution. The same issue arises for SCAD.

The solutions to the constrained and unconstrained regularization problems can coincide in some special cases.  First,
an unconstrained global minimizer $\hat\ttheta$ satisfying
   $ \sum_{j=1}^p J_\lambda(\hat\sigma_j\hat\theta_j)\leq\mu$
 is also a constrained minimizer.
Second, if the constraint is inactive, the constrained problem reduces locally to the unconstrained problem at the same $\lambda$. In general, there is no guarantee that the solutions will coincide, and thus the constrained problem cannot be simply viewed as reformulating the unconstrained problem. Instead, the additional constraint formulates a genuinely different optimization problem that restricts the parameter space to a bounded penalty that favors sparse solutions, and facilitates the technical proofs for characterizing the global minimizer.

\section{Proofs\protect\label{sec:Proofs}}  
\subsection{Notations}
In the proofs, we use $c$ and $C$, without superscripts or subscripts,
to denote generic positive constants that may vary from place to place.
The abbreviation ``w.p.a.1'' means ``with probability approaching one''.
For notational simplicity, in the proofs we assume $p\geq n^{\nu_{1}}$
for some absolute constant $\nu_{1}$, which is reasonable as we focus
on the high dimensional case with a larger $p$ relative to $n$. 
Identity-matrix subscripts may be omitted when there is no ambiguity with matrix dimensions.

Let $[p] = \mathcal{M}_z \cup \mathcal{C}_{\rm a} \cup \mathcal{C}_{\rm b} \cup \mathcal{M}_x$ be a disjoint union of the index sets of different types of regressors. Specifically, $\mathcal{M}_z$ indexes the stationary regressors, $\mathcal{C}_{\rm a}$ and $\mathcal{C}_{\rm b}$ index the cointegrated regressors $\xx_t^{\rm a}$ and $\xx_t^{\rm b}$, and $\mathcal{M}_x$ indexes the local unit root regressors. Define $\mathcal{C} = \mathcal{C}_{\rm a}\cup \mathcal{C}_{\rm b}$. For any $m$ by $n$ matrix $\CC_{m\times n} = (C_{i,j})_{i\in[m],j\in[n]}$ and any $\mathcal{A}\subset[m]$ and $\mathcal{B}\subset[n]$, let $\CC_{\mathcal{A},\mathcal{B}} = (C_{i,j})_{i\in\mathcal{A},j\in\mathcal{B}}$ denote the submatrix of $\CC$ with row indices in $\mathcal{A}$ and column indices in $\mathcal{B}$. For the same matrix $\CC$, define its induced $\ell_1$-norm as $\|\CC\|_1=\max_{j\in[n]}\sum_{i=1}^m|C_{i,j}|$, its spectral norm as $\|\CC\|_2=\sup_{\vv\in\mathbb{R}^n,\,\|\vv\|_2=1}\|\CC\vv\|_2$, and its maximum entry norm as $\|\CC\|_{\max}=\max_{i\in[m],j\in[n]}|C_{i,j}|$.

Define the outcome vector $\YY = (y_1, y_2, \ldots, y_n)^\top$ and the error term vector $\uu = (u_1, u_2, \ldots, u_n)^\top$.
For the regressors, let
$\zz_t^\sharp = (\zz_t^\top, \vv_{t}^{\rm a\top})^\top$ collect all stationary components in the infeasible regression, and $\xx_t^\sharp = (\xx_t^{\rm b\top},\xx_t^{\top})^\top$ collect all LUR components in the infeasible regression. Therefore,
$\ww_t^\sharp = (\zz_t^{\sharp\top}, \xx_t^{\sharp\top})^\top = \QQ^{-1}\ww_t$
with $\QQ$ defined in (\ref{eq: rotation}).
Let $\WW = (\ww_0, \ww_1, \ldots, \ww_{n-1})^\top$ and $\WW^\sharp = (\ww_0^\sharp, \ww_1^\sharp, \ldots, \ww_{n-1}^\sharp)^\top$ be the $n\times p$ matrices collecting the regressors $\ww_t$ in the feasible regression and $\ww_t^\sharp$ in the infeasible regression, respectively. For any  $\mathcal{A} \subseteq [p]$ with $|\mathcal{A}| < n$, define $\WW_{\mathcal{A}} = (\ww_{\mathcal{A},0},\ww_{\mathcal{A},1},\dots,\ww_{\mathcal{A},n-1})^\top$ and $\WW_{\mathcal{A}}^\sharp = (\ww_{\mathcal{A},0}^\sharp,\ww_{\mathcal{A},1}^\sharp,\dots,\ww_{\mathcal{A},n-1}^\sharp)^\top$.
Furthermore, define the projection matrices and residual matrices as  $\PP_{\mathcal{A}} = \WW_{\mathcal{A}}(\WW_{\mathcal{A}}^\top \WW_{\mathcal{A}})^{-1}\WW_{\mathcal{A}}^\top$
the projector to the linear space spanned by $\WW_{\mathcal{A}}$, and  $\MM_{\mathcal{A}} = \II_{n} - \PP_{\mathcal{A}}$.

\subsection{Technical Lemmas}
The first lemma shows the deviation bound (DB) condition for the cross-product of the regressors and the error term in the high-dimensional setting.
\begin{lemma}
\label{lem:DB} Under Assumptions \ref{assu:tail}-\ref{assu:asym_n},
we have
\begin{equation}
\|n^{-1}\sum_{t=1}^{n}\zz_{t-1}^\sharp u_{t}\|_{\infty}+\|n^{-3/2}\sum_{t=1}^{n}\xx_{t-1}^\sharp u_{t}\|_{\infty}+\|n^{-3/2}\sum_{t=1}^{n}\xx^{\sharp}_{t-1}\zz_{t-1}^{\sharp\top}\|_{\max}+\|n^{-3/2}\sum_{t=1}^{n}\xx^{\rm ab}_{t-1}u_{t}\|_{\infty}\lep\dfrac{(\log p)^{1+\frac{1}{2r}}}{\sqrt{n}},\label{eq:DB u}
\end{equation}
Therefore,
\begin{equation}\label{eq:DB u sharp}
\|n^{-1/2}\RR_n^{-1}\sum_{t=1}^{n}\ww_{t-1}^\sharp u_{t}\|_{\infty} \lep\dfrac{(\log p)^{1+\frac{1}{2r}}}{\sqrt{n}}.
\end{equation}
\end{lemma}
\begin{proof}[Proof of Lemma \ref{lem:DB}] By definition $\ww_t^\sharp = (\zz_t^{\sharp\top}, \xx_t^{\sharp\top})^\top$, and thus (\ref{eq:DB u sharp}) follows directly from the bound of $\|n^{-1}\sum_{t=1}^{n}\zz_{t-1}^\sharp u_{t}\|_{\infty}+\|n^{-3/2}\sum_{t=1}^{n}\xx_{t-1}^\sharp u_{t}\|_{\infty}$ given by (\ref{eq:DB u}). It thus suffices to prove (\ref{eq:DB u}).

The bound of $\|n^{-1}\sum_{t=1}^{n}\zz_{t-1}^\sharp u_{t}\|_{\infty}+\|n^{-3/2}\sum_{t=1}^{n}\xx_{t-1}^\sharp u_{t}\|_{\infty}+\|n^{-3/2}\sum_{t=1}^{n}\xx^{\sharp}_{t-1}\zz_{t-1}^{\sharp\top}\|_{\max}$ follows from the arguments for Lemma A.9 of \citetalias{gao2026lasso}. For $\|n^{-3/2}\sum_{t=1}^{n}\xx^{\rm ab}_{t-1}u_{t}\|_{\infty}$, note that 
\[
\begin{aligned}
\left\|n^{-3/2}\sum_{t=1}^n \xx_{t-1}^{\rm a}u_t\right\|_\infty
&\le
\max_j\|\AA_{j,\cdot}\|_1
\left\|n^{-3/2}\sum_{t=1}^n \xx_{t-1}^{\rm b}u_t\right\|_\infty 
 +
n^{-1/2}
\left\|n^{-1}\sum_{t=1}^n \vv_{t-1}^{\rm a}u_t\right\|_\infty,
\end{aligned}
\]
and $\xx_t^{\rm b}$ is a subvector of $\xx_t^\sharp$. The  bound of $\|n^{-3/2}\sum_{t=1}^n \xx_{t-1}^{\rm b}u_t\|_\infty$ thus follows from that of $\|n^{-1}\sum_{t=1}^{n}\zz_{t-1}^\sharp u_{t}\|_{\infty}+\|n^{-3/2}\sum_{t=1}^{n}\xx_{t-1}^\sharp u_{t}\|_{\infty}$.
\end{proof}

The next Lemma is the Gaussian approximation result for the sum of weakly dependent random variables, which is useful for deriving the asymptotic distribution of high-dimensional test statistics.

\begin{lemma}[\citetalias{mei2024lasso}'s Lemma B.4]
\label{lem:Gaussian Approximation} Under Assumptions \ref{assu:tail}--\ref{assu:asym_n},
there exist mutually independent standard Brownian motions $\{\mathcal{B}_{k}(t)\}$ with
independent increments $\mathcal{B}_{k}(t)-\mathcal{B}_{k}(s)\sim\mathcal{N}(0,t-s)$ for $t\geq s\geq0$ such that
\[
\sup_{k\in[p+1],t\in[n]}\left|\frac{1}{\sqrt{n}}\left(\sum_{s=0}^{t-1}\varepsilon_{k,s}-\mathcal{B}_{k}(t\cdot V_{k}^{*})\right)\right|=O_{p}\left(\dfrac{\log p}{\sqrt{n}}\right),
\]
where $V_{k}^{*}=\mathbb{E}\left[\sum_{d=-\infty}^{\infty}\varepsilon_{k,t}\varepsilon_{k,t-d}\right]$
is the long-run variance of $\{\varepsilon_{k,t}\}.$
\end{lemma}
The next lemma establishes the result that the LUR regressors $\xx_{t}$ and $\xx_t^{\rm b}$
with weakly dependent innovations can be approximated by another
vector of LUR processes with normally distributed innovations. Partition the rows of $\PPhi$ conformably with the stationary innovation vector $\vv_t = (u_t,\zz_t^\top,\vv_t^{\rm a\top},\ee_t^{\rm b\top},\ee_t^\top)^\top$; see Equation~\eqref{eq: v=Phiepsilon} in the main text. Then let $\PPhi^x$ be the $p_x\times (p+1)$ submatrix corresponding to $\ee_t$ and $\PPhi^{\rm b}$ be the $p_{\rm b}\times (p+1)$ submatrix corresponding to $\ee_t^{\rm b}$.  
\begin{lemma}
\label{lem:LUR Gaussian Approx} Suppose that Assumptions
\ref{assu:tail}--\ref{assu:asym_n} hold. There exist independent,
normally distributed vectors $\eeta_t \overset{\rm i.i.d}{\sim} \mathcal{N}({\bf 0}_{p+1},{\rm diag}(V_k^*)_{k\in[p+1]})$
such that the LUR processes $\xxi_{t} = \left(\bfI + \frac{\CC}{n}\right)\xxi_{t-1}+\PPhi^x\eeta_t$ and $\xxi^{\rm b}_{t} = \bigl(\bfI + \frac{\CC^{\mathrm{b}}}{n}\bigr)\xxi^{\rm b}_{t-1}+\PPhi^{\rm b}\eeta_t$ satisfy
\[
\sup_{t\in[n]}\left(\|\xx_t - \xxi_t\|_\infty + \|\xx_t^{\rm b} - \xxi_t^{\rm b}\|_\infty \right) = O_{p}(\log p).
\]
\end{lemma}
\begin{proof}[Proof of Lemma \ref{lem:LUR Gaussian Approx}]We focus on $\xx_t$, as the approximation of $\xx_t^{\rm b}$ is similar. Let $\rho_j = 1 + \frac{c_j}{n}$ denote the LUR coefficient of $x_{j,t}$. For simplicity, we assume zero initial values. Then we can decompose the approximation error as 
\[x_{j,t} - \xi_{j,t} =  \PPhi_{j,\cdot}^x\sum_{\ell=1}^t(\vareepsilon_{\ell} - \eeta_{\ell}) + (\rho_j - 1)\sum_{\ell = 1}^{t-1} \rho_j^{t-1-\ell} \PPhi_{j,\cdot}^x \sum_{s=1}^\ell (\vareepsilon_{s} - \eeta_{s}).\]
By Assumption \ref{assu:covMat} we have $\max_{j}\|\PPhi_{j,\cdot}^x\|_1 = O(1)$. By Assumption \ref{assu:cointegration}(d) that bounds $\max_j|c_j|$, we have $\max_{j}|\rho_j| = O(1)$, $\max_{j,t}|\rho_j|^t = O(1),\text{ and }$
\[\max_{j,t} \left|(\rho_j - 1)\sum_{\ell = 1}^{t-1}  \rho_j ^{t-1-\ell}\right| =\max_{j,t} \left| \rho_j^{t-1}-1 \right| = O(1).\]
Therefore,
\[\begin{aligned}
  \max_{j,t}|x_{j,t} - \xi_{j,t}| &\leq  \left(1 + \max_{j,t}\left|(\rho_j - 1)\sum_{\ell = 1}^{t-1}  \rho_j ^{t-1-\ell}\right|\right)\cdot \max_{j }\|\PPhi_{j,\cdot}^x\|_1\cdot\max_{t\leq n}\left\|\sum_{\ell=1}^t(\vareepsilon_{\ell} - \eeta_{\ell})\right\|_\infty\\
  &=O_p(\log p),
\end{aligned} \]
where the last step applies Lemma \ref{lem:Gaussian Approximation}. This concludes the proof of Lemma \ref{lem:LUR Gaussian Approx}.
\end{proof}

For any $p\times p$ matrix $\SSigma$, the restricted eigenvalue is defined as
\[\kappa(\SSigma,L,s) = \min_{\substack{\mathcal{J}\subseteq[p],\\|\mathcal{J}|\leq s}}\, \, \min_{\substack{\vv\neq 0,\\\|\vv_{\mathcal{J}^c}\|_1\leq L\|\vv_{\mathcal{J}}\|_1}}\frac{\vv^\top \SSigma \vv}{\|\vv\|_2^2}.\]
This is an important concept in high-dimensional statistics for the analysis of Lasso-type estimators \citep{bickel2009simultaneous},
which characterizes the minimum eigenvalue of the covariance matrix restricted to sparse vectors.

Define the sample Gram matrices of the LURs and stationary regressors
as $\hat{\SSigma}^{(\xx^\sharp)}=\frac{1}{n}\sum_{t=1}^{n}\xx^\sharp_{t-1}\xx_{t-1}^{\sharp\top}$
and $\hat{\SSigma}^{(\zz^\sharp)}=\frac{1}{n}\sum_{t=1}^{n}\zz^\sharp_{t-1}\zz_{t-1}^{\sharp\top}$,
respectively.
The following Lemma establishes the RE condition for LURs without
standardization. We focus on the demeaned regressors $\xx^\sharp_{t-1}-\bar{\xx}^\sharp$
where $\bar{\xx}^\sharp=n^{-1}\sum_{t=1}^{n}\xx^\sharp_{t-1}$, since the results will
be helpful to bound the standard deviations. For a generic vector of time series $\xxi_t$, define $\bar\xxi = n^{-1}\sum_{t=1}^{n}\xxi_{t-1}$,  $\ddot\xxi_{t-1} = \xxi_{t-1}-\bar{\xxi}$, and $\ddot{\SSigma}^{(\xxi)}=n^{-1}\sum_{t=1}^{n}\ddot\xxi_{t-1}\ddot\xxi_{t-1}^{\top}$.
\begin{lemma}
\label{lem:RE LUR generic L}Suppose that $s=o(n\wedge p)$
as $n\to\infty$. 
Then under Assumptions \ref{assu:tail}-\ref{assu:asym_n}, for any fixed $L\geq1$ the inequality
\begin{equation}
\frac{\kappa(\ddot{\SSigma}^{(\xx^\sharp)},L,s)}{n}\geq\dfrac{c_{0,\kappa}^{(\xx^\sharp)}}{L^{2}s\log p}\label{eq:RE-unit-1-1}
\end{equation}
holds w.p.a.1.
\end{lemma}
\begin{proof}[Proof of Lemma \ref{lem:RE LUR generic L}]When the innovations generating the LURs are i.i.d normally distributed, the proof follows Part (a) of the proof of Lemma A.4 in \citetalias{gao2026lasso}. When the innovations are weakly dependent without the normal assumption, with Lemma \ref{lem:LUR Gaussian Approx} the proof follows Part (b) of the proof of Lemma A.4 in \citetalias{gao2026lasso}.
\end{proof}
The following Lemma establishes the RE condition for LUR regressors without demeaning.
\begin{lemma}
\label{lem:RE LUR-1} Under Assumptions \ref{assu:tail}-\ref{assu:asym_n},
w.p.a.1
\begin{equation}
\kappa(\hat{\SSigma}^{(\xx^\sharp)},3,s)\geq\dfrac{c_{\kappa}^{(\xx^\sharp)}n}{s\cdot\log p}.\label{eq:RE I LUR}
\end{equation}
\end{lemma}
\begin{proof}[Proof of Lemma \ref{lem:RE LUR-1}]
 This lemma is a direct corollary of Lemma \ref{lem:RE LUR generic L}
by taking $L=3$ and $c_{\kappa}^{(\xx^\sharp)}=c_{0,\kappa}^{(\xx^\sharp)}/9$, and the fact
that $\kappa(\hat{\SSigma}^{(\xx^\sharp)},3,s)\geq\kappa(\ddot{\SSigma}^{(\xx^\sharp)},3,s)$.
\end{proof}
The following Lemma establishes the RE condition for stationary regressors. Define $\hat{\SSigma}^{(\zz^\sharp)}=n^{-1}\sum_{t=1}^{n}\zz^\sharp_{t-1} \zz_{t-1}^{\sharp\top}$.
\begin{lemma}
\label{lem:RE k} Under Assumptions \ref{assu:tail}-\ref{assu:asym_n}, there exists some absolute constant $c_{\kappa}^{(z)}$ such that
\begin{equation}
\kappa(\hat{\SSigma}^{(\zz^\sharp)},3,s)\geq c_{\kappa}^{(z)}\label{eq:RE ST MI}
\end{equation}
holds w.p.a.1.
\end{lemma}
\begin{proof}[Proof of Lemma \ref{lem:RE k}]
 The proof follows standard arguments using concentration inequalities
for weakly dependent time series, like (B.30) in \citetalias{mei2024lasso}.
\end{proof}

The following Lemma establishes the RE condition for the combined set of stationary and LUR regressors in the infeasible regression (\ref{eq:equivalent2}). Define $\hat\SSigma^\sharp=\RR_n^{-1}\sum_{t=1}^{n}\ww^\sharp_{t-1} \ww_{t-1}^{\sharp\top}\RR_n^{-1} = \RR_n^{-1}\WW^{\sharp\top}\WW^\sharp\RR_n^{-1}$.
\begin{lemma}
\label{lem:RE xz} Under Assumptions \ref{assu:tail}-\ref{assu:asym_n},  there exists some absolute constant $c_{\kappa}$ such that
\begin{equation}
\kappa(\hat{\SSigma}^{\sharp},3,s)\geq \frac{c_{\kappa}}{s\log p}\label{eq:RE xz}
\end{equation}
holds w.p.a.1.
\end{lemma}
\begin{proof}[Proof of Lemma \ref{lem:RE xz}] Define $\ddot{\SSigma}^\sharp = \RR_n^{-1}\sum_{t=1}^{n}(\ww^\sharp_{t-1} - \bar{\ww}^\sharp)(\ww_{t-1}^{\sharp} - \bar{\ww}^\sharp)^\top\RR_n^{-1}$.
 The same arguments in the proof of Proposition B.6 in \citetalias{mei2024lasso} apply to showing that
 \[\kappa(\ddot{\SSigma}^{\sharp},3,s)\geq \frac{c_{\kappa}}{s\log p},\]
 and Lemma \ref{lem:RE xz} follows from the fact that $\kappa(\hat{\SSigma}^{\sharp},3,s) \geq \kappa(\ddot{\SSigma}^{\sharp},3,s)$.
\end{proof}

The following lemma delivers the bounds of the standard deviations
used for scaling the penalties.
\begin{lemma}
\label{lem:sd bound} Under Assumptions \ref{assu:tail}-\ref{assu:asym_n},
\begin{itemize}
\item[(a)] For stationary regressors, there exist some absolute constants $\sigma_{\min}<\sigma_{\max}$,
such that w.p.a.1
\begin{equation}
\sigma_{\min}\leq\min_{j\in\mathcal{M}_{z}}\hat{\sigma}_{j}\leq\max_{j\in\mathcal{M}_{z}}\hat{\sigma}_{j}\leq\sigma_{\max}.\label{eq:sd bound ST MI}
\end{equation}
\item[(b)] For LUR and cointegrated regressors,
\begin{equation}
\sqrt{\dfrac{n}{\log p}}\lep\min_{j\in\mathcal{M}_{x}\cup\mathcal{C}}\hat{\sigma}_{j}\leq\max_{j\in\mathcal{M}_{x}\cup\mathcal{C}}\hat{\sigma}_{j}\lep\sqrt{n\log p}.\label{eq:sd bound LUR}
\end{equation}
\end{itemize}
\end{lemma}
 \begin{proof}[Proof of Lemma \ref{lem:sd bound}] The bounds for $j \in \mathcal{M}_z$ and $j \in \mathcal{M}_x\cup\mathcal{C}_{\rm b}$ follow those in Lemma A.7 of \citetalias{gao2026lasso}. For $j\in\mathcal{C}_{\rm a}$, by (A.26) of \citetalias{gao2026lasso} we have
\begin{equation}\label{eq: x sharp upper bound}
\sup_{t\in[n]}\|\xx_{t-1}^\sharp\|_\infty = O_p(\sqrt{n\log p}),
\end{equation} 
and therefore
\begin{align}
  \sup_{t\in[n]}\|\xx^{\rm ab}_{t-1}\|_\infty &\leq \sup_{t\in[n]}\|\xx_{t-1}^{\mathrm{a}}\|_\infty + \sup_{t\in[n]}\|\xx_{t-1}^{\mathrm{b}}\|_\infty \leq \sup_{t\in[n]}\|\AA\xx_{t-1}^{\mathrm{b}}+\vv_{t-1}^{\rm a}\|_\infty + \sup_{t\in[n]}\|\xx_{t-1}^{\sharp}\|_\infty \nonumber \\
  &\leq \left(1+ \max_j \|\AA_{j\cdot}\|_1\right)\cdot \sup_{t\in[n]}\|\xx_{t-1}^\sharp\|_\infty + \sup_{t\in[n]}\|\vv_{t-1}^{\rm a}\|_\infty= O_p(\sqrt{n\log p}). \label{eq: xab upper bound}
\end{align} 
The upper bound then follows from the inequality
  \[\max_{j\in\mathcal{C}_{\rm a}}\hat{\sigma}_{j} \leq \sqrt{n\max_{t\in [n]}\|\xx_{t-1}^{\rm ab}\|_\infty^2/n} = O_p(\sqrt{n\log p}),\]
  where the last step applies the upper bound in \eqref{eq: xab upper bound}. As for the lower bound, let $\AA_{j,\cdot}$ denote the $j$-th row of $\AA$. Then for $j\in\mathcal{C}_{\rm a}$, we have
  \begin{equation}\label{eq:sd lower bound Ca}
    \hat{\sigma}_{j} = \sqrt{\frac{1}{n}\sum_{t=1}^n (\ddot x_{j,t-1}^{\rm a})^2} \geq \sqrt{\frac{1}{n}\sum_{t=1}^n (\AA_{j,\cdot} \ddot\xx_{t-1}^{\rm b})^2} - \sqrt{\frac{1}{n}\sum_{t=1}^n (\ddot v_{j,t-1}^{\rm a})^2}.
  \end{equation}
  Since $v_{j,t}^{\rm a}$ is stationary, it follows that $\max_{j\in\mathcal{C}_{\rm a}}\frac{1}{n}\sum_{t=1}^n (\ddot v_{j,t-1}^{\rm a})^2 = O_p(1)$. Moreover, for each $j$, choose any $k_j\in\arg\max_k |A_{j,k}|$ and let $J_0 = \{k_j\}$. Therefore, $|J_0| = 1$ and
  \[ \max_{j\in\mathcal{C}_{\rm a}}\frac{\|\AA_{j,J_0^c}\|_1}{\|\AA_{j,J_0}\|_1}  \leq \max_{j\in\mathcal{C}_{\rm a}}\frac{\|\AA_{j,\cdot}\|_1}{\|\AA_{j,\cdot}\|_\infty} \leq \max_{j\in\mathcal{C}_{\rm a}}\frac{\|\AA_{j,\cdot}\|_1^2}{\|\AA_{j,\cdot}\|_2^2} = O(1),\]
  where the second step applies the inequality $\|\AA_{j,\cdot}\|_2^2 \leq \|\AA_{j,\cdot}\|_\infty \cdot \|\AA_{j,\cdot}\|_1$, and the last step applies Assumption \ref{assu:cointegration}. Therefore, for any $j\in\mathcal{C}_{\rm a}$, the vector $\AA_{j,\cdot}$ satisfies
$\|\AA_{j,J_0^c}\|_1 \leq L \|\AA_{j,J_0}\|_1$
  for some fixed $L>1$, with $|J_0| = 1$. Using Lemma \ref{lem:RE LUR generic L},
  \[\min_{j\in\mathcal{C}_{\rm a}}{\frac{1}{n}\sum_{t=1}^n (\AA_{j,\cdot} \ddot\xx_{t-1}^{\rm b})^2} \geq \min_{j\in\mathcal{C}_{\rm a}}\|\AA_{j,J_0}\|_2^2  \cdot \kappa(\ddot\SSigma^{(\xx^\sharp)},L,1) \geq \frac{cn}{\log p}.\]
  Together with \eqref{eq:sd lower bound Ca}, we obtain the desired lower bound for $\min\limits_{j\in\mathcal{C}_{\rm a}}\hat{\sigma}_{j} \geq \sqrt{cn/\log p}$ w.p.a.1. We conclude the proof of Lemma \ref{lem:sd bound}.
 \end{proof}
The last lemma characterizes the FCP function $J_\lambda(\cdot)$.
\begin{lemma}\label{lem: FCP}
  Under Condition \ref{cond:penalty}, the penalty function $J_\lambda(\cdot)$ satisfies
  \begin{equation}\label{eq: FCP lower derivative}
    J_\lambda(t) \geq a_1\lambda |t|\text{ for any } t \in [-a\lambda, a\lambda],
  \end{equation}
  and
  \begin{equation}\label{eq: FCP bound lambda2}
    J_\lambda(t) \leq C\lambda^2 \text{ for any } t \in \mathbb{R}.
  \end{equation}
\end{lemma}
\begin{proof}[Proof of Lemma \ref{lem: FCP}] It suffices to show the bounds for $t\geq 0$; the case for $t<0$ follows similarly by considering $J_\lambda(-t)$.

For (\ref{eq: FCP lower derivative}), by the mean value theorem and $J_\lambda(0) = 0$, for any $t \in [0,a\lambda] $ we have \[J_\lambda(t) = J_\lambda(0) + J'_\lambda(\xi) t = J'_\lambda(\xi)t\] for some $\xi \in (0,t)$. Under Condition \ref{cond:penalty}, we have $J'_\lambda(\xi) \geq a_1 \lambda$ for $\xi \in [0, a\lambda]$, which immediately gives \eqref{eq: FCP lower derivative} for $0\leq t \leq a\lambda$.

For \eqref{eq: FCP bound lambda2}, as $J_\lambda(t)$ is concave we have $\max_{a\lambda \leq t \leq b\lambda  }J'_\lambda(t) \leq \max_{0\leq t \leq a\lambda }J'_\lambda(t)$, and therefore
\[ J'_\lambda(t) \leq a_2\lambda\text{ for all } 0< t\leq b\lambda.\]
Applying the mean value theorem again, we have
\[J_\lambda(t) = J_\lambda(0) + J'_\lambda(\xi) t \leq a_2\lambda t\leq a_2 b \lambda^2 \text{ for any } t \in [0, b\lambda].\]
In addition, since $J'_\lambda(t) = 0$ when $t > b\lambda$, we have
\begin{equation}\label{eq: J t flat}
  J_\lambda(t) = J_\lambda(b\lambda) \leq a_2 b \lambda^2\text{ for any $t > b\lambda$. }
\end{equation}  This completes the proof of \eqref{eq: FCP bound lambda2}.
\end{proof}

\subsection{Proofs of Main Results}

\subsubsection{Proof of Proposition \ref{prop: OLS}}
Define
\begin{equation}\label{eq: def T set}
  \mathcal{T} = \{j\in[p]:\QQ_{\mathcal{S},j} \neq \boldsymbol{0}_{|\mathcal{S}|}\}.
\end{equation}  By Assumption \ref{assu:cointegration}(b), $|\mathcal{T}| = O(s)$. Furthermore, define $\dd =  \RR_n\QQ^\top(\hat\ttheta^{\rm or} - \ttheta^*)$. Then $\{j\in[p]:d_j \neq 0\} \subset \mathcal{T}$, and thus $\|\dd\|_0 \leq |\mathcal{T}| = O(s)$.

We first show that
\begin{equation}\label{eq: quadratic equality}
  \dd^\top \hat\SSigma^\sharp \dd = \dd^\top \RR_n^{-1}\WW^{\sharp\top}\uu,
\end{equation}
where $\hat\SSigma^\sharp = \RR_n^{-1}\WW^{\sharp\top} \WW^\sharp \RR_n^{-1}$.
We can deduce that
\begin{equation}\label{eq: deduce equality 1}
  \begin{aligned}
   \dd^\top \hat\SSigma^\sharp \dd &= (\hat\ttheta^{\rm or} - \ttheta^*)^\top \QQ \WW^{\sharp\top} \WW^\sharp \QQ^\top (\hat\ttheta^{\rm or} - \ttheta^*) \\
   &= (\hat\ttheta^{\rm or} - \ttheta^*)^\top \WW^{\top} \WW (\hat\ttheta^{\rm or} - \ttheta^*) \\
    &= (\hat\ttheta^{\rm or} - \ttheta^*)_{\mathcal{S}}^\top \WW_{\mathcal{S}}^{\top} \WW (\hat\ttheta^{\rm or} - \ttheta^*), \\
  \end{aligned}
\end{equation}
where the second row applies $\WW^\sharp \QQ^\top = \WW$, and the third row follows from the fact that $(\hat\ttheta^{\rm or} - \ttheta^*)_{\mathcal{S}^c} = \boldsymbol{0}$. To step further, note that $(\hat\ttheta^{\rm or} - \ttheta^*)_{\mathcal{S}^c} = \boldsymbol{0}$, and thus
\[\begin{aligned}
  \WW_{\mathcal{S}}^\top(\uu - \WW(\hat\ttheta^{\rm or} - \ttheta^*))  &= \WW_{\mathcal{S}}^\top\uu - \WW_{\mathcal{S}}^\top \WW_{\mathcal{S}}(\hat\ttheta^{\rm or} - \ttheta^*)_{\mathcal{S}} \\
  &=  \WW_{\mathcal{S}}^\top\uu - \WW_{\mathcal{S}}^\top \WW_{\mathcal{S}}(\WW_{\mathcal{S}}^\top \WW_{\mathcal{S}})^{-1} \WW_{\mathcal{S}}^\top\uu = \boldsymbol{0}_{|\mathcal{S}|},
\end{aligned}\]
which implies
\begin{equation}\label{eq: deduce equality 2}
    \WW_{\mathcal{S}}^\top\WW(\hat\ttheta^{\rm or} - \ttheta^*) = \WW_{\mathcal{S}}^\top\uu.
\end{equation}
Together with \eqref{eq: deduce equality 1}, we have
\[\begin{aligned}
   \dd^\top \hat\SSigma^\sharp \dd &= (\hat\ttheta^{\rm or} - \ttheta^*)_{\mathcal{S}}^\top \WW_{\mathcal{S}}^{\top} \uu \\
     &= (\hat\ttheta^{\rm or} - \ttheta^*)^\top \WW^{\top} \uu \\
     &= \left[(\hat\ttheta^{\rm or} - \ttheta^*)^\top \QQ \RR_n\right]\RR_n^{-1}\QQ^{-1}\WW^{\top} \uu = \dd^\top \RR_n^{-1}\WW^{\sharp\top}\uu,\\
\end{aligned}\]
where the first row follows from \eqref{eq: deduce equality 2}, the second row uses the fact that $(\hat\ttheta^{\rm or} - \ttheta^*)_{\mathcal{S}^c} = \boldsymbol{0}$, and the last step applies $\WW^\sharp \QQ^\top = \WW$. Then \eqref{eq: quadratic equality} is proved, which implies
\begin{equation}\label{eq:bound quadratic}
  \dd^\top \hat\SSigma^\sharp \dd \leq \|\dd\|_1 \|\RR_n^{-1}\WW^{\sharp\top}\uu\|_\infty = O_p\left(\sqrt{s}\|\dd\|_2\cdot(\log p)^{1+1/r}\right),
\end{equation}
where the last step applies $\|\dd\|_1 \leq \sqrt{|\mathcal{T}|}\cdot \|\dd\|_2 = O(\sqrt{s}\|\dd\|_2)$ and the DB condition in Lemma \ref{lem:DB}. In addition, since $\|\dd\|_0 = O(s)$, by the RE condition in Lemma \ref{lem:RE xz} we have $
\dd^\top \hat\SSigma^\sharp \dd \geq c \|\dd\|_2^2 / (s\log p)$ for some absolute constant $c$. Together with \eqref{eq:bound quadratic}, this implies that
\[
\|\dd\|_2^2 = O_p\left( s\log p \cdot \|\dd\|_2\cdot\sqrt{s}(\log p)^{1+1/r} \right),
\]
and thus
\[
\|\RR_n\QQ^\top(\hat\ttheta^{\rm or} - \ttheta^*)\|_2 = \|\dd\|_2 = O_p\left( s^{3/2}(\log p)^{2+1/r}\right).
\]
By the fact that $\|\dd\|_0 = O(s)$, we also have
\[
\|\RR_n\QQ^\top(\hat\ttheta^{\rm or} - \ttheta^*)\|_1 = O\left(\sqrt{s}\|\dd\|_2\right) = O_p\left( s^{2}(\log p)^{2+1/r} \right).
\]
This completes the proof of (\ref{eq: OLS1}). 

For (\ref{eq: OLS2}) when $q=2$, we first prove the following inequality: 
\begin{equation}\label{eq: dd_subset_bound}
  \|\dd_{\mathcal{M}_z \cup \mathcal{C}_{\rm a}}\|_2^2 \leq C \dd^\top \hat\SSigma^\sharp \dd 
\end{equation}
w.p.a.1. Note that 
\[\dd^\top \hat\SSigma^\sharp \dd = \dd_{\mathcal{M}_z \cup \mathcal{C}_{\rm a}}^\top \hat\SSigma^{(\zz^{\sharp})}  \dd_{\mathcal{M}_z \cup \mathcal{C}_{\rm a}} + \frac{1}{n}\dd_{(\mathcal{M}_z \cup \mathcal{C}_{\rm a})^c}^\top \hat\SSigma^{(\xx^{\sharp})}   \dd_{(\mathcal{M}_z \cup \mathcal{C}_{\rm a})^c} + 2  \dd_{\mathcal{M}_z \cup \mathcal{C}_{\rm a}}^\top \frac{1}{n^{3/2}}\sum_{t=1}^n \zz_{t-1}^{\sharp}\xx_{t-1}^{\sharp\top} \dd_{(\mathcal{M}_z \cup \mathcal{C}_{\rm a})^c}.\] 
Given that $\zz_t^\sharp$ is stationary, subexponential and $\alpha$-mixing with exponentially decaying mixing coefficients, applying standard concentration inequalities yields $\|\hat\SSigma^{(\zz^\sharp)} - \SSigma^{(\zz^\sharp)}\|_{\max} \lep \sqrt{\log p/n}$  with $\SSigma^{(\zz^\sharp)} = \mathbb{E}[\zz_t^\sharp \zz_t^{\sharp\top}]$, and thus 
\[\begin{aligned}
  \dd_{\mathcal{M}_z \cup \mathcal{C}_{\rm a}}^\top \hat\SSigma^{(\zz^{\sharp})}  \dd_{\mathcal{M}_z \cup \mathcal{C}_{\rm a}} &\geq  \dd_{\mathcal{M}_z \cup \mathcal{C}_{\rm a}}^\top  \SSigma^{(\zz^{\sharp})}  \dd_{\mathcal{M}_z \cup \mathcal{C}_{\rm a}} - \|\dd_{\mathcal{M}_z \cup \mathcal{C}_{\rm a}}\|_1^2\cdot \|\hat\SSigma^{(\zz^\sharp)} - \SSigma^{(\zz^\sharp)}\|_{\max} \\
  &\geq c \|\dd_{\mathcal{M}_z \cup \mathcal{C}_{\rm a}}\|_2^2 - \|\dd\|_1^2\cdot C \sqrt{\log p/n}  
\end{aligned}\]
w.p.a.1. Therefore, 
\[\begin{aligned}
  \dd^\top \hat\SSigma^\sharp \dd &\geq c \|\dd_{\mathcal{M}_z \cup \mathcal{C}_{\rm a}}\|_2^2 - \|\dd\|_1^2\cdot \left(\left\|n^{-3/2}\sum_{t=1}^n \zz_{t-1}^{\sharp}\xx_{t-1}^{\sharp\top}\right\|_{\max}+C\sqrt{\log p/n}\right) \\
  &\geq c \|\dd_{\mathcal{M}_z \cup \mathcal{C}_{\rm a}}\|_2^2 - O_p(\|\dd\|_2^2 \cdot s(\log p)^{1+\frac{1}{2r}} / \sqrt{n} )\\
  &\geq c \|\dd_{\mathcal{M}_z \cup \mathcal{C}_{\rm a}}\|_2^2 - O_p( \dd^\top \hat\SSigma^\sharp \dd \cdot s^2(\log p)^{2+\frac{1}{2r}} / \sqrt{n} ) \\
\end{aligned}\]
w.p.a.1, where the second row applies $\|\dd\|_1^2 = O(s\|\dd\|_2^2)$ and Lemma \ref{lem:DB}, and the last row applies the RE condition $\dd^\top \hat\SSigma^\sharp\dd \geq c\|\dd\|_2^2/(s\log p)$ w.p.a.1 by Lemma \ref{lem:RE xz}. Given that $s^2(\log p)^{2+\frac{1}{2r}} / \sqrt{n} \to 0$, we have 
\[
   c\|\dd_{\mathcal{M}_z \cup \mathcal{C}_{\rm a}}\|_2^2 \leq C \dd^\top \hat\SSigma^\sharp \dd
\]
w.p.a.1, and therefore (\ref{eq: dd_subset_bound}) holds. 

Also, (\ref{eq:bound quadratic}) together with the RE condition $\|\dd\|_2 \lep \sqrt{s\log p\cdot\dd^\top \hat\SSigma^\sharp \dd}$ implies that 
\[\dd^\top \hat\SSigma^\sharp\dd  \lep  {s^2 (\log p)^{3+\frac{2}{r}}},\]
which together with (\ref{eq: dd_subset_bound}) implies that 
\[\|\dd_{\mathcal{M}_z \cup \mathcal{C}_{\rm a}}\|_2^2 \lep s^2 (\log p)^{3+\frac{2}{r}}.\]
Therefore, 
\begin{equation}\label{eq:or_rate 1}
  \|\hat\ggamma^{\rm or} - \ggamma^*\|_2 +  \|\hat\pphi^{\rm or,a} - \pphi^{\rm a*}\|_2 \leq \frac{2\|\dd_{\mathcal{M}_z \cup \mathcal{C}_{\rm a}}\|_2}{\sqrt{n}} \lep \frac{s (\log p)^{3/2+\frac{1}{r}}}{\sqrt{n}}.
\end{equation} 
By definition, $ 
\hat\pphi^{\rm or,b}-\pphi^{\rm b*}
=
\frac{\dd_{\mathcal C_{\rm b}}}{n}
-\AA^\top(\hat\pphi^{\rm or,a}-\pphi^{\rm a*}).$ By the fact that the oracle estimation error is zero outside \(\mathcal S\), we have 
\[
\begin{aligned}
\|\AA^\top(\hat\pphi^{\rm or,a}-\pphi^{\rm a*})\|_2
&\leq
\sum_{i\in\mathcal S\cap\mathcal C_{\rm a}} \|\AA_{i,\cdot}\|_2\,|\hat\pphi_i^{\rm or,a}-\pphi_i^{\rm a*}| \leq
\max_i\|\AA_{i,\cdot}\|_1
\sqrt{s}\,
\|\hat\pphi^{\rm or,a}-\pphi^{\rm a*}\|_2\\
&\lep  \frac{s^{3/2}(\log p)^{3/2+1/r}}{\sqrt n}.
\end{aligned} 
\] 
Then 
\begin{equation}\label{eq:or_rate 2}
  \|\hat\pphi^{\rm or,b} - \pphi^{\rm b*}\|_2 \lep  \frac{s^{3/2}(\log p)^{3/2+1/r}}{\sqrt n}.
\end{equation} 
Finally, (\ref{eq: OLS1}) under $q=2$ implies that 
\begin{equation}
  \sqrt{n}\|\hat\bbeta^{\rm or} - \bbeta^*\|_2 \lep \frac{s^{3/2} (\log p)^{2+\frac{1}{r}}}{\sqrt{n}}.
\end{equation}
Then (\ref{eq: OLS2}) is verified when $q=2$ by (\ref{eq:or_rate 1}) and (\ref{eq:or_rate 2}).

For $q=1$, the $\ell_1$-norm convergence rate in (\ref{eq: OLS1}) implies
\[\begin{aligned}
  &\|\hat\ggamma^{\rm or}-\ggamma^*\|_1 + \|\hat\pphi^{\rm or,a}-\pphi^{\rm a*}\|_1 \\
  &+ \sqrt{n}\|\AA^\top\hat\pphi^{\rm or,a}+\hat\pphi^{\rm or,b}-\AA^\top\pphi^{\rm a*}-\pphi^{\rm b*}\|_1 +\sqrt{n}\|\hat\bbeta^{\rm or} - \bbeta^*\|_1 \\
  \lep\, &\frac{s^{2}(\log p)^{2+\frac{1}{r}}}{\sqrt{n}}.
\end{aligned}
\]
It also implies that
\[\begin{aligned}\|\hat\pphi^{\rm or,b} - \pphi^{\rm b*}\|_1 &\leq \|\AA^\top\|_1 \|\hat\pphi^{\rm or,a}-\pphi^{\rm a*}\|_1 + \frac{s^{2}(\log p)^{2+\frac{1}{r}}}{n}  \\
&\lep \frac{s^{2}(\log p)^{2+\frac{1}{r}}}{\sqrt{n}}.\end{aligned}\]
In terms of prediction consistency, note that
\[\mathbb{E}\left[\left(\mathbb{E} [ u_{n+1} |\mathcal{F}_n ]\right)^2\right] \leq \mathbb{E}\left[\mathbb{E}\left[ u_{n+1}^2 |\mathcal{F}_n\right]\right] = \mathbb{E}[u_{n+1}^2] = O(1),\]
where the last step applies the fact that  $u_{n+1} = \sum_{\ell = 1}^{p + 1} \Phi_{1,\ell} \varepsilon_{\ell,n+1}$ is sub-exponential.  Therefore, by the Markov inequality we have
\begin{equation}\label{eq:conditional-mean-u}
  \mathbb{E}(u_{n+1}|\mathcal{F}_n) = O_p(1).
\end{equation}
In addition, note that
\begin{align*}
  &\left(y_{n+1} - \ww_{n}^\top\hat\ttheta^{\rm or}\right)^2 - u_{n+1}^2 \\
  ={}& \left(\ww_{n}^\top(\ttheta^* - \hat \ttheta^{\rm or})\right)^2 + 2u_{n+1} \ww_{n}^\top(\ttheta^* - \hat \ttheta^{\rm or}) \\
  ={}& \left((\RR_n^{-1}\ww_{n}^{\sharp})^\top \RR_n\QQ^\top(\hat\ttheta^{\rm or} - \ttheta^*) \right)^2- 2u_{n+1}(\RR_n^{-1}\ww_{n}^{\sharp})^\top\RR_n\QQ^\top(\hat\ttheta^{\rm or} - \ttheta^*).
\end{align*}
We first bound $\RR_n^{-1}\ww_n^{\sharp}$. Note that $\zz_n^\sharp$ is a vector of sub-exponential random variables, and thus $\|\zz_n^\sharp\|_\infty = O_p(\log p)$. Together with the bound of (\ref{eq: x sharp upper bound}), we have 
\begin{equation}\label{eq: R w bound}
\|\RR_n^{-1}\ww_n^{\sharp}\|_\infty = O_p\left(\frac{\log p}{\sqrt{n}}\right).
\end{equation} 
Therefore,
\begin{align*}
  &\left|\mathbb{E}\left[  (y_{n+1} - \ww_{n}^\top\hat\ttheta^{\rm or})^2 - u_{n+1}^2 \, \big |\, \mathcal{F}_n\right]\right|\\
  \leq{}& \|\RR_n^{-1} \ww_n^{\sharp}\|_\infty^2 \| \RR_n\QQ^\top(\hat\ttheta^{\rm or} - \ttheta^*)\|_1^2  + 2|\mathbb{E}(u_{n+1}|\mathcal{F}_n)|\cdot \|\RR_n^{-1} \ww_n^{\sharp}\|_\infty \| \RR_n\QQ^\top(\hat\ttheta^{\rm or} - \ttheta^*)\|_1 \\
  ={}&O_p\left( \dfrac{(\log p)^2}{n} \right) \cdot O_p\left[\left( s^2(\log p)^{2+\frac{1}{r}}\right)^2\right] + O_p\left(1\right)\cdot O_p\left(\frac{\log p}{\sqrt{n}}\right) \cdot O_p\left(s^2(\log p)^{2+\frac{1}{r}}\right) \\
  ={}& o_p(1),
\end{align*}
where the third row applies (\ref{eq:conditional-mean-u}), (\ref{eq: R w bound}), and (\ref{eq: OLS1}) with $q=1$.
This completes the proof of Proposition \ref{prop: OLS}.

\subsubsection{Proof of Theorem \ref{thm: global}}

Define $\mathcal{B}=\{j\in\mathcal{S}^{c}:\hat\sigma_j|\tilde{\theta}_{j}|>a\lambda\}$.
When $\hat\sigma_j|\tilde{\theta}_{j}|>a\lambda$, by Condition \ref{cond:penalty} and (\ref{eq: FCP lower derivative}) in Lemma \ref{lem: FCP}, we have
\[
J_{\lambda}(\hat\sigma_j\tilde{\theta}_{j}) \geq J_\lambda(a\lambda)  \geq a_1\lambda\cdot a\lambda = aa_1\lambda^2.
\]
Then
\[
|\mathcal{B}|\lambda^{2}\leq \frac{1}{aa_1}\sum_{j\in\mathcal{B}}J_{\lambda}(\hat\sigma_j\tilde{\theta}_{j})\leq \frac{\mu}{aa_1}\ \Rightarrow\ |\mathcal{B}|\leq\dfrac{\mu}{aa_1\lambda^{2}} = \dfrac{Cs\lambda^2}{ \lambda^{2}} = Cs.
\]
Define $\mathcal{A} = \mathcal{S}\cup\mathcal{B}$, then
\begin{equation}\label{eq: A set property}
  |\mathcal{A}| \leq(C+1)s\text{  and  }\mathcal{A}^{c}=\{j\in\mathcal{S}^{c}:\hat\sigma_j|\tilde{\theta}_{j}|\leq a\lambda\}.
\end{equation}
Furthermore, define $\DDelta=(\Delta_1,\Delta_2,\ldots,\Delta_p)^\top$, where $\Delta_j = \tilde\theta_j-\theta^{*}_j$. For any $j\in\mathcal{A}^c$, $\hat\sigma_j|\Delta_{j}|=\hat\sigma_j|\tilde{\theta}_{j}|\leq a\lambda$. Also, by Condition \ref{cond:penalty}, when $\hat\sigma_j|\tilde{\theta}_{j}|\leq a\lambda$, we have $J_{\lambda}(\hat\sigma_j\tilde{\theta}_{j}) \geq a_1\lambda\cdot \hat\sigma_j|\tilde{\theta}_{j}|$.
Then
\[
\lambda\sum_{j\in\mathcal{A}^c}\hat\sigma_j|\Delta_{j}|=\lambda\sum_{j\in\mathcal{A}^c}\hat\sigma_j|\tilde{\theta}_{j}|\leq \frac{1}{a_1}\sum_{j\in\mathcal{A}^c}J_{\lambda}(\hat\sigma_j \tilde{\theta}_{j})\leq \dfrac{\mu}{a_1}.
\]
Therefore,
\begin{equation}\label{eq: sum A bound}
\sum_{j\in\mathcal{A}^c}\hat\sigma_j|\Delta_{j}|\leq\dfrac{\mu}{a_1\lambda} = Cs\lambda.
\end{equation}
Define $\mathcal{J} = \{j\in[p]:\QQ_{\mathcal{A},j} \neq \boldsymbol{0}_{|\mathcal{A}|}\}$. Assumption \ref{assu:cointegration}(b) implies that $|\mathcal{J}|= O(s)$. By definition, $\QQ_{\mathcal{A},\mathcal{J}^c} = \boldsymbol{0}_{|\mathcal{A}|\times |\mathcal{J}^c|}$, and thus $(\RR_n\QQ^\top)_{\mathcal{J}^c,\mathcal{A}} = \boldsymbol{0}_{|\mathcal{J}^c|\times |\mathcal{A}|}$. Define
\begin{equation}\label{eq: Delta sharp def}
  \tilde\DDelta^\sharp = n^{-1/2}\RR_n\QQ^\top\DDelta,
\end{equation}  
and $\DD=\operatorname{diag}(\hat\sigma_1,\ldots,\hat\sigma_p)$. Then 
\begin{align}
  \|\tilde\DDelta^\sharp_{\mathcal{J}^c}\|_1 &= n^{-1/2}\|(\RR_n\QQ^\top\DDelta)_{\mathcal{J}^c}\|_1 = n^{-1/2}\|(\RR_n\QQ^\top)_{\mathcal{J}^c,\mathcal{A}^c}\DDelta_{\mathcal{A}^c}\|_1  \nonumber \\
  &=n^{-1/2}\|(\RR_n\QQ^\top\DD^{-1})_{\mathcal{J}^c,\mathcal{A}^c}(\DD\DDelta)_{\mathcal{A}^c}\|_1  \leq n^{-1/2}\|(\RR_n\QQ^\top\DD^{-1})_{\mathcal{J}^c,\mathcal{A}^c}\|_1 \|(\DD\DDelta)_{\mathcal{A}^c}\|_1 \nonumber \\
  &= \max_{j\in\mathcal{A}^c}\frac{\sum_{i\in\mathcal{J}^c}  (\RR_n)_{i,i}|Q_{j,i}|}{\sqrt{n}\hat\sigma_j} \left(\sum_{j\in\mathcal{A}^c}\hat\sigma_j|\Delta_{j}|\right),\label{eq: sum Ac bound 0}
\end{align}
 By the definition of $\RR_n$ in Proposition \ref{prop: OLS} and that of $\QQ$ in (\ref{eq: rotation}), we have
\[
\begin{aligned}
    \frac{\sum_{i\in\mathcal{J}^c}  (\RR_n)_{i,i}|Q_{j,i}|}{\sqrt{n}\hat\sigma_j} \leq\, &\frac{1}{\hat\sigma_j}, &j\in\mathcal{M}_z, \\
  \frac{\sum_{i\in\mathcal{J}^c}  (\RR_n)_{i,i}|Q_{j,i}|}{\sqrt{n}\hat\sigma_j} \leq\, &\frac{1+\sqrt n\,\max_{k}\|A_{k,\cdot}\|_1}{\hat\sigma_j}, &j\in\mathcal{C}_{\rm a}, \\
    \frac{\sum_{i\in\mathcal{J}^c}  (\RR_n)_{i,i}|Q_{j,i}|}{\sqrt{n}\hat\sigma_j} \leq\,&\frac{\sqrt{n}}{\hat\sigma_j}, &j\in\mathcal{C}_{\rm b}\cup \mathcal{M}_x. \\
\end{aligned}
\]
Lemma \ref{lem:sd bound} implies that w.p.a.1
\[
\min_{j\in\mathcal{M}_z}\hat\sigma_j \geq c\text{, and }\min_{j\in\mathcal{M}_x\cup\mathcal{C}_{\rm a}\cup\mathcal{C}_{\rm b}}\hat\sigma_j \geq \frac{c\sqrt{n}}{\sqrt{\log p}}
\]
for some absolute constant $c$. Therefore,
\[
\max_{j\in\mathcal{A}^c}\frac{\sum_{i\in\mathcal{J}^c}  (\RR_n)_{i,i}|Q_{j,i}|}{\sqrt{n}\hat\sigma_j} \leq C\sqrt{\log p}
\]
w.p.a.1, which together with (\ref{eq: sum A bound}) and (\ref{eq: sum Ac bound 0}) implies that
\begin{equation}\label{eq: sum Ac bound}
  \|\tilde\DDelta^\sharp_{\mathcal{J}^c}\|_1 \leq C\sqrt{\log p} \left(\sum_{j\in\mathcal{A}^c}\hat\sigma_j|\Delta_{j}|\right) \leq C\sqrt{\log p}\cdot s\lambda.
\end{equation}
If
\begin{equation}\label{eq: sum A bound case 1}
  \|\tilde\DDelta^\sharp_{\mathcal{J}}\|_1 \leq \frac{C\sqrt{\log p}}{3} \times s\lambda,
\end{equation}
we have
\begin{equation}\label{eq: sum all bound case 1}
 \|\tilde\DDelta^\sharp\|_2 \leq \|\tilde\DDelta^\sharp\|_1 = \|\tilde\DDelta^\sharp_{\mathcal{J}}\|_1 + \|\tilde\DDelta^\sharp_{\mathcal{J}^c}\|_1 \leq \frac{4C\sqrt{\log p}}{3} \times s\lambda,
\end{equation}
which implies consistency. Otherwise, when (\ref{eq: sum A bound case 1}) does not hold, i.e.~$\|\tilde\DDelta^\sharp_{\mathcal{J}}\|_1 > \frac{C\sqrt{\log p}}{3} \times s\lambda$, combining (\ref{eq: sum Ac bound}) we have
\begin{equation}\label{eq: cone condition 1}
\|\tilde\DDelta^\sharp_{\mathcal{J}^c}\|_1 \leq 3\|\tilde\DDelta^\sharp_{\mathcal{J}}\|_1.
\end{equation}
That means the $\ell_{1}$-norm of the scaled rotated estimation error $\tilde\DDelta^\sharp$
is controlled by its subvector indexed by $\mathcal{J}$ with cardinality $|\mathcal{J}|=O(s)$. Then we can use the restricted eigenvalue condition.

By (\ref{eq: FCP bound lambda2}) in Lemma \ref{lem: FCP}, we have $\sum_{j=1}^p J_\lambda(\hat\sigma_j\theta^*_j)\leq s\cdot \sup_x J_\lambda(x) \leq Cs\lambda^2 = \mu$, so the true coefficient $\ttheta^*$ is feasible for the constrained optimization problem. Therefore, by the
fact that $\tilde\ttheta$ is the global minimizer,
\[\frac{1}{n}\|\YY-\WW\tilde\ttheta\|_2^2 + \sum_{j=1}^p J_\lambda(\hat\sigma_j\tilde\theta_j) \leq \frac{1}{n}\|\YY-\WW\ttheta^*\|_2^2 + \sum_{j=1}^p J_\lambda(\hat\sigma_j\theta^*_j).\]
Rearranging the terms, we obtain
\[\DDelta^\top \frac{\WW^\top \WW}{n}\DDelta \leq \frac{2\uu^\top\WW\DDelta}{n} +  \sum_{j=1}^p \left[J_\lambda(\hat\sigma_j\theta^*_j) - J_\lambda(\hat\sigma_j\tilde\theta_j)\right].\]
Define $\tilde \WW^\sharp = \sqrt{n}\WW^\sharp\RR_n^{-1}$, and note that $\WW\DDelta = \tilde\WW^\sharp \tilde\DDelta^\sharp$. Therefore,
\begin{equation}\label{eq: upper bound 1}
  \begin{aligned}
  \tilde\DDelta^{\sharp\top} \frac{\tilde\WW^{\sharp\top} \tilde\WW^\sharp}{n}\tilde\DDelta^\sharp     &\leq \frac{2\uu^\top\tilde\WW^\sharp\tilde\DDelta^\sharp}{n} +  \sum_{j=1}^p \left[J_\lambda(\hat\sigma_j\theta^*_j) - J_\lambda(\hat\sigma_j\tilde\theta_j)\right] \\
  &\leq 2\left\|\frac{\tilde\WW^{\sharp\top}\uu}{n}\right\|_\infty\cdot\|\tilde\DDelta^\sharp\|_1 +  \sum_{j=1}^p \left[J_\lambda(\hat\sigma_j\theta^*_j) - J_\lambda(\hat\sigma_j\tilde\theta_j)\right] \\
  &\leq 2\left\|\frac{\tilde\WW^{\sharp\top}\uu}{n}\right\|_\infty\cdot\|\tilde\DDelta^\sharp\|_1 +  2\mu. \\
  \end{aligned}
\end{equation}
By the RE condition in Lemma \ref{lem:RE xz}, we have $ \tilde\DDelta^{\sharp\top} \frac{\tilde\WW^{\sharp\top} \tilde\WW^\sharp}{n}\tilde\DDelta^\sharp  \geq \frac{c}{s\log p} \|\tilde\DDelta^\sharp\|_2^2$,  and thus
\[\|\tilde\DDelta^\sharp\|_2^2 \leq \frac{s\log p}{c} \left( 2\left\|\frac{\tilde\WW^{\sharp\top}\uu}{n}\right\|_\infty\cdot\|\tilde\DDelta^\sharp\|_1 +  2\mu \right)\leq \frac{s\log p}{c} \left(  \left\|\frac{\tilde\WW^{\sharp\top}\uu}{n}\right\|_\infty C\sqrt{s}\|\tilde\DDelta^\sharp\|_2 +  2\mu \right),\]
where the last inequality uses the fact that
$$\begin{aligned}
  \|\tilde\DDelta^\sharp\|_1 & = \sum_{j=1}^p |\tilde\DDelta^\sharp_j|
  = \sum_{j \in \mathcal{J}} |\tilde\DDelta^\sharp_j| + \sum_{j \notin \mathcal{J}} |\tilde\DDelta^\sharp_j|
  \leq \sum_{j \in \mathcal{J}} |\tilde\DDelta^\sharp_j| + 3 \sum_{j \in \mathcal{J}} |\tilde\DDelta^\sharp_j| \\
  &= 4 \sum_{j \in \mathcal{J}} |\tilde\DDelta^\sharp_j|  \leq 4\sqrt{|\mathcal{J}|} \|\tilde\DDelta^\sharp\|_2 \leq C  \sqrt{s}\|\tilde\DDelta^\sharp\|_2 .
\end{aligned} $$
By the DB condition in (\ref{eq:DB u sharp}) of Lemma \ref{lem:DB}, we have $\left\|\frac{\tilde\WW^{\sharp\top}\uu}{n}\right\|_\infty \leq C\frac{(\log p)^{1+1/r}}{\sqrt{n}}$. Therefore,
\[\begin{aligned}
  \|\tilde\DDelta^\sharp\|_2^2 &\leq \frac{s\log p}{c} \left( C^2 \cdot \frac{\sqrt{s}(\log p)^{1+1/r}}{\sqrt{n}}\cdot\|\tilde\DDelta^\sharp\|_2 +  2\mu \right)\\
  &\leq \frac{s\log p}{c} \left( \frac{C^4 s^2 (\log p)^{3+2/r}}{2cn} + \frac{c}{2s\log p}\|\tilde\DDelta^\sharp\|_2^2 +  2\mu \right) \\
  &\leq   \frac{C^4 s^3 (\log p)^{4+2/r}}{2c^2 n} + \frac{1}{2}\|\tilde\DDelta^\sharp\|_2^2 +  \frac{2s\log p}{c}\cdot \mu,
\end{aligned}\]
and therefore
\begin{equation} \label{eq: final bound L2}
  \|\tilde\DDelta^\sharp\|_2^2 \leq \frac{C^4 s^3 (\log p)^{4+2/r}}{c^2 n} +   \frac{4s\log p}{c}\cdot \mu = O\left(\frac{s^3 (\log p)^{4+2/r}}{n} +  s^2\log p\cdot \lambda^2 \right)
\end{equation}
when $\mu = O\left(s\lambda^2\right)$. It follows that
\begin{equation}\label{eq: final bound}
  \|\tilde\DDelta^\sharp\|_1 = O\left(\sqrt{s}\|\tilde\DDelta^\sharp\|_2\right) = O\left(\frac{s^2 (\log p)^{2+1/r}}{\sqrt{n}} +  s^{3/2}\sqrt{\log p}\cdot \lambda \right).
\end{equation}
The bound in (\ref{eq: final bound}) dominates that in (\ref{eq: sum all bound case 1}) and provides the rate of convergence for the $\ell_1$-norm of $\tilde\DDelta^\sharp$.

The $\ell_1$ and $\ell_2$ bounds for $\sqrt{n}\tilde\DDelta^\sharp = \RR_n\QQ^\top(\tilde\ttheta-\ttheta^*)$ in Theorem \ref{thm: global} are obtained by (\ref{eq: sum all bound case 1}), (\ref{eq: final bound L2}) and (\ref{eq: final bound}) when $\lambda = O\left(\frac{(\log p)^{1+1/r}}{\sqrt{n}}\right)$. The proof of prediction error follows directly from the arguments in the proof of  Proposition \ref{prop: OLS}. 

It remains to show the bounds of $\tilde\ggamma$, $\tilde\pphi$, and $\tilde\bbeta$. When $q=1$, the $\ell_1$-norm convergence rate of $\tilde\DDelta^\sharp$ implies
\[\begin{aligned}
  &\|\tilde\ggamma-\ggamma^*\|_1 + \|\tilde\pphi^{\rm a}-\pphi^{\rm a*}\|_1 \\
  &+ \sqrt{n}\|\AA^\top\tilde\pphi^{\rm a}+\tilde\pphi^{\rm b}-\AA^\top\pphi^{\rm a*}-\pphi^{\rm b*}\|_1 +\sqrt{n}\|\tilde\bbeta - \bbeta^*\|_1 \\
\lep  \,&\frac{s^{2}(\log p)^{2+\frac{1}{r}}}{\sqrt{n}}.
\end{aligned}
\]
It also implies that
\[\begin{aligned}\|\tilde\pphi^{\rm b} - \pphi^{\rm b*}\|_1 &\lep  \|\AA^\top\|_1 \|\tilde\pphi^{\rm a}-\pphi^{\rm a*}\|_1 + \frac{s^{2}(\log p)^{2+\frac{1}{r}}}{n}  \lep  \frac{s^{2}(\log p)^{2+\frac{1}{r}}}{\sqrt{n}}.\end{aligned}\]
For $q=2$, if (\ref{eq: sum A bound case 1}) holds, by the $\ell_1$ bound in (\ref{eq: sum all bound case 1}) we have 
\[\begin{aligned}
  &\|\tilde\ggamma-\ggamma^*\|_2 + \|\tilde\pphi^{\rm a}-\pphi^{\rm a*}\|_2 \\
  &+ \sqrt{n}\|\AA^\top\tilde\pphi^{\rm a}+\tilde\pphi^{\rm b}-\AA^\top\pphi^{\rm a*}-\pphi^{\rm b*}\|_2 +\sqrt{n}\|\tilde\bbeta - \bbeta^*\|_2 \\
\leq\,& \|\tilde\ggamma-\ggamma^*\|_1 + \|\tilde\pphi^{\rm a}-\pphi^{\rm a*}\|_1 \\
  &+ \sqrt{n}\|\AA^\top\tilde\pphi^{\rm a}+\tilde\pphi^{\rm b}-\AA^\top\pphi^{\rm a*}-\pphi^{\rm b*}\|_1 +\sqrt{n}\|\tilde\bbeta - \bbeta^*\|_1 \\
\lep  \,&\frac{s(\log p)^{3/2+\frac{1}{r}}}{\sqrt{n}},
\end{aligned}
\]
and 
\[\begin{aligned}\|\tilde\pphi^{\rm b} - \pphi^{\rm b*}\|_2 &\lep  \|\AA^\top\|_1 \|\tilde\pphi^{\rm a}-\pphi^{\rm a*}\|_1 + \frac{s (\log p)^{3/2+\frac{1}{r}}}{\sqrt{n}} \lep  \frac{s (\log p)^{3/2+\frac{1}{r}}}{\sqrt{n}}.\end{aligned}\]
This establishes the desired $\ell_2$ bound. 
Otherwise, (\ref{eq: cone condition 1}) holds, and therefore $\|\tilde\DDelta^\sharp\|_1^2 = O(s\|\tilde\DDelta^\sharp\|_2^2)$. Following the same arguments for (\ref{eq: dd_subset_bound}), we can show that 
\begin{equation}\label{eq: dd_subset_bound_2}
  \|\tilde\DDelta^\sharp_{\mathcal{M}_z \cup \mathcal{C}_{\rm a}}\|_2^2 \leq C \tilde\DDelta^{\sharp\top}  \hat\SSigma^\sharp \tilde\DDelta^\sharp  
\end{equation}
w.p.a.1.  In addition, the rightmost upper bound in (\ref{eq: upper bound 1}) together with the DB condition in Lemma \ref{lem:DB}, the $\ell_1$ bound in (\ref{eq: final bound}), and $\mu = Cs\lambda^2$ with $\lambda = C(\log p)^{1+1/r}/\sqrt{n}$ suggests that 
\[\tilde\DDelta^{\sharp\top}  \hat\SSigma^\sharp \tilde\DDelta^\sharp \lep \frac{s^2 (\log p)^{3+\frac{2}{r}}}{n}.\] 
Together with (\ref{eq: dd_subset_bound_2}), we obtain 
\begin{equation}\label{eq:l2 st}
\|\tilde\ggamma-\ggamma^*\|_2 + \|\tilde\pphi^{\rm a}-\pphi^{\rm a*}\|_2 \lep \frac{s (\log p)^{3/2+\frac{1}{r}}}{\sqrt{n}}.
\end{equation} 
In addition, note that the set $\mathcal{A}$ in (\ref{eq: A set property}) satisfies $|\mathcal{A}|=O(s)$, and by (\ref{eq: sum A bound}) and the bounds of standard deviations in Lemma \ref{lem:sd bound} we have
\[\sum_{j\in\mathcal{A}^c}|\Delta_j| \lep s\lambda.\]
Therefore, 
\begin{align}
  \|\AA^\top(\tilde\pphi^{\rm a}-\pphi^{\rm a*})\|_2 &\lep \|(\AA_{ \mathcal{A},\cdot})^\top(\tilde\pphi^{\rm a}-\pphi^{\rm a*})_{\mathcal{A}}\|_2+ \|(\AA_{ \mathcal{A}^c,\cdot})^\top(\tilde\pphi^{\rm a}-\pphi^{\rm a*})_{\mathcal{A}^c}\|_2 \notag\\
  &\leq \sum_{i\in\mathcal{A}}\|\AA_{i,\cdot}\|_2 \cdot  |(\tilde\pphi^{\rm a}_i-\pphi^{\rm a*}_i)| +  \max_{i\in\mathcal{A}^c}\|\AA_{i,\cdot}\|_1 \cdot \|(\tilde\pphi^{\rm a}-\pphi^{\rm a*})_{\mathcal{A}^c}\|_1 \label{eq: Aphi bound}\\
  &\leq \max_{i\in\mathcal{A}}\|\AA_{i,\cdot}\|_1 \cdot O(\sqrt{s} \|\tilde\pphi^{\rm a}-\pphi^{\rm a*}\|_2) + \max_{i\in\mathcal{A}^c}\|\AA_{i,\cdot}\|_1 \cdot s\lambda \notag\\
  & \lep \frac{s^{3/2} (\log p)^{3/2+\frac{1}{r}}}{\sqrt{n}} \notag
\end{align}
when $\lambda = C(\log p)^{1+1/r}/\sqrt{n}$. Then combining (\ref{eq: final bound L2}), we have
\begin{equation}\label{eq:l2 b}
  \begin{aligned}
\|\tilde\pphi^{\rm b}-\pphi^{\rm b*}\|_2 &\leq \frac{\|\tilde\DDelta^\sharp_{\mathcal{C}_{\rm b}}\|_2}{\sqrt{n}} + \|\AA^\top(\tilde\pphi^{\rm a}-\pphi^{\rm a*})\|_2  \lep \frac{s^{3/2} (\log p)^{2+\frac{1}{r}}}{\sqrt{n}}.
  \end{aligned}
\end{equation} 
Finally, (\ref{eq: final bound L2}) also implies that
\[\|\sqrt{n}(\tilde\bbeta-\bbeta^*)\|_2 =  \|\tilde\DDelta^\sharp_{\mathcal{M}_x}\|_2 \lep \frac{s^{3/2} (\log p)^{2+\frac{1}{r}}}{\sqrt{n}},\]
which together with (\ref{eq:l2 st}) and (\ref{eq:l2 b}) establishes the desired $\ell_2$ bounds for all components of $\tilde\ttheta-\ttheta^*$. We conclude the proof of Theorem \ref{thm: global}.

\subsubsection{Proof of Theorem \ref{thm: oracle}}
\underline{\bf Proof of Part (a)}.
Define $S_j(\ttheta) = \frac{1}{n}\sum_{t=1}^n  w_{j,t-1}(y_t -  \ww_{t-1}^\top\ttheta)$. We will show
\begin{itemize}
  \item[(A)] $S_j(\hat\ttheta^{\rm or}) = 0$ and $|\hat\sigma_j \hat\theta_j^{\rm or}| > b \lambda$ for $j\in\mathcal{S}$;
  \item[(B)] $|\hat\sigma_j^{-1} S_j(\hat\ttheta^{\rm or})| < \frac12a_1\lambda$ for $j\in\mathcal{S}^{c}$; 
  \item[(C)] The matrix $\hat\SSigma_{\mathcal{S},\mathcal{S}} = \frac{1}{n}\sum_{t=1}^n  \ww_{\mathcal{S},t-1} \ww_{\mathcal{S},t-1}^{\top}$ is positive definite w.p.a.1.
\end{itemize}
Note that Condition (A) implies that $\hat\sigma_j|\hat\theta_j^{\rm or}+\delta_j| \geq b \lambda$ for $j\in\mathcal{S}$ whenever \(\ddelta\) is sufficiently small such that $\max_{j\in\mathcal{S}}\hat\sigma_j|\delta_j| = o(\lambda)$. We then have
\begin{equation}\label{eq:Jl_constant_S}
  J_\lambda(\hat\sigma_j(\hat\theta_j^{\rm or}+\delta_j)) = J_\lambda(\hat\sigma_j\hat\theta_j^{\rm or})  \text{ for all }j\in\mathcal{S},
\end{equation}
and thus
\[
\sum_{j=1}^p {} \Bigl[J_\lambda(\hat\sigma_j(\hat\theta_j^{\rm or}+\delta_j)) - J_\lambda(\hat\sigma_j\hat\theta_j^{\rm or})\Bigr]  = \sum_{j\in\mathcal{S}^{c}} J_\lambda(\hat\sigma_j \delta_j).
\]
Therefore, for any $\ddelta \neq \boldsymbol{0}$ such that $\|\ddelta\|_1$ is sufficiently small, we have
\[\begin{aligned}
  \mathcal{L}_\lambda(\hat\ttheta^{\rm or} + \ddelta) - \mathcal{L}_\lambda(\hat\ttheta^{\rm or}) &= \frac{1}{n}\sum_{t=1}^n\left((y_t - \ww_{t-1}^\top (\hat\ttheta^{\rm or} + \ddelta))^2 - (y_t - \ww_{t-1}^\top\hat\ttheta^{\rm or})^2\right) \\
  &\,\,\,\,\ \ + \sum_{j=1}^{p}J_\lambda\left(\hat\sigma_j (\hat\theta_j^{\rm or}+\delta_j)\right) - \sum_{j=1}^{p}J_\lambda(\hat\sigma_j\hat\theta_j^{\rm or}) \\
  &\geq \frac{1}{n}\ddelta^\top \sum_{t=1}^n \ww_{t-1} \ww_{t-1}^\top \ddelta -2\sum_{j=1}^{p}S_j(\hat\ttheta^{\rm or})\delta_j +  \sum_{j\in\mathcal{S}^{c}} J_\lambda\left(\hat\sigma_j \delta_j\right)    \\
  &\geq  \frac{1}{n}\ddelta^\top \sum_{t=1}^n \ww_{t-1} \ww_{t-1}^\top \ddelta - 2\sum_{j\in\mathcal{S}^c} {} \left|\frac{S_j(\hat\ttheta^{\rm or})}{\hat\sigma_j}\right|\cdot \hat\sigma_j |\delta_j|  + a_1\lambda\sum_{j\in\mathcal{S}^{c}}\hat\sigma_j|\delta_j| \\
  &\geq  \frac{1}{n}\ddelta^\top \sum_{t=1}^n \ww_{t-1} \ww_{t-1}^\top \ddelta +\left(a_1\lambda -2 \max_{j\in\mathcal{S}^c} {} \left|\frac{S_j(\hat\ttheta^{\rm or})}{\hat\sigma_j}\right| \right)\sum_{j\in\mathcal{S}^c} \hat\sigma_j|\delta_j|.
\end{aligned}\]
If $\ddelta_{\mathcal{S}^c} \neq \boldsymbol{0}$, the second term on the right-hand side of the inequality is strictly positive due to Condition (B). Otherwise, when $\ddelta_{\mathcal{S}^c} = \boldsymbol{0}$, the first term
\[\begin{aligned}
  \frac{1}{n}\ddelta^\top \sum_{t=1}^n \ww_{t-1} \ww_{t-1}^\top \ddelta &=  \frac{1}{n}\ddelta_{\mathcal{S}}^\top \sum_{t=1}^n \ww_{\mathcal{S},t-1} \ww_{\mathcal{S},t-1}^\top \ddelta_{\mathcal{S}}
   > 0
\end{aligned}\]
where the last step applies Condition (C). Therefore, $\mathcal{L}_\lambda(\hat\ttheta^{\rm or} + \ddelta) - \mathcal{L}_\lambda(\hat\ttheta^{\rm or}) > 0$ for any $\ddelta \neq \boldsymbol{0}$ such that $\|\ddelta\|_1$ is sufficiently small, which verifies that $\hat\ttheta^{\rm or}$ is a strict local minimizer of $\mathcal{L}_\lambda(\ttheta)$. It then suffices to prove Conditions (A) to (C).

\underline{Proof of Condition (A)}. $S_j(\hat\ttheta^{\rm or}) = 0$ for $j\in\mathcal{S}$ holds trivially by the definition of the oracle OLS estimator, as the regressors are orthogonal to the residuals. Furthermore, note that Proposition \ref{prop: OLS} establishes the error bounds $|\hat\theta_j^{\rm or} - \theta_j^*|$ for $j\in\mathcal{S}$, and Assumption \ref{assu:beta-min} provides the lower bounds of the nonzero coefficients $|\theta_j^*|$ for $j\in\mathcal{S}$. Combining these results with the bounds of $\hat\sigma_j$ in Lemma \ref{lem:sd bound} yields that w.p.a.1,
\[ \begin{aligned}
  \min_{j\in\mathcal{S}\cap\mathcal{M}_x}\hat\sigma_j|\hat\theta_j^{\rm or}| &\geq
  \min_{j\in\mathcal{S}\cap\mathcal{M}_x}\hat\sigma_j\left(|\theta_j^*| -  |\hat\theta_j^{\rm or} - \theta_j^*| \right) \\
  &\geq
  \min_{j\in\mathcal{S}\cap\mathcal{M}_x}\hat\sigma_j\left(\min_{j\in\mathcal{S}\cap\mathcal{M}_x}|\theta_j^*| -  \max_{j\in\mathcal{S}\cap\mathcal{M}_x}|\hat\theta_j^{\rm or} - \theta_j^*| \right) \\
  &\gg  C \sqrt{\frac{n}{\log p}}\cdot \left(\frac{s^{3/2}(\log p)^{3+1/r}}{n}
 -  \left(\frac{s^{3/2}(\log p)^{2+1/r}}{n}\right)\right)\\
  &\geq  C\left(\frac{s^{3/2}(\log p)^{5/2+1/r}}{\sqrt{n}}\right) > b \lambda, \\
\end{aligned}
\]

\[ \begin{aligned}
  \min_{j\in\mathcal{S}\cap\mathcal{M}_z}\hat\sigma_j|\hat\theta_j^{\rm or}| &\geq
  \min_{j\in\mathcal{S}\cap\mathcal{M}_z}\hat\sigma_j\left(\min_{j\in\mathcal{S}\cap\mathcal{M}_z}|\theta_j^*| -  \max_{j\in\mathcal{S}\cap\mathcal{M}_z}|\hat\theta_j^{\rm or} - \theta_j^*| \right)\\
  &\gg  C \cdot \frac{s^{3/2}(\log p)^{3+1/r}}{\sqrt{n}}
 - c \cdot  \left(\frac{s^{3/2}(\log p)^{2+1/r}}{\sqrt{n}}\right)\\
  &\geq  C\cdot \left(\frac{s^{3/2}(\log p)^{5/2+1/r}}{\sqrt{n}}\right) > b \lambda, \\
\end{aligned}
\]
and
\[ \begin{aligned}
  \min_{j\in\mathcal{S}\cap\mathcal{C}}\hat\sigma_j|\hat\theta_j^{\rm or}| &\geq
  \min_{j\in\mathcal{S}\cap\mathcal{C}}\hat\sigma_j\left(\min_{j\in\mathcal{S}\cap\mathcal{C}}|\theta_j^*| -  \max_{j\in\mathcal{S}\cap\mathcal{C}}|\hat\theta_j^{\rm or} - \theta_j^*| \right) \\
    &\gg  C \sqrt{\frac{n}{\log p}}\cdot \left(\frac{s^{3/2}(\log p)^{3+1/r}}{\sqrt{n}}
 -   \frac{s^{3/2}(\log p)^{2+1/r}}{\sqrt{n}}\right)\\
  &\geq  C\left( s^{3/2}(\log p)^{5/2+1/r} \right) > b \lambda. \\
\end{aligned}
\]
This completes the proof of Condition (A).

\underline{Proof of Condition (B)}. Define  $\ww_{j,\cdot} = (w_{j,0}, w_{j,1}, \ldots, w_{j,n-1})^\top$. By Lemma \ref{lem:sd bound}, when $j\in\mathcal{M}_x\cup\mathcal{C}_{\rm a}\cup \mathcal{C}_{\rm b}$, we have $\hat\sigma_j \geq c\sqrt{n/\log p}$. By the fact that $\max\limits_{j\in\mathcal{M}_x\cup\mathcal{C}_{\rm a}\cup \mathcal{C}_{\rm b},t\in[n]}|w_{j,t}| = O_p(\sqrt{n \log p})$, we have
\begin{equation}\label{eq: hat sigma lower bound nonstationary}
  \max\limits_{j\in\mathcal{M}_x\cup\mathcal{C}_{\rm a}\cup \mathcal{C}_{\rm b}}\frac{\|\ww_{j,\cdot}\|_2}{n\cdot\hat\sigma_j} = O_p\left(\frac{\sqrt{n\cdot n\log p}}{n \cdot \sqrt{n/\log p}}\right) = O_p\left( \frac{\log p}{\sqrt{n}}\right).
\end{equation}
Similar arguments yield $ \max\limits_{j\in\mathcal{M}_z}\frac{\|\ww_{j,\cdot}\|_2}{n\cdot\hat\sigma_j} = O_p(\log p/\sqrt{n})$, which together with (\ref{eq: hat sigma lower bound nonstationary}) implies
\begin{equation}\label{eq: hat sigma lower bound}
  \max\limits_{j\in [p]}\frac{\|\ww_{j,\cdot}\|_2}{n\cdot\hat\sigma_j} =O_p\left( \frac{\log p}{\sqrt{n}}\right).
\end{equation}
Therefore,
\begin{equation}\label{eq: hat sigma S bound}
 \begin{aligned}
  \left|\hat\sigma_j^{-1}S_j(\hat\ttheta^{\rm or})\right| =  \left|\dfrac{\ww_{j,\cdot}^\top\MM_{\mathcal{S}}\uu}{n\cdot\hat\sigma_j}\right| &\leq \left|\frac{1}{n}\sum_{t=1}^n \frac{w_{j,t-1}}{\hat\sigma_j}u_t \right| + \left|\dfrac{\ww_{j,\cdot}^\top\PP_{\mathcal{S}}\uu}{n\cdot\hat\sigma_j}\right| \\
  &\leq \left|\frac{1}{n}\sum_{t=1}^n \frac{w_{j,t-1}}{\hat\sigma_j}u_t \right| + \frac{\|\ww_{j,\cdot}\|_2}{n\cdot\hat\sigma_j} \left\|\PP_{\mathcal{S}}\uu\right\|_2 \\
  &= \left|\frac{1}{n}\sum_{t=1}^n \frac{w_{j,t-1}}{\hat\sigma_j}u_t \right| + O_p\left( \frac{\log p}{\sqrt{n} } \right)\left\|\PP_{\mathcal{S}}\uu\right\|_2,
  \end{aligned}
\end{equation}
where the last step applies (\ref{eq: hat sigma lower bound}).
By the DB condition in (\ref{eq:DB u sharp}) of Lemma \ref{lem:DB} and the standard deviation bounds in Lemma \ref{lem:sd bound}, we have
\begin{equation}\label{eq: std DB}
  \max_{j\in\mathcal{S}^c}\left|\frac{1}{n}\sum_{t=1}^n \frac{w_{j,t-1}}{\hat\sigma_j}u_t \right| = O_p\left(\dfrac{(\log p)^{3/2+1/(2r)}}{\sqrt{n}}\right).
\end{equation}
Moreover, recall that we defined $ \mathcal{T} = \{j\in[p]:\QQ_{\mathcal{S},j} \neq \boldsymbol{0}_{|\mathcal{S}|}\}$ in (\ref{eq: def T set}), and we have $|\mathcal{T}| = O(s)$. Since the diagonal entries of $\QQ$, denoted by $Q_{j,j}$, are all equal to one, we have $\QQ_{\mathcal{S},j} \neq \boldsymbol{0}_{|\mathcal{S}|}$ for all $j\in \mathcal{S}$, which means $j \in \mathcal{T}$. Therefore, $\mathcal{S}\subset\mathcal{T}$. Since $\WW_{\mathcal{S}} = \WW_{\mathcal{T}}^\sharp(\QQ_{\mathcal{S},\mathcal{T}})^\top$, the column space of $\WW_{\mathcal{S}}$ is contained within the column space of $\WW_{\mathcal{T}}^\sharp$, which implies that
\[\begin{aligned}
  \uu^\top\PP_{\mathcal{S}}\uu &\leq    \uu^\top\WW_{\mathcal{T}}^\sharp(\WW_{\mathcal{T}}^{\sharp\top} \WW_{\mathcal{T}}^\sharp)^{-1}\WW_{\mathcal{T}}^{\sharp\top}\uu \\
  &\leq \left(\uu^\top\WW_{\mathcal{T}}^\sharp\RR_{n,\mathcal{T}}^{-1}\right)(\RR_{n,\mathcal{T}}^{-1}\WW_{\mathcal{T}}^{\sharp\top} \WW_{\mathcal{T}}^\sharp\RR_{n,\mathcal{T}}^{-1})^{-1}\RR_{n,\mathcal{T}}^{-1}\WW_{\mathcal{T}}^{\sharp\top}\uu \\
  &\leq \left\|\RR_{n,\mathcal{T}}^{-1}\WW_{\mathcal{T}}^{\sharp\top}\uu\right\|_2^2 \cdot \left\|(\hat\SSigma^\sharp_{\mathcal{T},\mathcal{T}})^{-1}\right\|_2  \leq |\mathcal{T}|\cdot \left\|\RR_{n,\mathcal{T}}^{-1}\WW_{\mathcal{T}}^{\sharp\top}\uu\right\|_\infty^2 \cdot \left\|(\hat\SSigma^\sharp_{\mathcal{T},\mathcal{T}})^{-1}\right\|_2, \\
\end{aligned}\]
where for the second line we define $\RR_{n,\mathcal{T}} = (\RR_{n})_{\mathcal{T},\mathcal{T}}$.
By $|\mathcal{T}| = O(s)$ and the RE condition in Lemma \ref{lem:RE xz}, we have $\|(\hat\SSigma^\sharp_{\mathcal{T},\mathcal{T}})^{-1}\|_2 = O_p(s\log p)$ and thus
\begin{equation}\label{eq: u quadratic bound}
  \|\PP_{\mathcal{S}}\uu\|_2 = \sqrt{\uu^\top\PP_{\mathcal{S}}\uu} = O_p\left(\sqrt{s^2 \log p \cdot \left\|\RR_{n,\mathcal{T}}^{-1}\WW_{\mathcal{T}}^{\sharp\top}\uu\right\|_\infty^2}\right)  = O_p\left(s (\log p)^{3/2+1/(2r)} \right),
\end{equation}
where the last step applies the DB condition in (\ref{eq:DB u sharp}) of Lemma \ref{lem:DB}. By (\ref{eq: hat sigma S bound}), (\ref{eq: std DB}) and (\ref{eq: u quadratic bound}), we have
\[\max_{j\in\mathcal{S}^c} \left|\hat\sigma_j^{-1}S_j(\hat\ttheta^{\rm or})\right| = O_p\left(\dfrac{(\log p)^{3/2+1/(2r)}}{\sqrt{n}}\right) + O_p\left(\frac{s (\log p)^{5/2+1/(2r)} }{\sqrt{n}}\right) =  O_p\left(\frac{s (\log p)^{5/2+1/(2r)} }{\sqrt{n}}\right).\]
This upper bound is far smaller than $\lambda$ w.p.a.1. We complete the proof of Condition (B).

\underline{Proof of Condition (C)}. Note that
\[
 \WW_{\mathcal{S}}^\top\WW_{\mathcal{S}}= \QQ_{\mathcal{S},\mathcal{T}}\WW_{\mathcal{T}}^{\sharp\top} \WW_{\mathcal{T}}^\sharp\QQ_{\mathcal{S},\mathcal{T}}^\top = \QQ_{\mathcal{S},\mathcal{T}}(\RR_{n})_{\mathcal{T},\mathcal{T}}\hat\SSigma^\sharp_{\mathcal{T},\mathcal{T}}(\RR_{n})_{\mathcal{T},\mathcal{T}}\QQ_{\mathcal{S},\mathcal{T}}^\top.
\]
The RE condition in Lemma \ref{lem:RE xz} and $|\mathcal{T}| = O(s)$ imply that $\hat\SSigma^\sharp_{\mathcal{T},\mathcal{T}}$ is positive definite w.p.a.1. In addition, both $\RR_n$ and $\QQ$ are invertible, and the fact that $\mathcal{S}\subset \mathcal{T}$ ensures that $(\RR_{n})_{\mathcal{T},\mathcal{T}}\QQ_{\mathcal{S},\mathcal{T}}^\top$ has full column rank. Specifically, since $\QQ_{\mathcal S,\mathcal T^c}= \boldsymbol{0}$ and the selected rows of the invertible matrix $\QQ$ are linearly independent, $\QQ_{\mathcal S,\mathcal T}$ has full row rank. Multiplication by the invertible $(\RR_{n})_{\mathcal{T},\mathcal{T}}$ preserves the full row rank. Therefore, $\WW_{\mathcal{S}}^\top\WW_{\mathcal{S}}$ is positive definite w.p.a.1, thereby completing the proof of Condition (C) and concluding the proof of Theorem \ref{thm: oracle}(a).

\medskip

\underline{\bf Proof of Part (b)}. The proof contains two steps. (1) Bound the estimation error $\DDelta = \tilde\ttheta-\ttheta^*$. (2) Prove Part (b) by contradiction. 

\underline{Step (1). Bound the estimation error.} Recall that $ \tilde\DDelta^\sharp = n^{-1/2}\RR_n\QQ^\top(\tilde\ttheta - \ttheta^*)$ is defined in (\ref{eq: Delta sharp def}). By (\ref{eq: FCP bound lambda2}) in Lemma \ref{lem: FCP}, we have $\sum_{j=1}^p J_\lambda(\hat\sigma_j\theta^*_j)\leq s\cdot \sup_x J_\lambda(x) \leq Cs\lambda^2 = \mu$, the true coefficient $\ttheta^*$ is feasible for the constrained optimization problem. Therefore, by the
fact that $\tilde\ttheta$ is the global minimizer,
\[\frac{1}{n}\|\YY-\WW\tilde\ttheta\|_2^2 + \sum_{j=1}^p J_\lambda(\hat\sigma_j\tilde\theta_j) \leq \frac{1}{n}\|\YY-\WW\ttheta^*\|_2^2 + \sum_{j=1}^p J_\lambda(\hat\sigma_j\theta^*_j).\]
Rearranging the terms, we obtain
\[\DDelta^\top \frac{\WW^\top \WW}{n}\DDelta + \sum_{j\in\mathcal{S}^c} J_\lambda(\hat\sigma_j\tilde\theta_j) \leq \frac{2\uu^\top\WW\DDelta}{n} +  \sum_{j\in\mathcal{S}} \left[J_\lambda(\hat\sigma_j\theta^*_j) - J_\lambda(\hat\sigma_j\tilde\theta_j)\right].\]
Note that in the proof of Theorem \ref{thm: global} we defined
\[\mathcal{B} = \{j\in\mathcal{S}^c : \hat\sigma_j|\tilde\theta_j| > a\lambda\},\,\mathcal{A} = \mathcal{S}\cup \mathcal{B}.\]
By the first inequality in (\ref{eq: sum Ac bound}), we have
\[\begin{aligned}
  \|\tilde\DDelta^\sharp_{\mathcal{J}^c}\|_1 &\leq C\sqrt{\log p}\sum_{j\in\mathcal{A}^c}  \hat\sigma_j|\tilde\theta_j| \\
  &\leq \frac{C\sqrt{\log p}\sum_{j\in\mathcal{A}^c}  J_\lambda(\hat\sigma_j\tilde\theta_j)}{a_1\lambda} \leq \frac{C\sqrt{\log p}\sum_{j\in\mathcal{S}^c}  J_\lambda(\hat\sigma_j\tilde\theta_j)}{a_1\lambda},
\end{aligned}\] 
where the second inequality follows from (\ref{eq: FCP lower derivative}) and the fact that $j\in\mathcal{A}^c$ implies $\hat\sigma_j|\tilde\theta_j| \leq a\lambda$, and the last inequality uses $\mathcal{A}^c \subset \mathcal{S}^c$.
Thus, we further have 
\[\DDelta^\top \frac{\WW^\top \WW}{n}\DDelta + \frac{a_1\lambda\|\tilde\DDelta^\sharp_{\mathcal{J}^c}\|_1}{C  \sqrt{\log p}} \leq \frac{2\uu^\top\WW\DDelta}{n} +  \sum_{j\in\mathcal{S}} \left[J_\lambda(\hat\sigma_j\theta^*_j) - J_\lambda(\hat\sigma_j\tilde\theta_j)\right].\]  
Define $\tilde \WW^\sharp = \sqrt{n}\WW^\sharp\RR_n^{-1}$, and note that $\WW\DDelta = \tilde\WW^\sharp \tilde\DDelta^\sharp$. Therefore,
\begin{equation}\label{eq: basic ineq 1}
  \begin{aligned}
  \tilde\DDelta^{\sharp\top} \frac{\tilde\WW^{\sharp\top} \tilde\WW^\sharp}{n}\tilde\DDelta^\sharp + \frac{a_1\lambda\|\tilde\DDelta^\sharp_{\mathcal{J}^c}\|_1}{C  \sqrt{\log p}}   &\leq \frac{2\uu^\top\tilde\WW^\sharp\tilde\DDelta^\sharp}{n} +  \sum_{j \in \mathcal{S}} \left[J_\lambda(\hat\sigma_j\theta^*_j) - J_\lambda(\hat\sigma_j\tilde\theta_j)\right] \\
  &\leq 2\left\|\frac{\tilde\WW^{\sharp\top}\uu}{n}\right\|_\infty\cdot\|\tilde\DDelta^\sharp\|_1 +  \sum_{j \in \mathcal{S}}  \left[J_\lambda(\hat\sigma_j\theta^*_j) - J_\lambda(\hat\sigma_j\tilde\theta_j)\right].\\
  \end{aligned}
\end{equation}
We now bound $ \sum_{j \in \mathcal{S}}  \left[J_\lambda(\hat\sigma_j\theta^*_j) - J_\lambda(\hat\sigma_j\tilde\theta_j)\right]$. Define the following objects:
\begin{align}
  \mathcal{K} &= \{j\in\mathcal{S}: \hat\sigma_j|\tilde\theta_j| < b\lambda\}, \notag\\ 
  \HH &= {\rm diag}(h_1,\ldots,h_p),\,h_j = \sqrt{n}\textbf{1}{\{j\in\mathcal{M}_x\}} + \textbf{1}{\{j\notin\mathcal{M}_x\}}.\label{eq: def H}
\end{align}
By the standard deviation bounds in Lemma \ref{lem:sd bound}, we have $\max_{j\in[p]} ({h_j}/{\hat\sigma_j}) \leq C\sqrt{\log p}$ 
w.p.a.1, and thus for any $j\in\mathcal{K}$
\[h_j|\tilde\theta_j| \leq C\sqrt{\log p} \hat\sigma_j|\tilde\theta_j| \leq bC\lambda\sqrt{\log p}.\]
Define $m_n := \min_{j\in\mathcal{S}} h_j |\theta_j^*|$. Since $m_n \gg \lambda\sqrt{\log p}$ by Assumption \ref{assu:beta-min}, the upper bound for $h_j|\tilde\theta_j|$ is eventually smaller than $m_n/2$. Therefore, 
\begin{equation}\label{eq: h theta diff}
  \begin{aligned}
     \min_{j\in\mathcal{K}} h_j|\tilde\theta_j - \theta_j^*| &\geq \min_{j\in\mathcal{K}} \left( h_j|\theta_j^*| - h_j|\tilde\theta_j| \right)  \geq   m_n - m_n/2 = m_n/2.
  \end{aligned} 
\end{equation}Summing up (\ref{eq: h theta diff}) over $j\in\mathcal{K}$, we obtain
\begin{equation*}
  \|\HH\DDelta\|_1 \geq \sum_{j\in\mathcal{K}} h_j|\tilde\theta_j - \theta_j^*| \geq |\mathcal{K}|\frac{m_n}{2},\,\,\,\,\|\HH\DDelta\|_2^2 \geq \sum_{j\in\mathcal{K}} h_j^2|\tilde\theta_j - \theta_j^*|^2 \geq |\mathcal{K}|\frac{m_n^2}{4},
\end{equation*}
or equivalently, 
\begin{equation}\label{eq: K bound}
  |\mathcal{K}| \leq \frac{2\|\HH\DDelta\|_1}{m_n}, \quad |\mathcal{K}| \leq \frac{4\|\HH\DDelta\|_2^2}{m_n^2}.
\end{equation} 
Note that for $j \in \mathcal{S}\backslash \mathcal{K}$, we have $\hat\sigma_j|\tilde\theta_j| \geq b\lambda$, and thus $J_\lambda(\hat\sigma_j\tilde\theta_j ) = J_\lambda(b\lambda)$ by (\ref{eq: J t flat}). In addition, 
\[\min_{j\in\mathcal{S}}\hat\sigma_j|\theta_j^*| \geq \frac{c\cdot m_n}{\sqrt{\log p}} \gg b\lambda,\]
and thus for $j \in \mathcal{S}\backslash \mathcal{K}$ we also have  $J_\lambda(\hat\sigma_j \theta_j^* ) = J_\lambda(b\lambda)$, suggesting that 
\[J_\lambda(\hat\sigma_j \theta_j^* ) - J_\lambda(\hat\sigma_j\tilde\theta_j) = 0\text{ for all $j \in \mathcal{S}\backslash \mathcal{K}$}.\] 
Therefore, 
\[ \sum_{j \in \mathcal{S}} \left[J_\lambda(\hat\sigma_j\theta^*_j) - J_\lambda(\hat\sigma_j\tilde\theta_j)\right] = \sum_{j\in\mathcal{K}} \left[J_\lambda(\hat\sigma_j\theta^*_j) - J_\lambda(\hat\sigma_j\tilde\theta_j)\right] \leq C|\mathcal{K}|\lambda^2,\]
where the last inequality applies (\ref{eq: FCP bound lambda2}). Applying (\ref{eq: K bound}), we obtain
\begin{equation}\label{eq: temp}
  \sum_{j \in \mathcal{S}} \left[J_\lambda(\hat\sigma_j\theta^*_j) - J_\lambda(\hat\sigma_j\tilde\theta_j)\right] \leq 4C\lambda^2  \frac{\|\HH\DDelta\|_q^q}{m_n^q}\text{ for } q=1,2.
\end{equation} 
By the definition of $\tilde\DDelta^\sharp = n^{-1/2}\RR_n\QQ^\top\DDelta $ in (\ref{eq: Delta sharp def}), we have  
\[\|\HH\DDelta\|_1 = \|\sqrt{n}\HH(\QQ^\top)^{-1} \RR_n^{-1} \tilde\DDelta^\sharp  \|_1 \leq \|\sqrt{n}\HH (\QQ^\top)^{-1} \RR_n^{-1} \|_1  \cdot \|\tilde\DDelta^\sharp\|_1 \leq C\|\tilde\DDelta^\sharp\|_1,\]
where the last inequality applies the fact that 
\[
\sqrt{n}\HH(\QQ^\top)^{-1}\RR_n^{-1}
=
\begin{pmatrix}
\bfI_{p_z} & \OO & \OO & \OO \\
\OO & \bfI_{p_{\rm a}} & \OO & \OO\\
\OO & -\AA^\top & n^{-1/2}\bfI_{p_{\rm b}} & \OO \\
\OO & \OO & \OO & \bfI_{p_x}
\end{pmatrix},
\]
and thus $\|\sqrt{n}\HH(\QQ^\top)^{-1}\RR_n^{-1}\|_1 \leq 1 + \max_{i}\|\AA_{i\cdot}\|_1 = O(1)$ implied by Assumption \ref{assu:cointegration}. Therefore, combining with (\ref{eq: temp}), we obtain
\begin{equation}\label{eq: temp2}
  \sum_{j \in \mathcal{S}} \left[J_\lambda(\hat\sigma_j\theta^*_j) - J_\lambda(\hat\sigma_j\tilde\theta_j)\right] \leq   \frac{4C^2\lambda^2}{m_n} \|\tilde\DDelta^\sharp\|_1.
\end{equation} 
The DB condition in Lemma \ref{lem:DB} implies that
\begin{equation}\label{eq: DB tilde W}
\left\|\frac{\tilde\WW^{\sharp\top}\uu}{n}\right\|_\infty \leq C\frac{ (\log p)^{1+1/r}}{\sqrt{n}}
\end{equation} 
w.p.a.1. 
By (\ref{eq: basic ineq 1}) and (\ref{eq: temp2}), we have
\begin{equation}\label{eq: basic ineq 2}
 \tilde\DDelta^{\sharp\top} \frac{\tilde\WW^{\sharp\top} \tilde\WW^\sharp}{n}\tilde\DDelta^\sharp +  \frac{a_1\lambda\|\tilde\DDelta^\sharp_{\mathcal{J}^c}\|_1}{C  \sqrt{\log p}} \leq  2\left(C\frac{ (\log p)^{1+1/r}}{\sqrt{n}} \|\tilde\DDelta^\sharp\|_1 +\frac{4C^2\lambda^2}{m_n}\|\tilde\DDelta^\sharp\|_1 \right),
\end{equation} 
which implies 
\[  {\|\tilde\DDelta^\sharp_{\mathcal{J}^c}\|_1} \leq  \frac{2C  \sqrt{\log p}}{a_1\lambda}\left(  \frac{ (\log p)^{1+1/r}}{\sqrt{n}}  + \frac{ \lambda^2 }{m_n} \right)\cdot\|\tilde\DDelta^\sharp\|_1.\] 
By the fact that $(\log p)^{3/2+1/r}/(\sqrt{n}\lambda) \to 0$ and $m_n \gg \lambda\sqrt{\log p}$ by the beta-min condition in Assumption \ref{assu:beta-min} when $\lambda = C_1\frac{s(\log p)^{5/2+1/r}}{\sqrt{n}}$, we have $\|\tilde\DDelta^\sharp_{\mathcal{J}^c}\|_1 \leq 0.75\|\tilde\DDelta^\sharp \|_1 = 0.75\left(\|\tilde\DDelta^\sharp_{\mathcal{J}^c}\|_1 + \|\tilde\DDelta^\sharp_{\mathcal{J}}\|_1\right)$, which implies $ \|\tilde\DDelta^\sharp_{\mathcal{J}^c}\|_1 \leq 3\|\tilde\DDelta^\sharp_{\mathcal{J}}\|_1$ as in \eqref{eq: cone condition 1}. This enables us to again establish the inequality
\begin{equation}\label{eq: dd_subset_bound_3}
  \|\tilde\DDelta^\sharp_{\mathcal{M}_z \cup \mathcal{C}_{\rm a}}\|_2^2 \leq C \tilde\DDelta^{\sharp\top}  \hat\SSigma^\sharp \tilde\DDelta^\sharp 
\end{equation} as in (\ref{eq: dd_subset_bound_2}).

Using the bounds of the standard deviations in Lemma \ref{lem:sd bound} for $j\in\mathcal{C}_{\rm a}$, we have 
\[\|\DDelta_{\mathcal{A}^c \cap \mathcal{C}_{\rm a}}\|_1 \lep \sqrt{\frac{\log p}{n}} \sum_{j\in\mathcal{A}^c \cap \mathcal{C}_{\rm a}}\hat\sigma_j |\DDelta_j| \lep \sqrt{\frac{\log p}{n}}s\lambda,\]
where the last step uses (\ref{eq: sum A bound}). Then similar to (\ref{eq: Aphi bound}), we have  
\begin{align}
  \|\AA^\top(\tilde\pphi^{\rm a}-\pphi^{\rm a*})\|_2 &\lep \max_{i}\|\AA_{i,\cdot}\|_1\left(\| \DDelta_{\mathcal{A}\cap\mathcal{C}_a}\|_1+\|\DDelta_{\mathcal{A}^c\cap\mathcal{C}_a}\|_1\right)\notag  \\
  &\leq \max_{i}\|\AA_{i,\cdot}\|_1 \cdot O(\sqrt{s} \|\tilde\pphi^{\rm a}-\pphi^{\rm a*}\|_2) + \max_{i}\|\AA_{i,\cdot}\|_1 \cdot s\lambda\sqrt{\frac{\log p}{n}} \notag\\
  & \leq C \left( \sqrt{s\tilde\DDelta^{\sharp\top}  \hat\SSigma^\sharp \tilde\DDelta^\sharp} + \frac{s\lambda\sqrt{\log p}}{\sqrt{n}}\right),\label{eq: Aphi bound 2}
\end{align}
where the last step applies (\ref{eq: dd_subset_bound_3}). 
Then by definition of $\HH$ and $\tilde\DDelta^\sharp$, 
\begin{equation}\label{eq: HH_Delta_bound}
  \begin{aligned}
  \|\HH  \DDelta\|_2^2 &= \|\tilde\DDelta^\sharp_{\mathcal{M}_z}\|_2^2 + \|\tilde\DDelta^\sharp_{\mathcal{C}_{\rm a}}\|_2^2 + \left\|\frac{\tilde\DDelta^\sharp_{\mathcal{C}_{\rm b}}}{\sqrt{n}} - \AA^\top(\tilde\pphi^{\rm a} - \pphi^{\rm a*})\right\|_2^2 + \|\tilde\DDelta^\sharp_{\mathcal{M}_{x}}\|_2^2 \\
  &\leq C\|\tilde\DDelta^\sharp\|_2^2 + C\|\AA^\top(\tilde\pphi^{\rm a}-\pphi^{\rm a*})\|_2^2 \\
  &\lep \left(s\log p \tilde\DDelta^{\sharp\top}  \hat\SSigma^\sharp \tilde\DDelta^\sharp\right) + \left(s\tilde\DDelta^{\sharp\top}  \hat\SSigma^\sharp \tilde\DDelta^\sharp + \frac{s^2\lambda^2 \log p}{n}\right) \\ &\leq C \left(s\log p \tilde\DDelta^{\sharp\top}  \hat\SSigma^\sharp \tilde\DDelta^\sharp + \frac{s^2\lambda^2 \log p}{n}\right),
  \end{aligned}
\end{equation}
where the third row applies the RE condition in Lemma \ref{lem:RE xz} that implies $ \tilde\DDelta^{\sharp\top} \frac{\tilde\WW^{\sharp\top} \tilde\WW^\sharp}{n}\tilde\DDelta^\sharp  \geq \frac{c}{s\log p} \|\tilde\DDelta^\sharp\|_2^2$, and the inequality (\ref{eq: Aphi bound 2}). Combining (\ref{eq: HH_Delta_bound}) and (\ref{eq: temp}) with $q=2$, we have 
\begin{equation}\label{eq: temp3}
  \sum_{j \in \mathcal{S}} \left[J_\lambda(\hat\sigma_j\theta^*_j) - J_\lambda(\hat\sigma_j\tilde\theta_j)\right] \leq   \frac{Cs\log p\lambda^2}{m_n^2} \tilde\DDelta^{\sharp\top}  \hat\SSigma^\sharp \tilde\DDelta^\sharp + \frac{Cs^2\lambda^4 \log p}{n\cdot m_n^2}.
\end{equation} 
By (\ref{eq: basic ineq 1}), (\ref{eq: DB tilde W}), and (\ref{eq: temp3}), we have
\begin{equation} \label{eq: temptemp}
 \tilde\DDelta^{\sharp\top} \frac{\tilde\WW^{\sharp\top} \tilde\WW^\sharp}{n}\tilde\DDelta^\sharp +  \frac{a_1\lambda\|\tilde\DDelta^\sharp_{\mathcal{J}^c}\|_1}{C  \sqrt{\log p}} \leq  C\left(\frac{ (\log p)^{1+1/r}}{\sqrt{n}} \|\tilde\DDelta^\sharp\|_1 +  \frac{ s\log p\lambda^2}{m_n^2} \tilde\DDelta^{\sharp\top}  \hat\SSigma^\sharp \tilde\DDelta^\sharp + \frac{ s^2\lambda^4 \log p}{n\cdot m_n^2} \right).
\end{equation} 
Given that $m_n^2 \gg s\log p\lambda^2$ by the beta-min condition in Assumption \ref{assu:beta-min}, we have $\frac{ C s\log p\lambda^2}{m_n^2} \tilde\DDelta^{\sharp\top}  \hat\SSigma^\sharp \tilde\DDelta^\sharp \leq 0.5\tilde\DDelta^{\sharp\top}  \hat\SSigma^\sharp \tilde\DDelta^\sharp$ with sufficiently large $n$. In addition, (\ref{eq: cone condition 1}) implies $\|\tilde\DDelta^\sharp\|_1 \leq C \sqrt{s} \|\tilde\DDelta^\sharp\|_2$ and the RE condition implies that $ \tilde\DDelta^{\sharp\top}  \hat\SSigma^\sharp \tilde\DDelta^\sharp \geq c\|\tilde\DDelta^\sharp\|_2^2/(s\log p)$. Therefore,  
\begin{equation}\label{eq: basic ineq 3}
  \begin{aligned}
\tilde\DDelta^{\sharp\top}  \hat\SSigma^\sharp \tilde\DDelta^\sharp &\leq  C \left(\frac{ (\log p)^{1+1/r}}{\sqrt{n}} \sqrt{s}\|\tilde\DDelta^\sharp\|_2+  \frac{Cs^2\lambda^4 \log p}{n\cdot m_n^2} \right) \\
&\leq \frac{Cs (\log p)^{3/2+1/r}}{\sqrt{n}}\sqrt{\tilde\DDelta^{\sharp\top}  \hat\SSigma^\sharp \tilde\DDelta^\sharp}+ \frac{Cs^2\log p \lambda^4}{n\cdot m_n^2} \\
&\leq 0.5\tilde\DDelta^{\sharp\top}  \hat\SSigma^\sharp \tilde\DDelta^\sharp + \frac{0.5Cs^2(\log p)^{3+2/r}}{n} + \frac{ Cs^2 \log p \lambda^4}{n\cdot m_n^2} \\
&\leq 0.5\tilde\DDelta^{\sharp\top}  \hat\SSigma^\sharp \tilde\DDelta^\sharp + \frac{ Cs^2(\log p)^{3+2/r}}{n}, 
  \end{aligned}
\end{equation}
where the last inequality applies the fact that $m_n^2 \gg \lambda^4$. Therefore, 
\begin{equation}
   \tilde\DDelta^{\sharp\top}  \hat\SSigma^\sharp \tilde\DDelta^\sharp \stackrel{\mathrm{p}}{\preccurlyeq} \left(\frac{s^{2} (\log p)^{3+2/r}}{n} \right).
\end{equation}  
 Combining (\ref{eq: HH_Delta_bound}), when $\lambda = Cs(\log p)^{5/2+1/r}/\sqrt{n}$ we have $\|\HH\DDelta\|_2^2 \lep s^3 (\log p)^{4+2/r}/n$. By the definition of $\HH$ in (\ref{eq: def H}),  we have 
\begin{equation}\label{eq:tildeDDelta_sharp_l2}
 \|\tilde\ggamma-\ggamma^*\|_2 + \|\tilde\pphi-\pphi^*\|_2 + \sqrt{n}\|\tilde\bbeta-\bbeta^*\|_2 \lep s^{3/2} (\log p)^{2+1/r}/\sqrt{n}.
\end{equation}

\underline{Step (2). Prove Part (b) by contradiction.}  Suppose that $\tilde\ttheta \neq \hat\ttheta^{\rm or}$. 
By (\ref{eq: A set property}), we have $|\mathcal{A}| = O(s)$.     
Then following similar arguments in the proof of Condition (A) in Part (a) for $\hat\ttheta^{\rm or}$, we can show that $\min_{j\in\mathcal{S}}|\hat\sigma_j \tilde\theta_j| > b\lambda$ using   (\ref{eq:tildeDDelta_sharp_l2}), the bounds of standard deviations in Lemma \ref{lem:sd bound}, and the beta-min condition in Assumption \ref{assu:beta-min}. Therefore, 
\begin{equation}\label{eq: J eq S}
  J_\lambda\left(\hat\sigma_j \tilde\theta_j\right) = J_\lambda\left(\hat\sigma_j\hat\theta_j^{\rm or}\right)\text{ for any $j\in\mathcal{S}$}.
\end{equation}
In addition, for any $j \in \mathcal{B}$ we have
\[ J_\lambda\left(\hat\sigma_j\tilde\theta_j\right) \geq J_\lambda(a\lambda) \geq aa_1\lambda^2,\]
where the second step applies  (\ref{eq: FCP lower derivative}) in Lemma \ref{lem: FCP}. Therefore,
\begin{equation}\label{eq: sum B J}
  \sum_{j\in\mathcal{B}}J_\lambda\left(\hat\sigma_j\tilde\theta_j\right) \geq aa_1|\mathcal{B}|\lambda^2.
\end{equation}
By (\ref{eq: J eq S}) and (\ref{eq: sum B J}), we have
\begin{equation}\label{eq: J diff lower}
  \sum_{j\in[p]}J_\lambda\left(\hat\sigma_j\tilde\theta_j\right) - \sum_{j\in[p]}J_\lambda(\hat\sigma_j\hat\theta_j^{\rm or}) \geq \sum_{j\in\mathcal{A}^c}J_\lambda\left(\hat\sigma_j\tilde\theta_j\right) + aa_1|\mathcal{B}|\lambda^2.
\end{equation}
Define $\hat\ttheta^{\mathcal{A}}=(\hat\theta_j^{\mathcal{A}})_{j\in[p]}$, where $\hat\theta_j^{\mathcal{A}}$ is the $j$-th component of the OLS estimator of $y_t$ on $\ww_{\mathcal{A},t-1}$ when $j\in\mathcal{A}$, and zero otherwise.
Then
\begin{equation}\label{eq: L diff lower}
  \begin{aligned}
  \mathcal{L}_\lambda(\tilde\ttheta) - \mathcal{L}_\lambda(\hat\ttheta^{\rm or}) &= \frac{1}{n}\sum_{t=1}^n\left((y_t - \ww_{t-1}^\top\tilde\ttheta)^2 - (y_t - \ww_{t-1}^\top\hat\ttheta^{\rm or})^2\right)  + \sum_{j\in[p]}J_\lambda\left(\hat\sigma_j\tilde\theta_j\right) - \sum_{j\in[p]}J_\lambda(\hat\sigma_j\hat\theta_j^{\rm or}) \\
  &\geq \frac{1}{n}\sum_{t=1}^n\left((y_t - \ww_{t-1}^\top\tilde\ttheta)^2 - (y_t - \ww_{t-1}^\top\hat\ttheta^{\mathcal{A}})^2\right) + \frac{1}{n} \uu^\top(\MM_{\mathcal{A}} - \MM_{\mathcal{S}})\uu\\
  &\,\,\,\,\,\,\,\,\,\,   + \sum_{j\in\mathcal{A}^c}J_\lambda\left(\hat\sigma_j\tilde\theta_j\right) + aa_1|\mathcal{B}|\lambda^2\\
  &= \frac{1}{n}(\tilde\ttheta-\hat\ttheta^{\mathcal{A}})^\top \sum_{t=1}^n \ww_{t-1} \ww_{t-1}^\top (\tilde\ttheta-\hat\ttheta^{\mathcal{A}}) -2\sum_{j=1}^{p}S_j(\hat\ttheta^{\mathcal{A}})( \tilde\theta_j - \hat\theta_j^{\mathcal{A}}) \\
  &\,\,\,\,\,\,\,\,\,\,  +\left(aa_1|\mathcal{B}|\lambda^2+\frac{1}{n} \uu^\top(\PP_{\mathcal{S}} - \PP_{\mathcal{A}})\uu\right) +  \sum_{j\in\mathcal{A}^{c}} J_\lambda\left(\hat\sigma_j \tilde\theta_j\right).
\end{aligned}
\end{equation}
Note that $S_j(\hat\ttheta^{\mathcal{A}}) = 0$ for any $j \in \mathcal{A}$, and thus
\begin{equation}\label{eq:Sj sum bound}
  \left|\sum_{j=1}^{p}S_j(\hat\ttheta^{\mathcal{A}})(\hat\theta_j^{\mathcal{A}} - \tilde\theta_j)\right| = \left|\sum_{j\in\mathcal{A}^c}S_j(\hat\ttheta^{\mathcal{A}})\tilde\theta_j\right|\leq \max_{j\in\mathcal{A}^c} \frac{|S_j(\hat\ttheta^{\mathcal{A}})|}{\hat\sigma_j} \sum_{j\in\mathcal{A}^c} \hat\sigma_j|\tilde\theta_j|.
\end{equation}
Also, note that for any $j\in\mathcal{A}^c$, we have $\hat\sigma_j|\tilde\theta_j| \leq a\lambda$ by definition. Therefore, by  (\ref{eq: FCP lower derivative}) in Lemma \ref{lem: FCP}, we have
\[J_\lambda\left(\hat\sigma_j \tilde\theta_j\right) \geq a_1\lambda\hat\sigma_j|\tilde\theta_j|.\]
Together with (\ref{eq:Sj sum bound}), we have
\begin{equation}\label{eq: J minus S lower}
  \sum_{j\in\mathcal{A}^{c}} J_\lambda\left(\hat\sigma_j \tilde\theta_j\right) - 2\sum_{j=1}^{p}S_j(\hat\ttheta^{\mathcal{A}})\tilde\theta_j \geq \left(a_1\lambda - \max_{j\in\mathcal{A}^c} \frac{2|S_j(\hat\ttheta^{\mathcal{A}})|}{\hat\sigma_j}\right) \sum_{j\in\mathcal{A}^c} \hat\sigma_j|\tilde\theta_j| =: d_n \sum_{j\in\mathcal{A}^c} \hat\sigma_j|\tilde\theta_j|,
\end{equation}
where $d_n := a_1\lambda - \max_{j\in\mathcal{A}^c} \frac{2|S_j(\hat\ttheta^{\mathcal{A}})|}{\hat\sigma_j}$. By (\ref{eq: L diff lower}) and (\ref{eq: J minus S lower}), we have
\begin{equation}
   \begin{aligned}
  \mathcal{L}_\lambda(\tilde\ttheta) - \mathcal{L}_\lambda(\hat\ttheta^{\rm or})
  &\geq \frac{1}{n}(\tilde\ttheta-\hat\ttheta^{\mathcal{A}})^\top \sum_{t=1}^n \ww_{t-1} \ww_{t-1}^\top (\tilde\ttheta-\hat\ttheta^{\mathcal{A}})  +  d_n \sum_{j\in\mathcal{A}^c} \hat\sigma_j|\tilde\theta_j| +B_n,
\end{aligned}
\end{equation}
where
$
 B_n := aa_1|\mathcal{B}|\lambda^2+\frac{1}{n} \uu^\top(\PP_{\mathcal{S}} - \PP_{\mathcal{A}})\uu.
$
Recall that $|\mathcal{A}|=O(s)$ and thus the same arguments for Condition (B) in Part (a) yield that $d_n > 0$ w.p.a.1. We then show that
\begin{equation}\label{eq: L diff result positive}
  \mathcal{L}_\lambda(\tilde\ttheta) - \mathcal{L}_\lambda(\hat\ttheta^{\rm or}) > 0.
\end{equation}

\underline{Case I. $|\mathcal{B}|=0$}. In this case, $\mathcal{A}=\mathcal{S}$ and therefore $B_n = 0$ and $\hat\ttheta^{\mathcal{A}} = \hat\ttheta^{\rm or}$. Therefore,
\begin{equation*}
   \begin{aligned}
  \mathcal{L}_\lambda(\tilde\ttheta) - \mathcal{L}_\lambda(\hat\ttheta^{\rm or})
  &\geq \frac{1}{n}(\tilde\ttheta-\hat\ttheta^{\rm or})^\top \sum_{t=1}^n \ww_{t-1} \ww_{t-1}^\top (\tilde\ttheta-\hat\ttheta^{\rm or})  +  d_n \sum_{j\in\mathcal{S}^c} \hat\sigma_j|\tilde\theta_j|.
\end{aligned}
\end{equation*}
If $\tilde\ttheta_{\mathcal{S}^c} \neq \boldsymbol{0}$, the second term on the right-hand side of the inequality is strictly positive since $d_n > 0$. Otherwise, when $\tilde\ttheta_{\mathcal{S}^c} = \boldsymbol{0}$, by $\tilde\ttheta \neq \hat\ttheta^{\rm or}$ we have $\tilde\ttheta_{\mathcal{S}}-\hat\ttheta^{\rm or}_{\mathcal{S}} \neq \boldsymbol{0}$, and thus the first term
\[\begin{aligned}
 \frac{1}{n}(\tilde\ttheta-\hat\ttheta^{\rm or})^\top \sum_{t=1}^n \ww_{t-1} \ww_{t-1}^\top (\tilde\ttheta-\hat\ttheta^{\rm or}) &=  \frac{1}{n}(\tilde\ttheta_{\mathcal{S}}-\hat\ttheta^{\rm or}_{\mathcal{S}})^\top \sum_{t=1}^n \ww_{\mathcal{S},t-1} \ww_{\mathcal{S},t-1}^\top (\tilde\ttheta_{\mathcal{S}}-\hat\ttheta_{\mathcal{S}}^{\rm or}) > 0,
\end{aligned}\]
where the last step applies Condition (C) in the proof of Part (a). Therefore, (\ref{eq: L diff result positive}) holds.

\underline{Case II. $|\mathcal{B}|\geq 1$}. In this case, (\ref{eq: u quadratic bound}) implies that
\[\begin{aligned}
  \frac{\uu^\top\PP_{\mathcal{S}}\uu}{n} = O_p\left(\frac{s^2(\log p)^{3+\frac{1}{r}}}{n}\right),
\end{aligned}\]
and similar arguments yield
\[\begin{aligned}
  \frac{\uu^\top\PP_{\mathcal{A}}\uu}{n} &= O_p\left(\frac{s^2(\log p)^{3+\frac{1}{r}}}{n}\right).
\end{aligned}\]
This upper bound is dominated by $aa_1|\mathcal{B}|\lambda^2$ when $\lambda = C_1s(\log p)^{5/2+1/(r)}/\sqrt{n}$ for a sufficiently large absolute constant $C_1$. Therefore, we have $B_n > 0$, and thus
\begin{equation*}
   \begin{aligned}
  \mathcal{L}_\lambda(\tilde\ttheta) - \mathcal{L}_\lambda(\hat\ttheta^{\rm or})
  &\geq B_n > 0.
\end{aligned}
\end{equation*}
That means (\ref{eq: L diff result positive}) holds. Therefore, whenever $\tilde\ttheta \neq \hat\ttheta^{\rm or}$, we have $\mathcal{L}_\lambda(\tilde\ttheta) - \mathcal{L}_\lambda(\hat\ttheta^{\rm or}) > 0$. By (\ref{eq: FCP bound lambda2}) in Lemma \ref{lem: FCP} we have $\sum_{j=1}^p J_\lambda(\hat\sigma_j\hat\theta^{\rm or}_j)\leq s\cdot \sup_x J_\lambda(x) \leq Cs\lambda^2 = \mu$, and thus the oracle estimator $\hat\ttheta^{\rm or}$ is feasible for the constrained optimization problem. Therefore, $\mathcal{L}_\lambda(\tilde\ttheta) - \mathcal{L}_\lambda(\hat\ttheta^{\rm or}) > 0$ means that $\tilde\ttheta$ does not minimize $\mathcal{L}_\lambda(\ttheta)$. The contradiction implies that $\hat\ttheta^{\rm or}$ is indeed the unique minimizer of the constrained problem. We complete the proof of Theorem \ref{thm: oracle}. 

\subsubsection{Proof of Proposition \ref{prop: inconsistent LASSO}}

\underline{Proof of (a).} Note that under (\ref{eq: Aphi1+phi2}),  we have
\[\begin{aligned}
  y_t - \ww_{t-1}^\top\breve\ttheta &= u_t + \xx_{t-1}^{\rm ab\top}(\pphi^* - \breve\pphi) \\
  &= u_t +  \vv_{t-1}^{\rm a\top}(\pphi^{\rm a*} - \breve{\pphi}^{\rm a}) + \xx_{t-1}^{\rm b\top} (\AA^\top\pphi^{\rm a*}+\pphi^{\rm b*}-\AA^\top\breve\pphi^{\rm a}-\breve\pphi^{\rm b}) \\
  &= u_t +  \vv_{t-1}^{\rm a\top}(\pphi^{\rm a*} - \breve{\pphi}^{\rm a}).
\end{aligned}\]
Therefore,
\begin{equation}\label{eq: upper bound MSE breve}
  \begin{aligned}
  \frac{1}{n}\sum_{t=1}^n(y_t - \ww_{t-1}^\top\breve\ttheta)^2 &= \frac{1}{n}\sum_{t=1}^n\left(u_t +  \vv_{t-1}^{\rm a\top}(\pphi^{\rm a*} - \breve{\pphi}^{\rm a})\right)^2 \\
  &\leq  \frac{2}{n}\sum_{t=1}^n u_t^2 + 2\|\pphi^{\rm a*} - \breve{\pphi}^{\rm a}\|_1^2 \cdot \Biggl\|\frac{1}{n}\sum_{t=1}^n   \vv_{t-1}^{\rm a}  \vv_{t-1}^{\rm a\top}\Biggr\|_{\max}.
  \end{aligned}
\end{equation}
The first term $\frac{2}{n}\sum_{t=1}^n u_t^2 = O_p(1)$
 given the stationarity and finite variance of $u_t$. The second term is also $O_p(1)$, because
\[\|\pphi^{\rm a*} - \breve{\pphi}^{\rm a}\|_1 \leq \|\pphi^{\rm a*}\|_1 +  \|\breve{\pphi}^{\rm a}\|_1 \leq (1+c_1)\|\pphi^{\rm a*}\|_1 = O(1),\]
and
\begin{align*}
   \Biggl\|\frac{1}{n}\sum_{t=1}^n   \vv_{t-1}^{\rm a}  \vv_{t-1}^{\rm a\top}\Biggr\|_{\max} &\leq \Biggl\|\frac{1}{n}\sum_{t=1}^n   \vv_{t-1}^{\rm a}  \vv_{t-1}^{\rm a\top} - \mathbb{E}(  \vv_{t-1}^{\rm a}  \vv_{t-1}^{\rm a\top})\Biggr\|_{\max} + \Bigl\|\mathbb{E}(  \vv_{t-1}^{\rm a}  \vv_{t-1}^{\rm a\top})\Bigr\|_{\max} \\
   &= O_p\biggl(\sqrt{\frac{\log p}{n}}\biggr) + O(1) = O_p(1),
\end{align*}
where the order of the first term applies the Bernstein inequality for sub-exponential time series, and the order of the second term applies the bounded eigenvalues of the population covariance matrix by Assumption \ref{assu:covMat}. Therefore, by (\ref{eq: upper bound MSE breve}) we have
\begin{equation}\label{eq: MSE Op1}
  \frac{1}{n}\sum_{t=1}^n(y_t - \ww_{t-1}^\top\breve\ttheta)^2 = O_p(1).
\end{equation}
In addition, for all $j\in\mathcal{C}\cap\mathcal{S}$, we have
\begin{equation}\label{eq: nondegenerate d_j}
  \hat\sigma_j/\sqrt{n} \overset{\rm d}{\to} d_j,
\end{equation}
where $d_j$ is a non-degenerate random variable whose expression can be found in the proof of Corollary 2 in \citet{lee2022lasso}. Therefore,
\[\lambda\sum_{j\in\mathcal{C}\cap\mathcal{S}}\hat\sigma_j|\theta_j^*|  \stackrel{\mathrm{p}}{\asymp}  \sqrt{n}\lambda \sum_{j\in\mathcal{C}\cap\mathcal{S}}\left(\frac{\hat\sigma_j}{\sqrt{n}}\right)\overset{\rm p}{\to} \infty\]
when $\sqrt{n}\lambda \to \infty$,
which together with (\ref{eq: MSE Op1}) implies
\begin{equation}\label{eq: MSE upper}
  \frac{1}{n}\sum_{t=1}^n(y_t - \ww_{t-1}^\top\breve\ttheta)^2 \leq \frac{1-c_1}{2}\cdot \lambda\sum_{j\in\mathcal{C}\cap\mathcal{S}}\hat\sigma_j|\theta_j^*|
\end{equation}
w.p.a.1. In addition, note that $|\breve\pphi| \preceq c_1|\pphi^*|$, and other subvectors in $\breve\ttheta$ and $\ttheta^*$ are identical. Therefore, $\breve\theta_j = 0$ whenever $\theta^*_j = 0$, which implies that
\begin{align}\label{eq: sigma theta diff lower 1}
  \sum_{j=1}^p\hat\sigma_j(|\theta_j^*|-|\breve\theta_j|)&= \sum_{j\in\mathcal{C}\cap\mathcal{S}}\hat\sigma_j(|\theta_j^*|-|\breve\theta_j|)\\
   &\geq (1-c_1)\sum_{j\in\mathcal{C}\cap\mathcal{S}}\hat\sigma_j|\theta_j^*|.\label{eq: sigma theta diff lower 2}
\end{align}
Hence,
\begin{equation}\label{eq: lower bound positive}
  \begin{aligned}
 \mathcal{L}_\lambda(\ttheta^*) - \mathcal{L}_\lambda(\breve\ttheta) &\geq  \lambda\sum_{j=1}^p\hat\sigma_j |\theta^*_j| - \frac{1}{n}\sum_{t=1}^n(y_t - \ww_{t-1}^\top\breve\ttheta)^2 - \lambda\sum_{j=1}^p \hat\sigma_j |\breve\theta_j| \\
 &=\lambda\sum_{j\in\mathcal{C}\cap\mathcal{S}} \hat\sigma_j (|\theta^*_j|-|\breve\theta_j|) -  \frac{1}{n}\sum_{t=1}^n(y_t - \ww_{t-1}^\top\breve\ttheta)^2   \\
 &\geq (1-c_1)\lambda\sum_{j\in\mathcal{C}\cap\mathcal{S}}\hat\sigma_j |\theta_j^*| - \frac{1-c_1}{2}\lambda\sum_{j\in\mathcal{C}\cap\mathcal{S}}\hat\sigma_j |\theta_j^*| \\
  &= \frac{1-c_1}{2} \lambda\sum_{j\in\mathcal{C}\cap\mathcal{S}}\hat\sigma_j |\theta^*_j| \geq \frac{c_\phi(1-c_1)}{2} \cdot \sqrt{n}\lambda\sum_{j\in\mathcal{C}\cap\mathcal{S}}\frac{\hat\sigma_j}{\sqrt{n}} \overset{\rm p}{\to} \infty,
  \end{aligned}
\end{equation}
where the second row applies (\ref{eq: sigma theta diff lower 1}), the third row applies (\ref{eq: MSE upper}) and (\ref{eq: sigma theta diff lower 2}), and the last step applies $\sqrt{n}\lambda\to \infty$ and (\ref{eq: nondegenerate d_j}). Then (\ref{eq: LASSO inconsistency}) follows. 

We then show that $\|\hat\ttheta - \ttheta^*\|_\infty \overset{\rm p}{\not\to} 0$ by contradiction. Note that
\begin{equation*}
   \left|\lambda\sum_{j\in  \mathcal{S}}\hat\sigma_j(|\hat\theta_j|-|\theta_j^*|)\right| \leq \sqrt{n}\lambda \sum_{j\in   \mathcal{S}} \left(\frac{\hat\sigma_j}{\sqrt{n}}\right)\|\hat\ttheta - \ttheta^*\|_\infty.
\end{equation*}
We have $\sum_{j\in   \mathcal{S}} \frac{\hat\sigma_j}{\sqrt{n}} = O_p(1)$ given that $s$ is fixed.  If $\|\hat\ttheta - \ttheta^*\|_\infty \overset{\rm p}{\to} 0$, we have
\begin{equation}\label{eq: error 1}
   \left|\lambda\sum_{j\in  \mathcal{S}}\hat\sigma_j(|\hat\theta_j|-|\theta_j^*|)\right| = o_p(\sqrt{n}\lambda) \leq \frac{c_\phi(1-c_1)}{4} \sqrt{n}\lambda\sum_{j\in \mathcal{C}\cap \mathcal{S}} \left(\frac{\hat\sigma_j}{\sqrt{n}}\right)
\end{equation}
w.p.a.1 when $n$ is sufficiently large. In addition, when $\sqrt{n}\lambda\to\infty$,
\begin{equation}\label{eq: error 2}
  \frac{1}{n}\sum_{t=1}^n u_t^2 = O_p(1) \leq \frac{c_\phi(1-c_1)}{8} \sqrt{n}\lambda\sum_{j\in \mathcal{C}\cap \mathcal{S}} \left(\frac{\hat\sigma_j}{\sqrt{n}}\right).
\end{equation}
By (\ref{eq: error 1}) and (\ref{eq: error 2}), we have
  \begin{align}
  \mathcal{L}_\lambda(\hat\ttheta) - \mathcal{L}_\lambda(\ttheta^*) &\geq \lambda\sum_{j\in\mathcal{S}}\hat\sigma_j|\hat\theta_j|-\lambda\sum_{j\in\mathcal{S}}\hat\sigma_j|\theta_j^*|-\frac{1}{n}\sum_{t=1}^n u_t^2 \notag\\
  &\geq  -\frac{ c_\phi(1-c_1)}{4} \sqrt{n}\lambda\sum_{j\in \mathcal{C}\cap \mathcal{S}} \left(\frac{\hat\sigma_j}{\sqrt{n}}\right) -\frac{ c_\phi(1-c_1)}{8} \sqrt{n}\lambda\sum_{j\in \mathcal{C}\cap \mathcal{S}} \left(\frac{\hat\sigma_j}{\sqrt{n}}\right) \notag\\
  &= -\frac{ 3c_\phi(1-c_1)}{8} \sqrt{n}\lambda\sum_{j\in \mathcal{C}\cap \mathcal{S}} \left(\frac{\hat\sigma_j}{\sqrt{n}}\right). \label{eq: lower bound negative}
  \end{align}
By (\ref{eq: lower bound positive}) and (\ref{eq: lower bound negative}), we have
\begin{align*}
  \mathcal{L}_\lambda(\hat\ttheta) - \mathcal{L}_\lambda(\breve\ttheta) &= \left(\mathcal{L}_\lambda(\ttheta^*) - \mathcal{L}_\lambda(\breve\ttheta)\right) + \left(\mathcal{L}_\lambda(\hat\ttheta) - \mathcal{L}_\lambda(\ttheta^*)\right) \\
  &\geq\frac{  c_\phi(1-c_1)}{2} \sqrt{n}\lambda\sum_{j\in \mathcal{C}\cap \mathcal{S}} \left(\frac{\hat\sigma_j}{\sqrt{n}}\right) -\frac{ 3c_\phi(1-c_1)}{8} \sqrt{n}\lambda\sum_{j\in \mathcal{C}\cap \mathcal{S}} \left(\frac{\hat\sigma_j}{\sqrt{n}}\right) \\
  &= \frac{  c_\phi(1-c_1)}{8} \sqrt{n}\lambda\sum_{j\in \mathcal{C}\cap \mathcal{S}} \left(\frac{\hat\sigma_j}{\sqrt{n}}\right) \overset{\rm p}{\to} \infty,
\end{align*}
which contradicts the fact that $\hat\ttheta$ minimizes $\mathcal{L}_\lambda(\ttheta)$. Therefore, we conclude $\|\hat\ttheta - \ttheta^*\|_\infty \overset{\rm p}{\not\to} 0$.

\underline{Proof of (b).} Let
\[
\breve{\pphi}^{\rm a} =  c_3{\pphi}^{\rm a*},\quad\breve{\pphi}^{\rm b} =  \left(1-c_3\right)\AA^\top \pphi^{\rm a*} + \pphi^{\rm b*}
\]
with $c_3 = c_2(2-c_2)$. We show that these coefficients satisfy the conditions in Part (a).  It is straightforward to show that
\[\AA^\top\breve\pphi^{\rm a} + \breve\pphi^{\rm b} = \AA^\top (c_3\pphi^{\rm a*}) + ((1-c_3)\AA^\top \pphi^{\rm a*} + \pphi^{\rm b*}) = \AA^\top \pphi^{\rm a*} + \pphi^{\rm b*},\]
which satisfies (\ref{eq: Aphi1+phi2}). In addition, define $c_1 = 1-(1-c_2)^3$. It is easy to show that $c_3\leq c_1$ when $c_2\in(0,1)$, and thus $|\breve\pphi^{\rm a}| = c_3|\pphi^{\rm a*}| \preceq c_1|\pphi^{\rm a*}|$. Therefore, it suffices to show that $|\breve\pphi^{\rm b}| \preceq c_1|\pphi^{\rm b*}|$.

For a generic vector $\aa = (a_j)$, define $|\aa| = (|a_j|)$. Note that  $\AA^\top \pphi^{\rm a*} \odot \pphi^{\rm b*} \preceq \textbf{0}$ implies that
\begin{equation}\label{eq: abs equiv}
  \begin{aligned}
|\breve\pphi^{\rm b}| = |(1-c_3)\AA^\top \pphi^{\rm a*} + \pphi^{\rm b*}| &=  \left||\pphi^{\rm b*}| -(1-c_3)|\AA^\top \pphi^{\rm a*}|\right| \\
&= |\pphi^{\rm b*}| -(1-c_3)|\AA^\top \pphi^{\rm a*}|.
  \end{aligned}
\end{equation}
The last equality in (\ref{eq: abs equiv}) applies the fact that
\[
\begin{aligned}
  |\pphi^{\rm b*}| -(1-c_3)|\AA^\top \pphi^{\rm a*}| &\succeq |\pphi^{\rm b*}| -(1-c_3)(1+c_2)|\pphi^{\rm b*}|\\
   &= c_2(1+c_2-c_2^2) |\pphi^{\rm b*}| \\
   &\succeq \textbf{0}, \\
\end{aligned}
\]
where the second step applies $c_3 = c_2(2-c_2)$ and the third row applies $c_2 \in (0,1)$. Therefore, by $(1-c_2)|\pphi^{\rm b*}| \preceq |\AA^\top \pphi^{\rm a*}|$,
\[\begin{aligned}
  |(1-c_3)\AA^\top \pphi^{\rm a*} + \pphi^{\rm b*}|  \preceq |\pphi^{\rm b*}| -(1-c_3)(1-c_2)|\pphi^{\rm b*}| = \left(1-(1-c_2)^3\right)|\pphi^{\rm b*}| = c_1|\pphi^{\rm b*}|,
\end{aligned}\]
which completes the proof of Part (b).

\section{Additional Numerical Studies\protect\label{sec:add-simul}}

\subsection{Tuning Parameter Selection\protect\label{subsec: tuning selection}}

The following data-driven tuning parameter selection applies to all simulation studies in the main text and this supplement. Regarding the initial value for FCPs, we first initialize the coefficients as the zero vector to select $\lambda$ using cross-validation. Under the selected $\lambda$, we generate five random initial values centered at $\hat\ttheta^{\rm Las}$, the LASSO estimate, with a small normal perturbation, and pick the initial value that yields the smallest penalized criterion function $\mathcal{L}_\lambda(\hat\ttheta)$ with an FCP penalty.
The random perturbation is to shake the inconsistent $\hat\ttheta^{\rm Las}$.
For the constrained SCAD and MCP in (\ref{eq: constrained regularized estimation}), we first choose the $\lambda$ by cross-validation from the unconstrained counterparts and store the penalty level $\sum_{j=1}^pJ_\lambda(\hat\sigma_j\hat\theta_j)$ of the unconstrained regularization. We then select the best constraint parameter $\mu = c\sum_{j=1}^pJ_\lambda(\hat\sigma_j\hat\theta_j)$ with a grid of $c$'s that minimizes the penalized loss function in (\ref{eq: constrained regularized estimation}).

\subsection{Additional Simulation Results\protect\label{sec:add-simul-results}}
This section considers two additional cointegration simulation designs: common stochastic trends and pairwise cointegration.

\textbf{DGP A1. Common stochastic trends.} We consider the following cointegrated variables with nonstationary common factors
\[
\xx_{t}^{\rm a}=-\BB\ff_t+\vv^{\rm a}_{t},\qquad
\xx_{t}^{\rm b}=\BB\ff_t+\vv^{\rm b}_{t},
\]
where $\xx_{t}^{\rm a}$ and $\xx_{t}^{\rm b}$ are $(p/4)$-dimensional cointegrated regressors,   $\ff_t=\rho_n\ff_{t-1}+\vareepsilon_t^f$ is a two-dimensional vector of nonstationary latent factors with $\rho_n=1$ (unit roots) or $\rho_n=1-1/n$ (local unit roots), and the elements of $\vv_{t}^{\rm a}$, $\vv_{t}^{\rm b}$, $\vareepsilon_t^f$, and $\zz_t$ are i.i.d standard normal variables. All entries in the $(p/4)\times 2$  loading matrix $\BB$ are i.i.d~$\mathcal{N}(0,1)$ variables. In addition, we have $p/4$ i.i.d standard normal regressors $\zz_t$ and $p/4$ LUR regressors $\xx_t = \rho_n\xx_{t-1} + \ee_t$, where $\ee_t \overset{\rm i.i.d}{\sim} \mathcal{N}(0,\bfI_{p/4})$. The predictive regression is
\[
y_t=\xx_{t-1}^{\rm a\top}\bbeta_1^*+\xx_{t-1}^{\rm b\top}\bbeta_2^*
      +\xx_{t-1}^{\top}\bbeta^*+\zz_{t-1}^{\top}\ggamma^*+u_t,
\]
where the sparse coefficient vectors  $\bbeta_1^*=\bbeta_2^*=(1,0.5,\boldsymbol{0}_{p/4-2}^{\top})^{\top}$, $\bbeta^*=(5/\sqrt{n},\boldsymbol{0}_{p/4-1}^{\top})^{\top}$, and $\ggamma^*=(1,\boldsymbol{0}_{p/4-1}^{\top})^{\top}$.

\textbf{DGP A2. Pairwise cointegration.} The regressors $\xx_t=(x_{j,t})_{1\leq j \leq p/2}$ include $p/4$ cointegrated pairs. Specifically, for $j=1,\ldots,p/4$, we generate
\[
x_{2j,t}=\rho_nx_{2j,t-1}+e_{2j,t},\qquad
x_{2j-1,t}=x_{2j,t}+v_{2j-1,t},
\]
where $e_{2j,t}$ and $v_{2j-1,t}$ are i.i.d.~$\mathcal{N}(0,1)$. Furthermore, the $(p/2)$-dimensional vector $\zz_t$ contains $p/2$ stationary regressors that are i.i.d.~$\mathcal{N}(0,1)$. The predictive regression is
\[
y_t=\xx_{t-1}^{\top}\bbeta^*+\zz_{t-1}^{\top}\ggamma^*+u_t,
\]
where the sparse coefficients $\bbeta^*=(1,-1,1,-1,\boldsymbol{0}_{p/2-4}^{\top})^{\top}$, $\ggamma^*=(1,1,\boldsymbol{0}_{p/2-2}^{\top})^{\top}$, and the error term $u_t\overset{\rm i.i.d}{\sim} \mathcal{N}(0,1)$.

We retain the sample sizes, variable dimensions, regularization methods, and tuning parameter selection as in Section \ref{sec: simulations} of the main text. All results are averaged over 1000 replications. Figures \ref{fig:dgp2_ur_app} and \ref{fig:dgp3_ur_app} exhibit the results for DGPs A1 and A2 respectively under unit roots with $\rho_n = 1$, which are similar to those for the DGP in Section \ref{sec: simulations}. The prediction error decreases as the sample size $n$ grows, while increases with a larger number of regressors $p$. FCPs evidently dominate LASSO in terms of MSPE, coefficient estimation error, and variable selection accuracy. Figures \ref{fig:dgp1_lur_app}--\ref{fig:dgp3_lur_app} present the results under local unit roots, which closely mirror the unit root case. These additional simulation results evidence the robustness of FCPs in high-dimensional predictive regression with cointegration.

\begin{figure}[p]
\centering
\includegraphics[width=\textwidth]{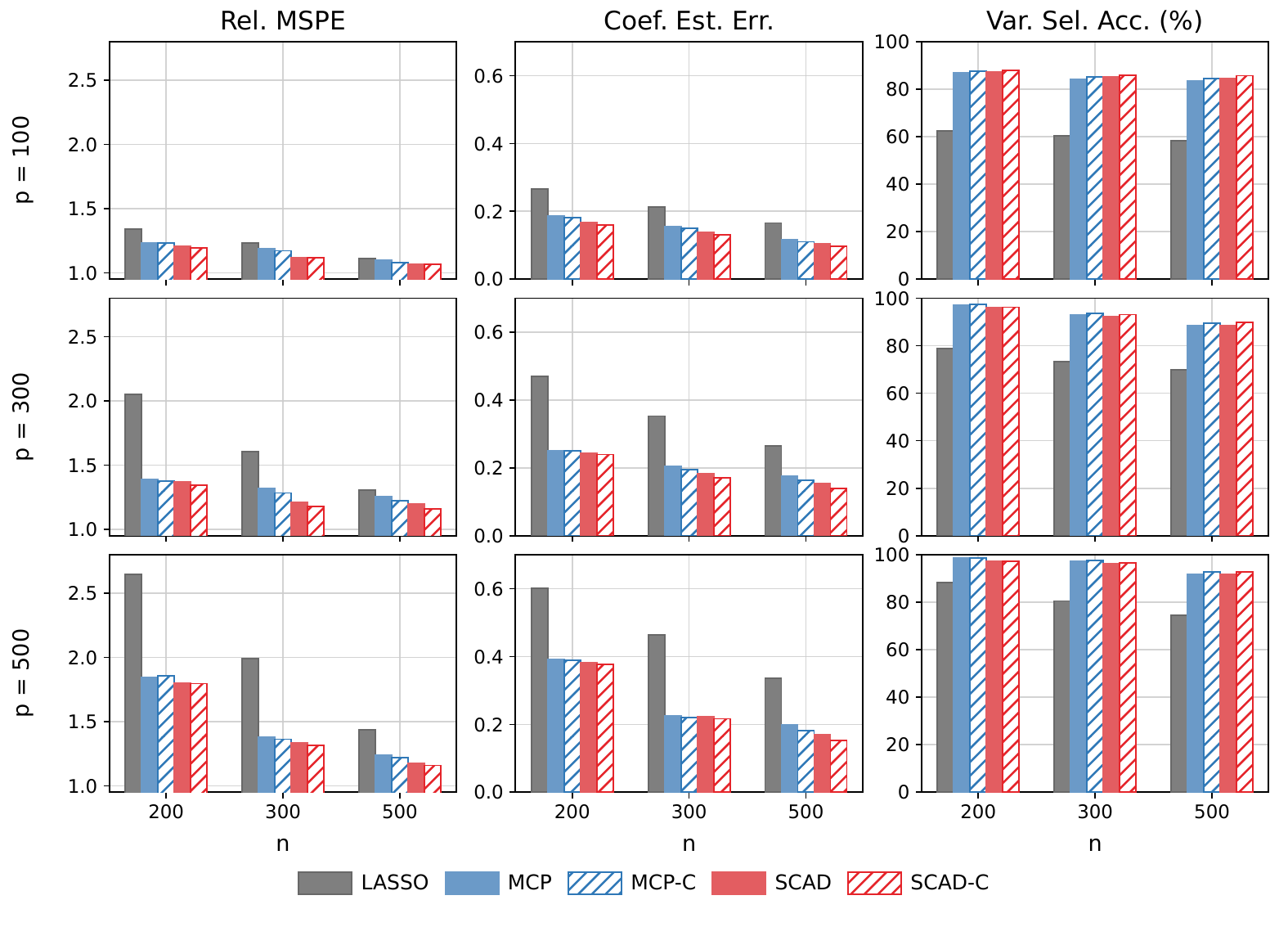}
\caption{Simulation results for DGP A1 under unit roots ($\rho_n=1$)}
{{\begin{flushleft}
\singlespacing\footnotesize \textit{Notes}: The $x$-axis represents the sample size $n$, and the rows from top to bottom correspond to $p=100$, 300, and 500. The left column \texttt{Rel.~MSPE} exhibits the OOS mean squared prediction error (MSPE) relative to the oracle OLS, the middle column \texttt{Coef.~Est.~Err.} shows the relative $\ell_2$-norm of the coefficient estimation error, and the right column \texttt{Var.~Sel.~Acc.} displays the variable selection accuracy.
\end{flushleft}}}
\label{fig:dgp2_ur_app}
\end{figure}

\begin{figure}[p]
\centering
\includegraphics[width=\textwidth]{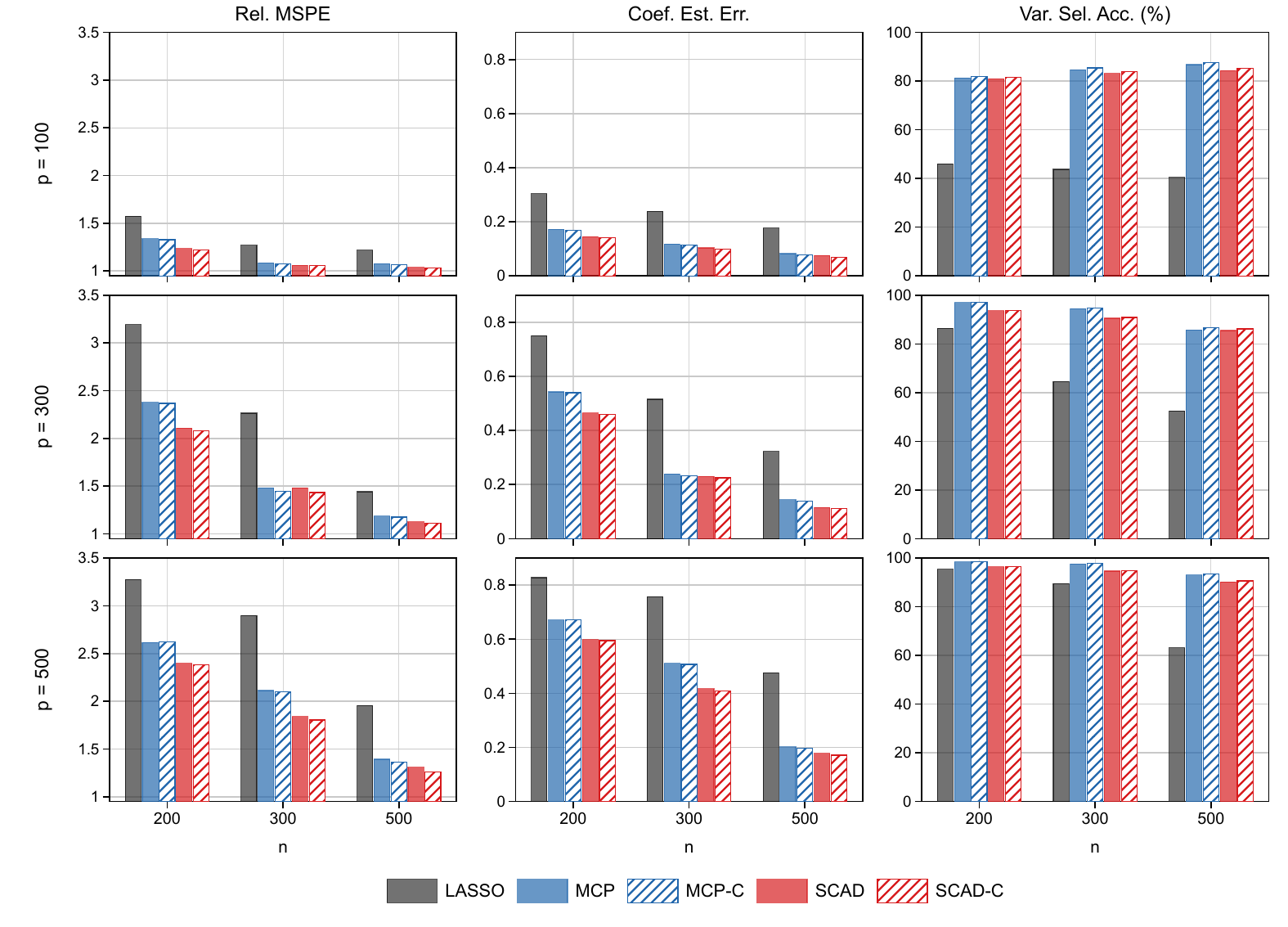}
\caption{Simulation results for DGP A2 under unit roots ($\rho_n=1$)}
{{\begin{flushleft}
\singlespacing\footnotesize \textit{Notes}: The $x$-axis represents the sample size $n$, and the rows from top to bottom correspond to $p=100$, 300, and 500. The left column \texttt{Rel.~MSPE} exhibits the OOS mean squared prediction error (MSPE) relative to the oracle OLS, the middle column \texttt{Coef.~Est.~Err.} shows the relative $\ell_2$-norm of the coefficient estimation error, and the right column \texttt{Var.~Sel.~Acc.} displays the variable selection accuracy.
\end{flushleft}}}
\label{fig:dgp3_ur_app}
\end{figure}

\begin{figure}[p]
\centering
\includegraphics[width=\textwidth]{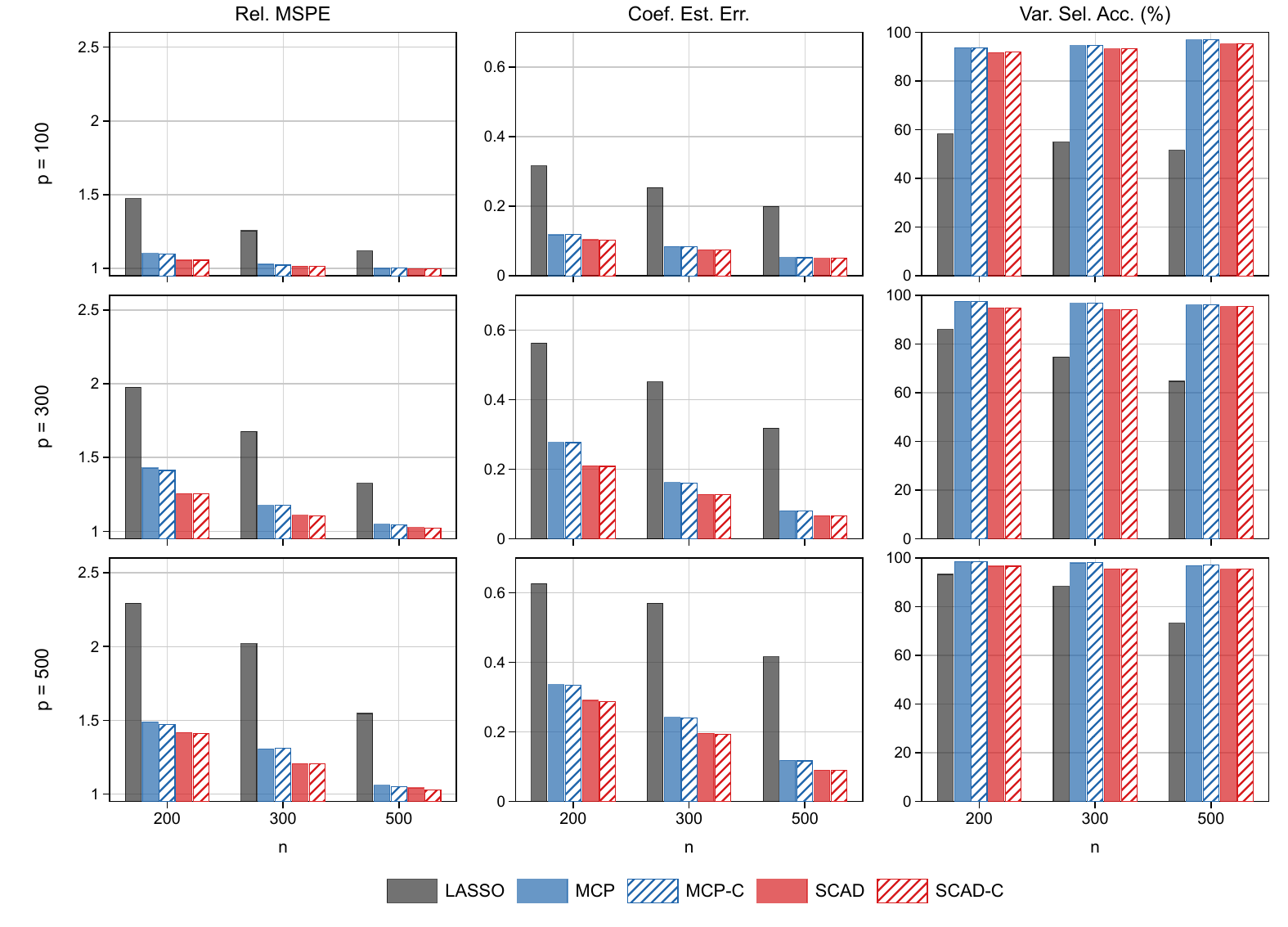}
\caption{Simulation results for Section \ref{sec: simulations} under local unit roots ($\rho_n=1-1/n$)}
{{\begin{flushleft}
\singlespacing\footnotesize \textit{Notes}: The $x$-axis represents the sample size $n$, and the rows from top to bottom correspond to $p=100$, 300, and 500. The left column \texttt{Rel.~MSPE} exhibits the OOS mean squared prediction error (MSPE) relative to the oracle OLS, the middle column \texttt{Coef.~Est.~Err.} shows the relative $\ell_2$-norm of the coefficient estimation error, and the right column \texttt{Var.~Sel.~Acc.} displays the variable selection accuracy.
\end{flushleft}}}
\label{fig:dgp1_lur_app}
\end{figure}

\begin{figure}[p]
\centering
\includegraphics[width=\textwidth]{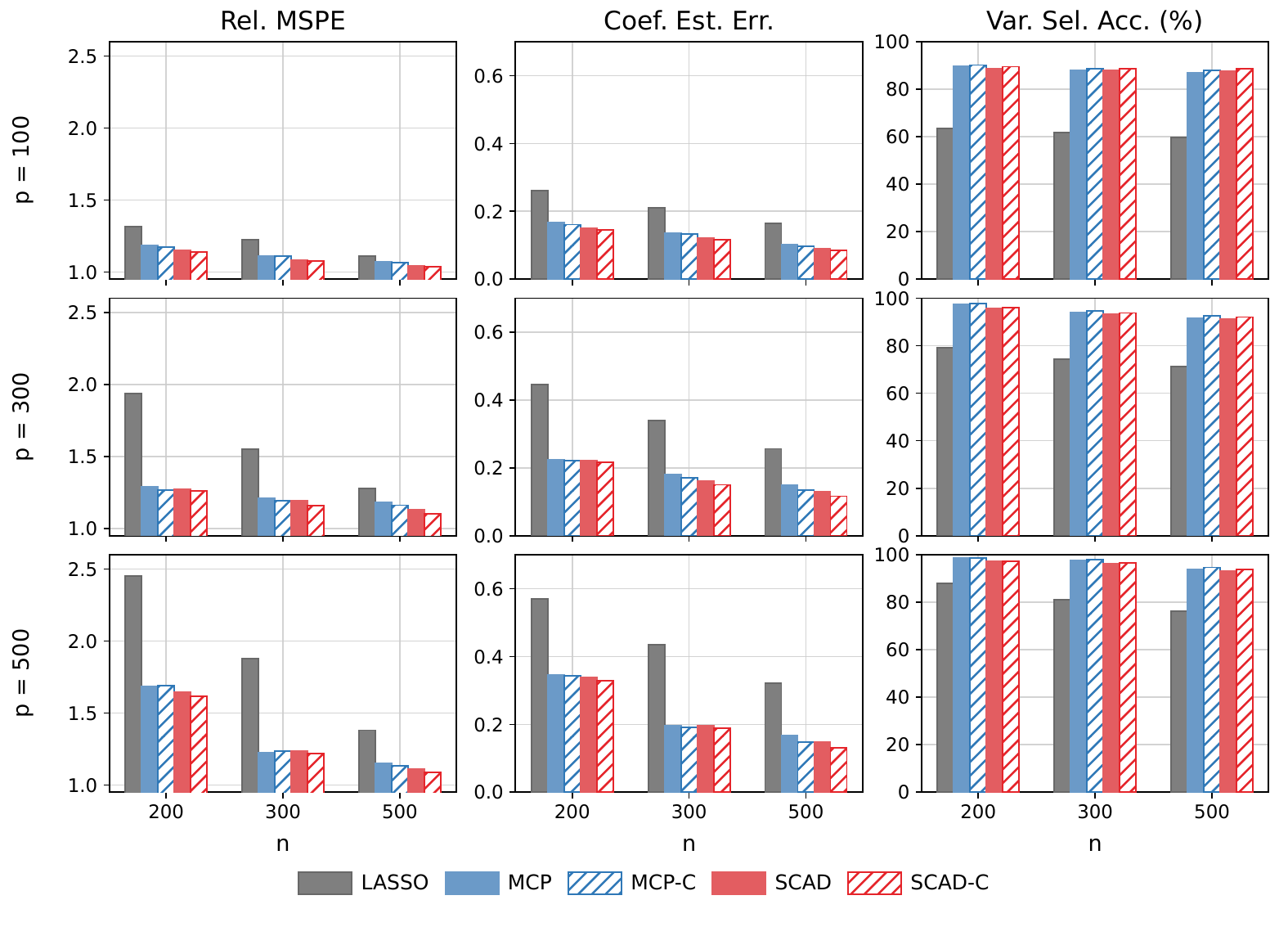}
\caption{Simulation results for DGP A1 under local unit roots ($\rho_n=1-1/n$)}
{{\begin{flushleft}
\singlespacing\footnotesize \textit{Notes}: The $x$-axis represents the sample size $n$, and the rows from top to bottom correspond to $p=100$, 300, and 500. The left column \texttt{Rel.~MSPE} exhibits the OOS mean squared prediction error (MSPE) relative to the oracle OLS, the middle column \texttt{Coef.~Est.~Err.} shows the relative $\ell_2$-norm of the coefficient estimation error, and the right column \texttt{Var.~Sel.~Acc.} displays the variable selection accuracy.
\end{flushleft}}}
\label{fig:dgp2_lur_app}
\end{figure}

\begin{figure}[p]
\centering
\includegraphics[width=\textwidth]{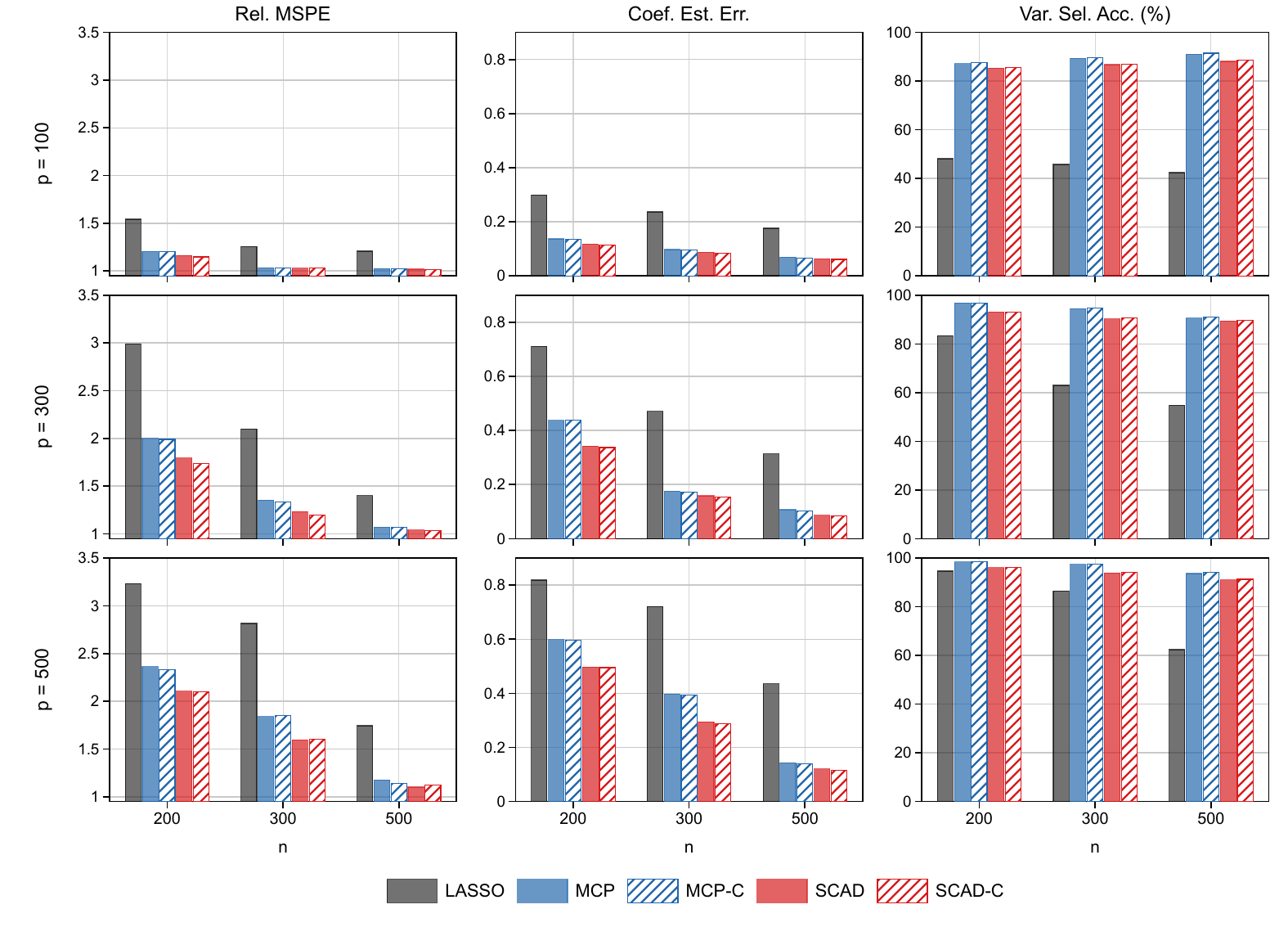}
\caption{Simulation results for DGP A2 under local unit roots ($\rho_n=1-1/n$)}
{{\begin{flushleft}
\singlespacing\footnotesize \textit{Notes}: The $x$-axis represents the sample size $n$, and the rows from top to bottom correspond to $p=100$, 300, and 500. The left column \texttt{Rel.~MSPE} exhibits the OOS mean squared prediction error (MSPE) relative to the oracle OLS, the middle column \texttt{Coef.~Est.~Err.} shows the relative $\ell_2$-norm of the coefficient estimation error, and the right column \texttt{Var.~Sel.~Acc.} displays the variable selection accuracy.
\end{flushleft}}}
\label{fig:dgp3_lur_app}
\end{figure}

\end{appendices} 

\end{document}